\documentclass[%
reprint,
nofootinbib,
superscriptaddress,
amsmath,
amssymb,
prb,
nobibnotes
]{revtex4-1}
\usepackage{braket}
\usepackage{commath}
\usepackage{etoolbox}

\usepackage[usenames,dvipsnames]{xcolor}

\usepackage[unicode=true,colorlinks=true,citecolor=blue,urlcolor=blue]{hyperref}

\makeatletter
\patchcmd{\@outputpage@head}{\@ifx{\LS@rot\@undefined}{}{\LS@rot}}{}{}{}
\makeatother

\usepackage{makecell}
\usepackage{multirow}
\usepackage{cellspace}

\usepackage{graphicx}

\usepackage{bm}
\usepackage{array}
\usepackage{enumitem}

\usepackage{graphics}

\usepackage{todonotes}
\usepackage{soul}

\usepackage{bbold}

\usepackage{balance}

\usepackage{changes}
\colorlet{Changes@Color}{magenta}
\newcommand{\stkout}[1]{\ifmmode\text{\sout{\ensuremath{#1}}}\else\sout{#1}\fi}
\setdeletedmarkup{\stkout{#1}}

\usepackage{mathrsfs}
\usepackage{empheq}
\usepackage{titlesec}
\usepackage{mdframed}
\usepackage{amsthm}
\usepackage{enumitem}

\allowdisplaybreaks

\usepackage{ifxetex,ifluatex}
\usepackage{fixltx2e} 
\ifnum 0\ifxetex 1\fi\ifluatex 1\fi=0 
  \usepackage[T1]{fontenc}
  \usepackage[utf8]{inputenc}
\else 
  \ifxetex
    \usepackage{mathspec}
  \else
    \usepackage{fontspec}
  \fi
  \defaultfontfeatures{Ligatures=TeX,Scale=MatchLowercase}
\fi
\IfFileExists{upquote.sty}{\usepackage{upquote}}{}
\IfFileExists{microtype.sty}{%
\usepackage[]{microtype}
\UseMicrotypeSet[protrusion]{basicmath} 
}{}
\usepackage{graphicx,grffile}
\makeatletter
\def\maxwidth{\ifdim\Gin@nat@width>\linewidth\linewidth\else\Gin@nat@width\fi}
\def\maxheight{\ifdim\Gin@nat@height>\textheight\textheight\else\Gin@nat@height\fi}
\makeatother
\setkeys{Gin}{width=\maxwidth,height=\maxheight,keepaspectratio}
\IfFileExists{parskip.sty}{%
\usepackage{parskip}
}{
\setlength{\parindent}{0pt}
\setlength{\parskip}{6pt plus 2pt minus 1pt}
}
\makeatletter
\def\fps@figure{htbp}
\makeatother

\newcommand{\matr}[1]{\mathord{\buildrel{\lower3pt\hbox{\scriptsize$\leftrightarrow$}}\over{ \mathbf{#1}}}}

\newmdtheoremenv[innertopmargin=0pt,innerbottommargin=7pt,skipabove=15pt,skipbelow=15pt]{theorem}{Theorem}

\newmdtheoremenv[innertopmargin=0pt,innerbottommargin=7pt,skipabove=15pt,skipbelow=15pt]{corollary}{Corollary}
\newmdtheoremenv[innertopmargin=0pt,innerbottommargin=7pt,skipabove=15pt,skipbelow=15pt]{definition}{Definition}

\mdtheorem[innertopmargin=0pt,innerbottommargin=7pt,skipabove=15pt,skipbelow=15pt,theoremseparator=:]{necessary}{Necessary conditions for genuine non-Abelian AB systems:}

\begin{document}

\title{Non-Abelian gauge field optics}

\author{Yuntian Chen}\thanks{These authors contributed equally to this work.}%
\affiliation{School of Optical and Electronic Information, Huazhong University of Science and Technology, Wuhan 430074, China}
\affiliation{Wuhan National Laboratory of Optoelectronics, Huazhong University of Science and Technology, Wuhan 430074, China}

\author{Ruo-Yang Zhang}\thanks{These authors contributed equally to this work.}
\affiliation{Department of Physics, The Hong Kong University of Science and Technology, Clear Water Bay, Hong Kong}

\author{Zhongfei Xiong}
\affiliation{School of Optical and Electronic Information, Huazhong University of Science and Technology, Wuhan 430074, China}

\author{Zhi~Hong Hang}
\affiliation{School of Physical Science and Technology and Institute for Advanced Study, Soochow University, Suzhou 215006, China}

\author{Jensen Li}
\affiliation{Department of Physics, The Hong Kong University of Science and Technology, Clear Water Bay, Hong Kong}

\author{Jian Qi Shen}
\email{jqshen@zju.edu.cn}
\affiliation{Centre for Optical and Electromagnetic Research, State Key Laboratory of Modern Optical Instrumentations, Zhejiang University, Hangzhou 310058, China;}

\author{C. T. Chan}
\email{phchan@ust.hk}
\affiliation{Department of Physics, The Hong Kong University of Science and Technology, Clear Water Bay, Hong Kong}


\begin{abstract}
The concept of gauge field is a cornerstone of modern physics and the synthetic gauge field has emerged as a new way to manipulate particles in many disciplines. In optics, several schemes of Abelian synthetic gauge fields have been proposed. Here, we introduce a new platform for realizing synthetic $\mathrm{SU}(2)$ non-Abelian gauge fields acting on two-dimensional optical waves in a wide class of anisotropic materials and discover novel phenomena. We show that a virtual non-Abelian Lorentz force arising from material anisotropy can induce light beams to travel along \textit{Zitterbewegung} trajectories even in homogeneous media. We further design an optical non-Abelian Aharonov–-Bohm system which results in the exotic spin density interference effect. We can extract the Wilson loop of an arbitrary closed optical path from a series of gauge fixed points in the interference fringes. Our scheme offers a new route to study $\mathrm{SU}(2)$ gauge field related physics using optics.
\end{abstract}

\maketitle

Gauge fields originated from classical electromagnetism, and have become the kernel of fundamental physics after being extended to non-Abelian by Yang and Mills~\cite{yang1954conservation}. Apart from real gauge bosons, emergent gauge fields in either real~\cite{mead1979determination} or parameter spaces~\cite{berry1984quantal,wilczek1984appearance} have recently been widely used to elucidate the complicated dynamics in a variety of physical systems~\cite{aidelsburger2018artificial}, including electronic~\cite{xiao2010berry,fujita2011gauge}, ultracold atom~\cite{lin2009synthetic,dalibard2011colloquium,Goldman2014}, and photonic~\cite{onoda2004hall,BliokhPRE2004,onoda2006geometrical,bliokh2007non,bliokh2008geometrodynamics,ma2016spin,sawada2005optical,wang2008reflection,fang2013effective,Hafezi2011NaturePhyscs,fang2012photonic,FangNP2012,fang2013controlling,rechtsman2013strain,jia2019observation,ChenWj2014,JensenLi2015PRL,liu2015polarization,liu2017invisibility,liu2018realizing,schine2016synthetic} systems. The geometric nature~\cite{Wuyang} of gauge theory makes it a powerful tool for studying the topological phases of matter~\cite{Qi2011,senthil2015symmetry,lu2014topological,ozawa2018topological}.

The concept of emergent gauge fields has offered us new insights in optics and photonics, such as the manifestation of the gauge structure (Berry connection and curvature) in momentum space~\cite{onoda2004hall,onoda2006geometrical,BliokhPRE2004,bliokh2007non,bliokh2008geometrodynamics,ma2016spin}. Artificial gauge fields realized by breaking time reversal symmetry with magnetic effects~\cite{sawada2005optical,wang2008reflection,fang2013effective} or dynamic modulation~\cite{fang2012photonic,FangNP2012,fang2013controlling} have given rise to new paradigms for controlling light trajectories in real space. Even for time-reversal-invariant systems,  a pair of virtual magnetic fields -- each being the time-reversed partner of the other -- can be generated using methods such as coupled optical resonators~\cite{Hafezi2011NaturePhyscs}, engineering  lattices with strain~\cite{rechtsman2013strain,jia2019observation}, or reciprocal metamaterials~\cite{ChenWj2014,JensenLi2015PRL,liu2015polarization,liu2017invisibility,liu2018realizing}.
However, except for a few works revealing the non-Abelian gauge structure in momentum space~\cite{onoda2006geometrical,bliokh2007non,ma2016spin}, all of these schemes of synthetic gauge fields in real space are restricted to the Abelian type.  

Recently,  anisotropic metamaterials were used to manipulate light through artificial Abelian gauge fields~\cite{JensenLi2015PRL,liu2015polarization,liu2017invisibility,liu2018realizing}. It was demonstrated that the off-diagonal components of permittivity and permeability appear as a pair of ``spin-dependent'' vector potentials  in the 2-dimensional (2D) wave equation for certain anisotropic media. Though the material parameters are subjected to strong restriction in this scheme, the internal pseudo-spin degree of freedom implies the possible generalization to a synthetic non-Abelian gauge field theory for light by coupling the spin-up and spin-down states. 

In this work, we discover that the transport of optical waves in a wide class of anisotropic media can be associated with an emergent 2D non-Abelian $\mathrm{SU}(2)$ gauge interaction in real space, enabling us to obtain the first scheme for realizing synthetic non-Abelian gauge field for classical waves. Contrary to intuition, we show that a more exotic general $\mathrm{SU}(2)$ gauge framework can manifest in 2D optical dynamics, provided the restriction on the material parameters employed in refs.~\cite{JensenLi2015PRL,liu2015polarization,liu2017invisibility,liu2018realizing} is relaxed. Our platform presents broader applicability and allows the study of novel optical phenomena not found in Abelian synthetic gauge field systems.
We illustrate our idea with two examples. The first example is the \emph{Zitterbewegung} (ZB) of light in homogeneous non-Abelian media, which refers to the trembling motion of wave packets~\cite{zawadzki2011zitterbewegung}. ZB has been realized in systems possessing Dirac dispersion~\cite{vaishnav2008observing,gerritsma2010quantum,zhang2008observing,dreisow2010classical,fan2015plasmonic,chenjing2016OE}, but we will see that ZB of light can arise from a distinctly different mechanism: emergent non-Abelian Lorentz force. In the second example, we propose for the first time a concrete design of a genuine non-Abelian Aharonov-Bohm (AB) system~\cite{Raman1986} using two synthetic non-Abelian vortices, and reveal that the noncommutativity of winding around the two vortices gives rise to nontrivial interference results. In particular, we show that there exists a series of fixed points in the interference fringes invariant under gauge transformation, from which we can obtain the Wilson loops of the closed path concatenated by the two interfering optical paths. As evidenced by the examples, our scheme offers a fresh angle to understand the dynamic effects of light in anisotropic media, and also suggests an optical approach to probe new physics accompanied by $\mathrm{SU}(2)$ gauge fields.

\textbf{Results}\\[3pt]
\textbf{Non-Abelian gauge fields acting on light. }
Our scheme focuses on 2D propagating optical waves in nondissipative anisotropic media characterized by the permittivity and permeability tensors:
\begin{equation}\arraycolsep=4pt\def\arraystretch{1.5}\label{material parameters}
\matr{\varepsilon}/\varepsilon_0=\left(\begin{array}{c|c}  \matr{\varepsilon}_T  & \mathbf{g}_1 \\\hline
\mathbf{g}_1^\dagger & \varepsilon_z
\end{array}\right),\quad
\matr{{\mu}}/\mu_0=\left(\begin{array}{c|c}
\matr{{\mu}}_T & \mathbf{g}_2 \\\hline
\mathbf{g}_2^\dagger & \mu_z
\end{array}\right).
\end{equation}
Here, all of the parameters depend on $x,y$; the diagonal blocks $\matr{\varepsilon}_T$, $\matr{\mu}_T$, $\varepsilon_z$, $\mu_z$ are real numbers, while the off-block-diagonal components $\mathbf{g}_i=(g_{i\,x},g_{i\,y})^\intercal= g_{i\,x}{\mathbf{e}_x}+g_{i\,y}{\mathbf{e}_y}$ ($i=1,2$)
are in-plane complex vectors whose imaginary parts could be induced by the gyrotropic effect with in-plane gyration vectors. The only constraint on the media is the ``in-plane duality'', $\matr{\varepsilon}_T=\alpha\,\matr{\mu}_T$, where $\alpha$ is a positive constant. For simplicity, we set $\alpha=1$ in the following, and $\alpha\neq 1$ results can be  obtained directly by redefining $\varepsilon_0\rightarrow \alpha\,\varepsilon_0$. Under this constraint, the in-plane monochromatic wave equation of frequency $\omega$ can be written as
\begin{equation}\label{wave equation}
  \hat{H}|\psi\rangle=\left[\frac{1}{2}\big(\hat{\mathbf{p}}-\hat{\mathcal{A}}\big)\cdot\matr{m}^{-1}\cdot\big(\hat{\mathbf{p}}-\hat{\mathcal{A}}\big)-\hat{\mathcal{A}}_0+V_0\right]
  |\psi\rangle=0.
\end{equation}
Here $|\psi\rangle=(E_z,\,\eta_0 H_z)^\intercal$ ($\eta_0=\sqrt{\mu_0/\varepsilon_0}$ ) serves as a two-component wave function, and $\hat{H}$ resembles the Hamiltonian of a non-relativistic spin-1/2 particle traveling in $\mathrm{SU}(2)$ non-Abelian gauge potentials~\cite{wong1970field}, where $\hat{\mathbf{p}}=-\mathrm{i}\hat{\sigma}_0\partial_i\mathbf{e}^{i}$ ($i=1,2$) is the canonical momentum operator with $\hat{\sigma}_0$ being the 2D identity matrix,  $\matr{m}=\matr{\varepsilon}_T^{-1}\det{(\matr{\varepsilon}_T)}/2$
represents an effective anisotropic mass, in particular,  $\hat{\mathcal{A}}=\mathcal{A}^1\hat{\sigma}_1+\mathcal{A}^2\hat{\sigma}_2$ and $\hat{\mathcal{A}}_0=\mathcal{A}_0^a\hat{\sigma}_a$ ($\hat{\sigma}_a$ ($a=1,2,3$) are Pauli matrices) can be interpreted as emergent non-Abelian vector and scalar potentials respectively, and $V_0$ is an additional Abelian scalar potential. As shown in Table~\ref{gauge potentials}, the emergent gauge potentials are determined by the material parameters, especially, the vector potential directly corresponds to the off-diagonal terms of $\matr{\varepsilon}$ and $\matr{\mu}$. This correspondence can be intuitively understood from the $\mathrm{SU}(2)$ gauge covariance of 2D Maxwell's equations (see Methods), and the detailed derivation of Eq.~(\ref{wave equation}) is given in the Supplementary Note 1. Thereby, in this broad class of anisotropic media, the materials' influence on the 2D optical waves imitates a $\mathrm{SU}(2)$ gauge interaction. Furthermore, if the background media are extended to be bi-anisotropic materials, a complete construction of $\mathrm{U}(2)=\mathrm{SU}(2)\rtimes \mathrm{U}(1)$ gauge fields for light can be achieved (see Supplementary Note 1).

The emergent $\mathrm{SU}(2)$ gauge potential $\{\hat{\mathcal{A}}_\mu\}=\{\hat{\mathcal{A}}_0,\hat{\mathcal{A}}\}$ induces a synthetic $\mathrm{SU}(2)$ gauge field acting on light:
\begin{equation}
\hat{\mathcal{F}}_{\mu\nu}=\mathrm{i}[\hat{\mathcal{D}}_\mu,\hat{\mathcal{D}}_\nu]=\partial_\mu\hat{\mathcal{A}}_\nu-\partial_\nu\hat{\mathcal{A}}_\mu-\mathrm{i}[\hat{\mathcal{A}}_\mu,\hat{\mathcal{A}}_\nu],
\end{equation}
where $\hat{\mathcal{D}}_\mu=\hat{\sigma}_0\partial_\mu-\mathrm{i}\,\hat{\mathcal{A}}_\mu$ ($\mu=0,1,2$) is the covariant derivative.
Analogous to real electromagnetic (EM) fields, the synthetic $\mathrm{SU}(2)$ gauge field can be separated into a non-Abelian magnetic field $\hat{\mathcal{B}}=\frac{1}{2}\epsilon^{ij}\hat{\mathcal{F}}_{ij}{\mathbf{e}}_z$  along the $z$ axis and a non-Abelian in-plane electric field $\hat{\mathcal{E}}=-\hat{\mathcal{F}}_{0i}{\mathbf{e}}_i$, which are associated with the gauge potential as
\begin{equation}\label{nonabelian fields}
\hat{\mathcal{B}}=\nabla\times\hat{\mathcal{A}}-\mathrm{i}\hat{\mathcal{A}}\times\hat{\mathcal{A}},\quad\  \hat{\mathcal{E}}=\nabla\hat{\mathcal{A}}_0+\mathrm{i}[\hat{\mathcal{A}}_0,\hat{\mathcal{A}}].
\end{equation}
The second terms of $\hat{\mathcal{B}}$, $\hat{\mathcal{E}}$ cannot be found in the Abelian case since they are induced entirely by the noncommutativity of the non-Abelian gauge potential.
Indeed, a matrix-valued gauge potential would not be regarded as (apparently) non-Abelian, unless some of its components do not commute with each other $[\hat{\mathcal{A}}_\mu,\hat{\mathcal{A}}_\nu]\neq0$~\cite{Goldman2014}.  For instance,  the scheme in ref.~\cite{JensenLi2015PRL} is actually a specific reduction of ours with the strict constraints on the media that (i) $\mathbf{g}_1=-\mathbf{g}_2$ being real and (ii) $\varepsilon_z=\mu_z$. In this case, the vector potential only has $\hat{\sigma}_1$ component  $\hat{\mathcal{A}}=\mathcal{A}^1\hat{\sigma}_1$ and the scalar potential $\hat{\mathcal{A}}_0$ vanishes. As such, $[\hat{\mathcal{A}}_i,\hat{\mathcal{A}}_j]\equiv0$, and the gauge group is reduced to the Abelian subgroup $\mathrm{U}(1)$ of $\mathrm{SU}(2)$. In general, if Eq.~(\ref{wave equation}) has any  $\mathrm{U}(1)$ spin rotation symmetry, which means $\hat{U}\hat{H}\hat{U}^\dagger=\hat{H}$ for $\hat{U}=\exp\left(\mathrm{i}\phi\,\vec{n}\cdot\vec{\hat\sigma}\right)$ with a parameter $\phi$, the gauge potential would be reducible. Hence, only for those materials that can imitate irreducible $\mathrm{SU}(2)$ gauge potentials, we call them non-Abelian media.

\begin{table}[t!]
    \caption{\label{gauge potentials}The expressions of the synthetic $\mathrm{SU}(2)$ and $\mathrm{U}(1)$ gauge potentials in terms of the anisotropic parameters of the media.}\vspace{5pt}
 
\def\arraystretch{1.5} 
  \begin{tabular}{Sc|Sc|Sc}\hline\hline
  \multirow{8}*{$\mathrm{SU}(2)$} & \multirow{3}*{\makecell[c]{vector\\[-0.75 ex] potential\\[3pt]$\hat{\mathcal{A}}=\mathcal{A}^a\hat{\sigma}_a$}} &
  ${\mathcal{A}^1}=k_0\mathrm{Re}\left(\mathbf{g}_-\right)\times\mathbf{e}_z$  \footnote{Here $k_0=\omega/c$ is  vacuum wave number and $\mathbf{g}_\pm=(\mathbf{g}_1\pm\mathbf{g}_2^*)/2$.}\\
   & & ${\mathcal{A}^2}=k_0\mathrm{Im}\left(\mathbf{g}_-\right)\times\mathbf{e}_z$  \\
   & & ${\mathcal{A}^3}=0$   \\ \cline{2-3}
  &\multirow{4}*{\makecell[c]{scalar\\[-0.75 ex] potential\\[3pt]$\hat{\mathcal{A}}_0=\mathcal{A}_0^a\hat{\sigma}_a$}}   &  $\displaystyle
  {\mathcal{A}_0^1}= k_0 \mathbf{e}_{z}\cdot\left[\nabla\times\left(\tensor{\varepsilon}_T^{-1}\cdot\mathrm{Im}\left(\mathbf{g}_+\right)\right)\right]$   \rule{0pt}{4ex}\\[3pt]
  & & $\displaystyle{\mathcal{A}_0^2}=- k_0\ \mathbf{e}_{z}\cdot\left[\nabla\times\left(\tensor{\varepsilon}_T^{-1}\cdot\mathrm{Re}\left(\mathbf{g}_+\right)\right)\right] $ 
  \\[3pt]
  & &  $\displaystyle{\mathcal{A}_0^3}= k_0^2\left[\frac{\varepsilon_z-\mu_z}{2}-2\,\mathrm{Re}\left(\mathbf{g}_-^\dagger\cdot\tensor{\varepsilon}_T^{-1}\cdot\mathbf{g}_+\right)\right]$ \rule[-2.6ex]{0pt}{0pt}\\\hline
  $\mathrm{U}(1)$ &   \makecell[c]{scalar\\[-0.75 ex] potential} &$\displaystyle {V_0}= k_0^2\left[\left(\mathbf{g}_+^\dagger\cdot\tensor{\varepsilon}_T^{-1}\cdot\mathbf{g}_+\right)-\frac{\varepsilon_z+\mu_z}{2}\right]$ \rule[-2.6ex]{0pt}{6.8ex} 
   \\\hline\hline
\end{tabular}
\end{table}

The two-component wave function of light $|\psi\rangle$ behaves like a spin-1/2 spinor with  the pseudo-spin at a local point
\begin{equation}
\vec{s}=\langle\psi|\vec{\hat{\sigma}}|\psi\rangle/|\psi|^2,
\end{equation}
where the overhead arrow indicates a vector in the pseudo-spin space, and $\langle\psi|\vec{\hat{\sigma}}|\psi\rangle$ gives the local spin density. The frame $\{\vec{e}_a\}$ in the pseudo-spin space can be chosen arbitrarily. The rotation of the frame corresponds to a gauge transformation of spinor $|\psi'\rangle=\hat{U}(\mathbf{r})|\psi\rangle$, where in general $\hat{U}(\mathbf{r})$ is a space-varying $\mathrm{SU}(2)$ matrix. By substituting $|\psi'\rangle$ into Eq.~(\ref{wave equation}), one can easily check that the wave equation is gauge covariant as long as the material is transformed accordingly (see Supplementary Note 2), while the synthetic gauge potentials and fields obey the gauge transformations
\begin{gather}
  \hat{\mathcal{A}}^{\prime}_\mu=\hat{U}\hat{\mathcal{A}}_\mu\hat{U}^\dagger+\mathrm{i}\,\hat{U}\partial_\mu\hat{U}^\dagger,\\
  \hat{\mathcal{B}}'=\hat{U}\hat{\mathcal{B}}\hat{U}^\dagger,\quad \hat{\mathcal{E}}'=\hat{U}\hat{\mathcal{E}}\hat{U}^\dagger.
\end{gather}

In addition, it is worth comparing the present idea of non-Abelian gauge field optics (NAGFO) with the transformation optics (TO)~\cite{pendry2006controlling,schurig2006calculation,chen2010transformation,leonhardt2006general,leonhardt2009transformation}. 
When TO is applied to design invisibility cloaks, it results in anisotropic media whose permittivity and permeability are real and equal $\matr{\varepsilon}=\matr{\mu}$~\cite{pendry2006controlling,schurig2006calculation}.
Due to the equivalence of the constitutive tensor and  the metric of a curved spacetime for light,  such kind of duality symmetric materials can also be used to mimic gravitational effects~\cite{leonhardt2006general,leonhardt2009transformation,genov2009mimicking,chen2010Schwarzschild,sheng2013trapping}.  
In contrast to TO, NAGFO involves a more general class of complex-valued media respecting in-plane duality symmetry. The in-plane block $\matr{\varepsilon}_T$ of permittivity, which determines the effective mass  in Eq.~(\ref{wave equation}), can alternatively be equated to the metric of a virtual 2D curved space as with TO, whereas, apart from $\matr{\varepsilon}_T$, all the remaining components of $\matr{\varepsilon}$ and $\matr{\mu}$ contribute to the synthetic $\mathrm{SU}(2)$ gauge potentials. Therefore, NAGFO proposes an optical way to simulate the 2D spinor systems under both a $\mathrm{SU}(2)$ gauge interaction and the influence of a curved space. To highlight the effects stemming purely from the non-Abelian gauge interaction, we will hereinafter concentrate on the simplified scenario that $\matr{\varepsilon}_T=\varepsilon_T\matr{I}_{2\times2}$ is isotropic and homogeneous. As such, the virtual 2D background space is trivialized to be flat, and the effective mass is reduced to $m=\varepsilon_T/2$.

\begin{figure*}
\includegraphics[width=0.83\textwidth]{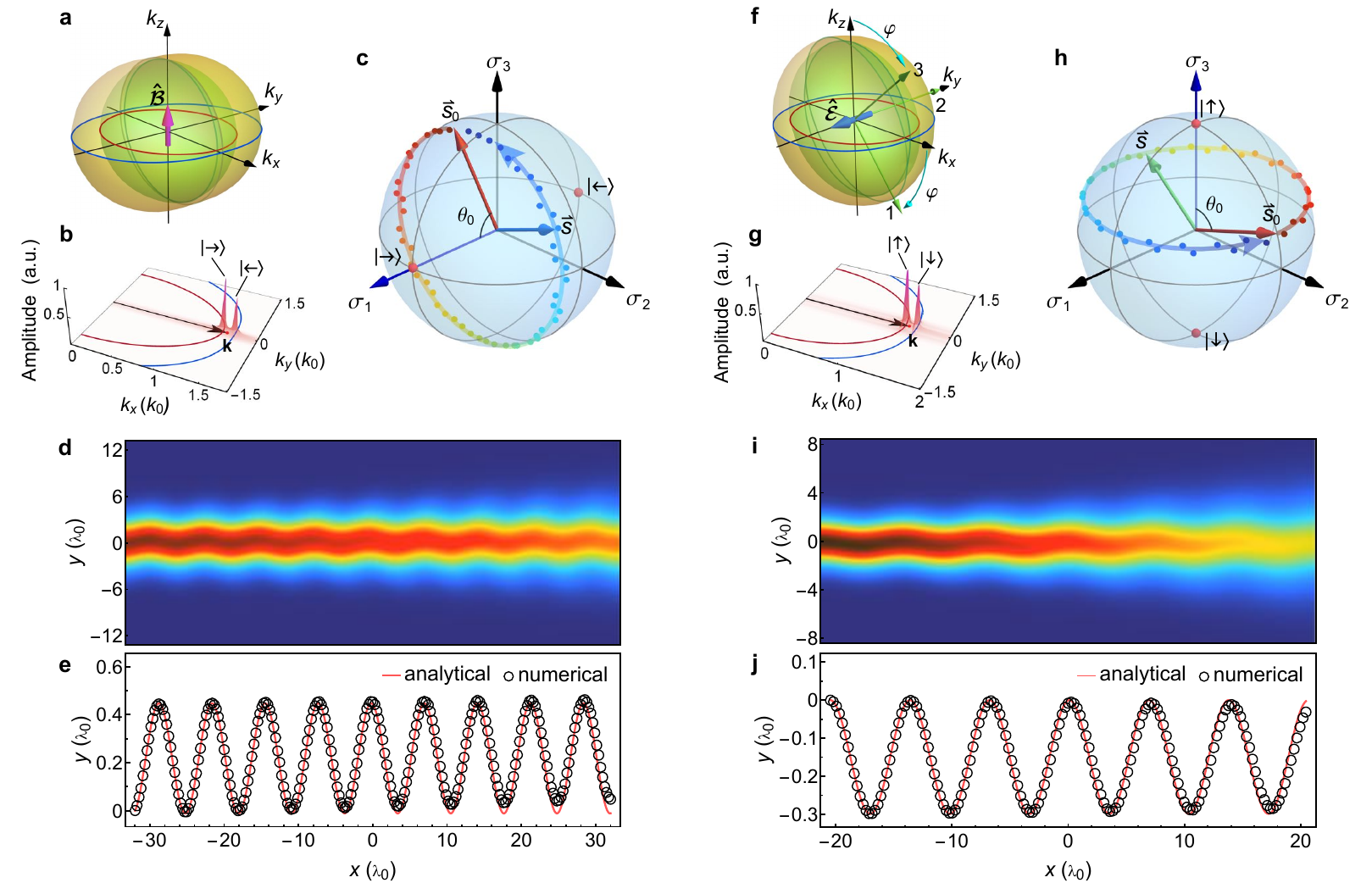}
\caption{\label{zitterbewegung} Zitterbewegung effect in homogeneous non-Abelian media. \textbf{a}--\textbf{e}\ ZB induced by a synthetic non-Abelian magnetic field in a gyrotropic medium with the parameters $\matr{\varepsilon}_T=\matr{\mu}_T=1.5\matr{I}_{2\times2}$, $\varepsilon_z=\mu_z=1.5$, $\mathbf{g}_1=-\mathbf{g}_2^*=(0.3\,\mathrm{i},-0.07)^\intercal$. This medium produces a synthetic $\mathrm{SU}(2)$ magnetic field along $z$ direction, $\hat{\mathcal{B}}=-k_0^2\,0.042\,{\mathbf{e}_z}\hat{\sigma}_3$, with a null $\mathrm{SU}(2)$ electric field $\hat{\mathcal{E}}=0$. 
 \textbf{f}--\textbf{j}\  ZB  induced by a synthetic non-Abelian electric field in a biaxial non-magnetic medium with the parameters $\varepsilon_1=1.65$, $\varepsilon_2=2.45$, $\varepsilon_3=3$,  and $\mu/\mu_0=1$.  The synthetic $\mathrm{SU}(2)$ electric field, $\hat{\mathcal{E}}=-k_0^3\,0.08919\, {\mathbf{e}_y} \hat{\sigma}_2$, is along the $y$-aixs, while the $\mathrm{SU}(2)$ magnetic field vanishes $\hat{\mathcal{B}}=0$.
 \textbf{a},\textbf{f}\ The isofrequency surfaces and their $xy$ cross sections (red and blue curves) of both cases. The green arrows in \textbf{f} are the three principal axes 1,\,2,\,3 of permittivity tensor. \textbf{b},\textbf{g}\ Fourier spectra in $k$-space of the beams in the two media. In each case, the two peaks in the spectrum correspond to the two eigenmodes with wave vectors in the $x$ direction.  And the average wave vectors $\mathbf{k}$ are marked by the black arrows. \textbf{c},\textbf{h}\ The spin precession along each beam on the Bloch sphere. The colored dots are the numerical data within one ZB period. \textbf{d},\textbf{i}\ Full-wave simulated intensity distributions, where the beam waists equal $4.4\lambda_0$ and $6.2\lambda_0$ respectively ($\lambda_0=2\pi/k_0$ is the wavelength in vacuum).
 \textbf{e},\textbf{j} Numerical (black circles) and analytical (red curves) trajectories of the intensity centroid.}\vspace{-7pt}
\end{figure*}

\textbf{Zitterbewegung of optical beams.}
The wave packet dynamics in homogeneous media can give the most straightforward effect distinguishing the non-Abelian media from the Abelian type. The effective Abelian electric and magnetic fields vanish in homogeneous media~\cite{JensenLi2015PRL}, whereas the non-Abelian fields persist due to the noncommutativity of $\hat{\mathcal{A}}_\mu$. In our case, $\hat{\mathcal{B}}=\mathcal{B}\hat{\sigma}_3$ with $\mathcal{B}=\mathrm{i} k_0^{\,2}(\mathbf{g}_-\times\mathbf{g}^*_-)$, and $\hat{\mathcal{E}}=2\mathcal{A}^3_0\left(\mathcal{A}^2\hat{\sigma}_1-\mathcal{A}^1\hat{\sigma}_2\right)$. 
We consider the propagation of 2D optical beams in homogeneous non-Abelian media. In general, there are two non-degenerate branches of plane wave eigenstates. Because the two eigenstates of a certain direction of  wave vector $\mathbf{k}$ are always orthogonal, their pseudo-spins correspond to a pair of antipodal points on the Bloch sphere. Generally speaking, the non-degenerate eigenmodes would evolves independently along different semiclassical trajectories. However, if the two eigenstates for a particular direction of $\mathbf{k}$ are quasi-degenerate, in the overlapped region, their superposed wave can be viewed as an intact ``semiclassical particle'' with an internal spin degree of freedom, whose centroid trajectory follows the Hamilton's canonical equations (see Methods)
\begin{align}\label{momentum equation}
\frac{d}{d\tau}\langle\hat{\mathbf{p}}\rangle &= \mathrm{i}\,\langle[\hat{H},\hat{\mathbf{p}}]\rangle\equiv0\quad \Rightarrow  \quad \langle\hat{\mathbf{p}}\rangle\equiv \mathbf{k},\\\label{velocity}
\frac{d}{d\tau}\langle\hat{\mathbf{r}}\rangle &=  \langle\hat{\mathbf{v}}\rangle=\frac{1}{m}\big(\mathbf{k}-\mathcal{A}^a\langle\hat{\sigma}_a\rangle\big).
\end{align}
Here $\hat{\mathbf{v}}=\frac{d}{d\tau}\hat{\mathbf{r}}=\mathrm{i}[\hat{H},\hat{\mathbf{r}}]=(\hat{\mathbf{p}}-\hat{\mathcal{A}})/m$ is the velocity operator, $\tau$ represents path parameter along the beam, and $\langle\hat{a}\rangle(\mathbf{r}_0)=\int d\mathbf{r}_\perp\,\psi^\dagger(\mathbf{r}_0+\mathbf{r}_\perp)\hat{a}(\mathbf{r}_0+\mathbf{r}_\perp)\,\psi(\mathbf{r}_0+\mathbf{r}_\perp)$ means the expectation value of an operator $\hat{a}$ averaged over the transverse cross section of a point $\mathbf{r}_0$ along an optical beam, differing from the local expectation value $\langle\psi| \hat{a}|\psi\rangle(\mathbf{r})=\psi^\dagger(\mathbf{r})\hat{a}(\mathbf{r})\,\psi(\mathbf{r})$. According to Eq.~(\ref{momentum equation}), the canonical momentum along the beam is conserved, and is equal to the quasi-degenerate eigen wave vector $\mathbf{k}$ (see  Eq.~(\ref{quasi-degeneratre k})). Moreover, it turns out that the in-plane projection of the total energy flux over the transverse cross section of the beam is always parallel to the velocity given by Eq.~(\ref{velocity}) (see proof in Methods), therefore the canonical equations do describe the path of energy propagation. 
Along the optical beam, the pseudo-spin $\vec{s}=\langle\vec{\hat{\sigma}}\rangle$ undergoes precession as follows:
\begin{equation}\label{spin precession}
  \frac{d}{d\tau}\vec{s}=\mathrm{i}\langle[\hat{H},\vec{\hat{\sigma}}]\rangle=
\vec{\Omega}\times\vec{s},
\end{equation}
where
$  \vec{\Omega}=-
2\Big(\mathcal{A}^a_0+\frac{1}{m}\mathbf{k}\cdot\mathcal{A}^a\Big)\vec{e}_a
$ 
is the precession angular velocity. 
During precession, the pseudo-spin component parallel to $\vec{\Omega}$ is conserved. 

In terms of Eqs.~(\ref{momentum equation})--(\ref{spin precession}), we arrive at the Newton-type equation of motion where a virtual non-Abelian Lorentz force~\cite{wong1970field,Goldman2014} associated with the  non-Abelian fields emerges 
\begin{equation}\label{EOM}
  \begin{split}
m\frac{d^2}{d\tau^2}\langle\hat{\mathbf{r}}\rangle=&\ \frac{1}{2}\langle\hat{\mathbf{v}}\times\hat{\mathcal{B}}+\hat{\mathcal{B}}\times\hat{\mathbf{v}}\rangle+\langle\hat{\mathcal{E}}\rangle\\
=&\ \langle\hat{\mathbf{j}}_{\hat{\sigma}_3}\rangle\times\mathcal{B}  +   \mathcal{E}^a\langle\hat{\sigma}_a\rangle,
\end{split}
\end{equation}
Here, $\hat{\mathbf{j}}_{\hat{\sigma}_3}=\frac{1}{2}\left(\hat{\mathbf{v}}\,\hat{\sigma}_3+\hat{\sigma}_3\hat{\mathbf{v}}\right)=\frac{1}{m}\hat{\mathbf{p}}\,\hat{\sigma}_3$ represents the $\hat{\sigma}_3$-component of the linear spin current operator~\cite{shen2005prl},
thus the non-Abelian Lorentz force  can also be regarded as a spin-induced force with a magnetic part acting on the spin current and an electric part acting on the average spin over the transverse cross sections of the beam.
In particular, the magnetic part of the force, $ \mathbf{f}_{\hat{\sigma}_3}=\big\langle\hat{\mathbf{j}}_{\hat{\sigma}_3}\big\rangle\times\mathcal{B}$, duplicates the ``spin transverse force'' in electronics which acts on an electronic spin current exerted by a vertical electric field~\cite{shen2005prl}.

The integration of either the canonical equations or Eq.~(\ref{EOM})  yields the intensity centroid trajectory of the beam
\begin{equation}\label{trajectory}
  \begin{split}
\langle\hat{\mathbf{r}}\rangle=&\frac{1}{m}\left[\mathbf{k}-\mathcal{A}^a{s_0}_a+\frac{1}{\Omega^2}\mathbf{F}^a\epsilon_{abc}{\Omega}^b s_0^{\,c}\right]\tau\\
 -&\frac{\mathbf{F}^a}{m\Omega^2}\left[\big(\cos(\tau\,\Omega)-1\big)\delta_{ac}
+\frac{\sin(\tau\,\Omega)}{\Omega}\epsilon_{abc}{\Omega}^b\right]s_0^{\,c},
  \end{split}
\end{equation}
where $\mathbf{F}^a=(\vec{\Omega}\times\hat{\mathcal{A}})^a=\mathcal{E}^a+{\mathbf{k}}\times\mathcal{B}^a/m$, $\Omega=|\vec{\Omega}|$, $\vec{s}_0$ represents the initial spin, $\epsilon_{abc}$ is the Levi-Civita symbol, and the initial position of the beam is assumed at $\langle\hat{\mathbf{r}}\rangle_0=0$.
The first line of the equation refers to a straight path, while the second line shows that the beam oscillates around the equilibrium path periodically.
As a result, the emergent non-Abelian Lorentz force may lead to wavy trajectories for optical beams propagating in the non-Abelian media. 
This phenomenon resembles the ZB effect of Dirac particles~\cite{zawadzki2011zitterbewegung}. 
According to Eq.~(\ref{trajectory}), the trembling motion of light depends not only on the non-Abelian gauge fields but also on the initial spin $\vec{s}_0$ of the beam. If the initial state is purely one of the eigenmodes with the wave vector in $\mathbf{k}$ direction, i.e., $\vec{s}_0$ is along $\vec{\Omega}(\mathbf{k})$, the trembling term in Eq.~(\ref{trajectory}) will vanish. This implies
the present ZB effect stems from the interference of the two quasi-degenerate eigenmodes just as electronic ZB is caused by the superposition of positive and negative energy components (see Supplementary Note 3).

In recent years, ZB has been investigated for spin-orbit coupled atoms~\cite{vaishnav2008observing,gerritsma2010quantum} and photons~\cite{zhang2008observing,dreisow2010classical,fan2015plasmonic,chenjing2016OE}. However, unlike most schemes of ZB for light realized in periodic systems~\cite{zhang2008observing,dreisow2010classical,fan2015plasmonic}, our result shows that light can travel along curved paths even if the background medium is homogeneous. 
At first glance, this counterintuitive curved trajectory seems to violate the momentum conservation in translation invariant systems. 
However, it is well known that the kinetic momentum associated with centroid movement can be different from the canonical momentum for a particle traveling in a background vector potential. This conclusion is also valid for our situation. 
As shown in Eqs.~(\ref{momentum equation}) and (\ref{velocity}), the semiclassical canonical momentum $\langle\hat{\mathbf{p}}\rangle$  is always conserved in homogeneous media, while the kinetic momentum $m\langle\hat{\mathbf{v}}\rangle$ deviates from $\langle\hat{\mathbf{p}}\rangle$ and can change along the path by virtue of the synthetic non-Abelian potential $\hat{\mathcal{A}}$. A more rigorous analysis shows that the conserved quantity protected by space translational symmetry in generic non-Abelian media is the time-averaged Minkowski-type momentum $\int d^3x\, \mathrm{Re}\left(\mathbf{D}^*\times\mathbf{B}\right)$, while the centroid motion corresponds to the Abraham-type momentum $\int d^3x\, \mathrm{Re}\left(\mathbf{E}^*\times\mathbf{H}\right)/c^2$.

\textbf{Example I: ZB induced by non-Abelian magnetic field. }
According to the theory, the ZB effect for monochromatic beams can be generated by either non-Abelian magnetic fields or non-Abelian electric fields.
In Fig.~\ref{zitterbewegung}a--e, we first show an example of ZB induced solely by a non-Abelian magnetic field. To realize nonzero $\hat{\mathcal{B}}$ but vanishing $\hat{\mathcal{E}}$, we let the medium satisfy $\varepsilon_z=\mu_z$, $\mathbf{g}_1=-\mathbf{g}_2^*=(-\mathrm{i}\mathcal{A}^2_y/k_0,\,\mathcal{A}^1_x/k_0)^\intercal$, then the synthetic $\mathrm{SU}(2)$ magnetic field in this medium is given by $\hat{\mathcal{B}}=2{\mathcal{A}}_x^1{\mathcal{A}}_y^2\,\mathbf{e}_z\hat{\sigma}_3$.
The isofrequency surfaces of eigenmodes are illustrated in Fig.~\ref{zitterbewegung}a. Along the $k_x$ direction, the two eigenstates are $|\mathord{\rightarrow}\rangle=(1,1)^\intercal/
\sqrt{2}$ and $|\mathord{\leftarrow}\rangle=(1,-1)^\intercal/\sqrt{2}$ with the wave vectors $\mathbf{k}_{\pm}=\big[\sqrt{{k_0}^2\varepsilon_T\varepsilon_z-(\mathcal{A}^2_y)^2}\pm \mathcal{A}^1_x\big]\mathbf{e}_x$, and their pseudo-spins are polarized along the $\hat{\sigma}_1$-axis, as labeled on the Bloch sphere in Fig.~\ref{zitterbewegung}(c). As long as $|\mathcal{A}^1_x|\ll k=\sqrt{{k_0}^2\varepsilon_T\varepsilon_z-(\mathcal{A}^2_y)^2}$, the quasi-degenerate approximation is valid for beams incident from $x$ direction. 
In this case, the precession angular velocity is $\vec{\Omega}=-4k\mathcal{A}^1_x/\varepsilon_T\,\vec{e}_1$, so the pseudo-spin will precess about the $\hat{\sigma}_1$-axis. For the initial spin $\vec{s}_0=(\cos\theta_0,\,\sin\theta_0\cos\phi_0,\,\sin\theta_0\sin\phi_0)^\intercal$ at an angle $\theta_0$ from $\hat{\sigma}_1$-axis, we can obtain the centroid trajectory of the beam by eliminating the ray parameter $\tau$ in Eq.~(\ref{trajectory}),
\begin{equation}\label{trajectory mag}
   y(x)=Y_{\mathrm{ZB}}\left[\sin\big(k_{\rm ZB}( x-x_0)-\phi_0\big) +\sin\phi_0\right],
\end{equation}
where $x,y$ are the coordinates of centroid. The ZB amplitude
\begin{equation}\label{ZB amplitude}
Y_{\mathrm{ZB}}=\frac{-{\mathcal{A}}_y^2\,\sin\theta_0}{2\,k\,{\mathcal{A}}_x^1}=\frac{-{\mathcal{A}}_y^2\,\sin\theta_0}{2{\mathcal{A}}_x^1\sqrt{k_0^{\ 2}\varepsilon_T\varepsilon_z-(\mathcal{A}^2_y)^2}}
\end{equation}
is proportional to $\sin\theta_0$, so ZB reaches the maximum when the initial spin $\vec{s}_0$ is perpendicular to $\vec{\Omega}$, corresponding to the equal-weighted superposition of the two eigenmodes. Meanwhile, the ZB wave number
\begin{equation}
  k_{\rm ZB}=\frac{2k{\mathcal{A}}_x^1}{k-{\mathcal{A}}_x^1\cos\theta_0}\approx 2{\mathcal{A}}_x^1=k_+-k_-
\end{equation}
is  equal to the difference of the two eigen wave vectors, showing that ZB originates from the beating between the two eigenstates. Yet we should emphasize the phase beating is not a sufficient condition to realize ZB, and the ZB amplitude cannot be obtained without the knowledge of the non-Abelian dynamics. For instance, if $\mathcal{A}^2_y=0$ in the present medium, the beat of the two states persists, however, as the medium is relegated to the Abelian-type with $\hat{\mathcal{B}}=0$, ZB just disappears.

We have performed the full-wave simulation of a Gaussian beam propagating in this medium using COMSOL Multiphysics.  The beam is emitted along $x$-direction and the angle between its initial spin and $\hat{\sigma}_1$-axis is set as $\theta_0=0.43\pi$. Fig.~\ref{zitterbewegung}b shows the $k$-space Fourier  amplitude of the simulated wave function $\psi$, the two peaks in the spectrum manifest that the beam is mainly comprised of the two eigenstates $|\mathord{\rightarrow}\rangle$ and $|\mathord{\leftarrow}\rangle$.
The numerical time-averaged energy densities plotted in Fig.~\ref{zitterbewegung}d show clearly a transverse tremor along the beam. As shown in Fig.~\ref{zitterbewegung}e, the centroid trajectory  extracted from the full-wave result agrees perfectly with the analytical expression in Eq.~(\ref{trajectory mag}). And according to the numerical data of the pseudo-spins in one ZB period shown in Fig.~\ref{zitterbewegung}c, the spin precession about the $\hat{\sigma}_1$-axis is also demonstrated.

\textbf{Example II: ZB induced by non-Abelian electric field. }
In the previous example, the non-Abelian medium contains both gyroelectric and gyromagnetic components.
In fact, the synthetic non-Abelian gauge fields as well as ZB can be simply realized with reciprocal  media without gyrotropy. Here, we elaborate on synthesizing non-Abelian electric field with a biaxial non-magnetic material and the ZB effect in it.

We consider a non-magnetic material with the biaxial permittivity   $\tilde{\matr{\varepsilon}}/\varepsilon_0=\mathrm{diag}(\varepsilon_1,\varepsilon_2,\varepsilon_3)$
($\varepsilon_1<\varepsilon_2<\varepsilon_3$) along the principal axis and the permeability $\mu/\mu_0=1$. If the second principal axis of $\matr{\varepsilon}$ is fixed along the $y$-axis, while the first principal axis is rotated by an  angle $\varphi$ with respect to the $x$-axis such that $\cos^2\varphi\,\varepsilon_1+\sin^2\varphi\,\varepsilon_3=\varepsilon_2$ ($|\varphi|<\pi/2$) as shown in Fig.~\ref{zitterbewegung}f, the permittivity tensor in the $xyz$ coordinate system is given by
\begin{equation}\label{gyroelectric}
  \matr{\varepsilon}/\varepsilon_0=\left(\begin{array}{cc|c} \varepsilon_2 &  0 & g\\
0 & \varepsilon_2  & 0 \\\hline
g & 0 & \varepsilon_z
\end{array}\right),\
\end{equation}
with $\varepsilon_z=\varepsilon_1+\varepsilon_3-\varepsilon_2$ and $g=\mathrm{sgn}(\varphi)\sqrt{(\varepsilon_2-\varepsilon_1)(\varepsilon_3-\varepsilon_2)}$.
Since the in-plane duality condition is satisfied as $\varepsilon_T=\varepsilon_2\mu_T$ ($\mu_T=1$), by rescaling the vacuum permittivity $\varepsilon'_0=\varepsilon_2\varepsilon_0$, 
we obtain the synthetic gauge potentials
\begin{equation}
    \hat{\mathcal{A}}=-\frac{k_0\, g}{2\sqrt{\varepsilon_2}}\mathbf{e}_y\,\hat{\sigma}_1,\qquad
    \hat{\mathcal{A}}_0={k_0^{\ 2}}\,\frac{\varepsilon_1\varepsilon_3-\varepsilon_2^{\ 2}}{2\varepsilon_2}\hat{\sigma}_3,   
\end{equation}
and a uniform non-Abelian electric field polarized along the second principal axis
\begin{equation}
      \hat{\mathcal{E}}=k_0^{\ 3}\frac{g\,(\varepsilon_1\varepsilon_3-\varepsilon_2^{\ 2})}{2\ \varepsilon_2^{\ 3/2}}\mathbf{e}_y\,\hat{\sigma}_2.
\end{equation}
The two eigenstates in the $x$-direction are $|\mathord{\uparrow}\rangle=(1,0)^\intercal$ and $|\mathord{\downarrow}\rangle=(0,1)^\intercal$ corresponding to the two poles along the $\hat{\sigma}_3$-axis on the Bloch sphere (see Fig.~\ref{zitterbewegung}h), the eigen wave vectors are $\mathbf{k}_\uparrow=\sqrt{\varepsilon_1\varepsilon_3/\varepsilon_2}\,k_0\,\mathbf{e}_x$ and $\mathbf{k}_\downarrow=\sqrt{\varepsilon_2}\,k_0\,\mathbf{e}_x$ respectively. Providing that $|\sqrt{\varepsilon_1\varepsilon_3}/\varepsilon_2-1|$ is small enough, the centroid trajectory of a beam mainly consisting of these two states satisfies
\begin{equation}\label{trajectory elec}
     y(x)=Y_{\mathrm{ZB}}\left[\sin\big(k_{\rm ZB}( x-x_0)+\phi_0\big) -\sin\phi_0\right],
\end{equation}
where $\theta_0$, $\phi_0$ are the Euler angles of the initial spin $\vec{s}_0=(\sin\theta_0\cos\phi_0,\,\sin\theta_0\sin\phi_0,\,\cos\theta_0)^\intercal$, the ZB amplitude is
\begin{equation}\label{ZB amplitude NAElectric}
  Y_{\mathrm{ZB}}=\frac{\mathcal{A}^1_y\,\sin\theta_0}{\mathcal{A}_0^3}=\frac{\sqrt{\varepsilon_2}\,g\,\sin\theta_0}{k_0(\varepsilon_2^{\ 2}-\varepsilon_1\varepsilon_3)},
\end{equation}
and the ZB wave number
\begin{equation}\label{ZB wavenumber NAElectric}
  k_{\mathrm{ZB}}=\frac{k_0^{\ 2}(\varepsilon_2-\varepsilon_1\varepsilon_3/\varepsilon_2)}{2k}=\frac{{k_\downarrow}^2-{k_\uparrow}^2}{2k}\approx k_\downarrow-k_\uparrow
\end{equation}
is still determined by the beating of the two eigenstates. 
In the full-wave simulation of Fig.~\ref{zitterbewegung}i, we obtained a trembling beam (also see Fig.~\ref{zitterbewegung}g for its Fourier spectrum) where the decay of intensity along the beam is due to the beam divergence, the extracted centroid trajectory faithfully reproduces the analytic prediction of Eq.~(\ref{trajectory elec}), shown by Fig.~\ref{zitterbewegung}j. In Fig.~\ref{zitterbewegung}h, the numerical spin trajectory on the Bloch sphere also verifies that the pseudo-spin precesses about the $\hat{\sigma}_3$-axis.
\begin{figure}[t!]
\includegraphics[width=0.48\textwidth]{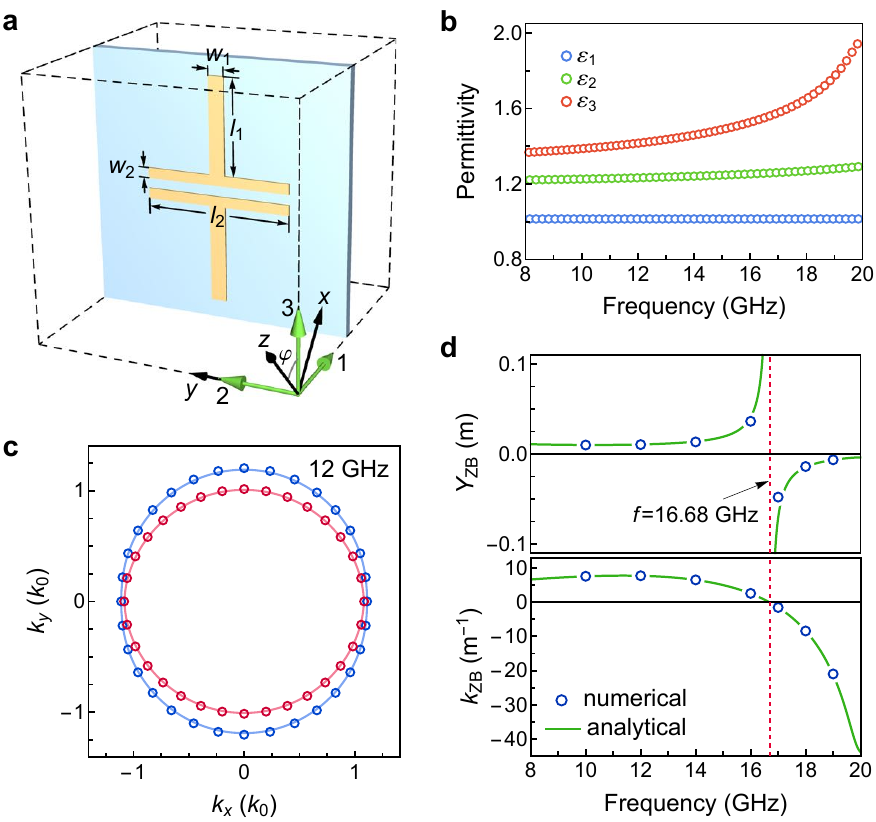}
\caption{\label{ZBmetamaterial} Design of a biaxial metamaterial. \textbf{a} Unit cell of the simple cubic structure with lattice constants $d=5\,\mathrm{mm}$, where a FR4 PCB slab (light blue) with thickness $0.2\,\mathrm{mm}$ and relative permittivity $\varepsilon_{\mathrm{pcb}}=3.3$ fills the coronal plane, and a copper structure of thickness $0.035\,\mathrm{mm}$ is printed on the PCB slab with the geometric parameters $l_1=2\,\mathrm{mm}$, $l_2=2.8\,\mathrm{mm}$, $w_1=0.3\,\mathrm{mm}$, and $w_2=0.2\,\mathrm{mm}$. \textbf{b} Dispersions of the retrieved relative permittivities along the three principal axes. \textbf{c} Isofrequency contours of the metamaterial in the $xy$-plane at $12\,\mathrm{GHz}$, where the circular markers and solid curves correspond, respectively, to the real structure in \textbf{a} and the retrieved homogeneous medium. \textbf{d} Full-wave simulated (blue circles) and analytical (green curves) ZB amplitude $Y_{\mathrm{ZB}}$ and ZB wave number $k_{\mathrm{ZB}}$ in the homogenized media varying with frequency. }
\end{figure}

In principle, the ZB effect induced by non-Abelian electric field can be observed in any natural and artificial biaxial non-magnetic materials. Here, we designed a simple metamaterial structure with the unit cell shown in Fig.~\ref{ZBmetamaterial}a for realizing ZB in microwave regime. The copper strips on printed circuit board (PCB) layers support strong and anisotropic electric dipole resonances along  principal axes labelled as 1,  2. Consequently, all the three principal values $\varepsilon_i\ (i=1,2,3)$ of the effective permittivity are different, and their dispersions obtained by $S$-parameter retrieval approach~\cite{liu2007description} are plotted in Fig.~\ref{ZBmetamaterial}b. According to our theory, the ZB beams should travel in the $xy$ plane whose orientation is determined by $\varepsilon_i$ and thus is frequency-dependent. As an example, we compared in Fig.~\ref{ZBmetamaterial}c the isofrequency contours in $xy$-plane of the real structure and that of the homogenized medium at 12~GHz. Their perfect consistency confirms the retrieval result. To test the ZB effect in the metamaterial, we numerically simulated the ZB beams with a constant waist of $0.2\,\mathrm{m}$ propagating along $x$-direction in the retrieved media at some discrete frequencies and extracted the ZB amplitudes $Y_{\mathrm{ZB}}$ and ZB wave numbers $k_\mathrm{ZB}$. We find good agreement with the theoretical predictions given by Eqs.~(\ref{ZB amplitude NAElectric}) and (\ref{ZB wavenumber NAElectric}) as shown in Fig.~\ref{ZBmetamaterial}d. Notably, both of the ZB amplitude and period tend to infinity at a singular frequency $f=16.68\,\mathrm{GHz}$, due to the fact that $\varepsilon_1\varepsilon_3=\varepsilon_2^{\,2}$ is accidentally satisfied at the frequency such that the material is reduced to Abelian type with $\hat{\mathcal{E}}=0$, and the beam splits into two branches~\cite{liu2018realizing}. We have also analyzed the finite width effect of the beam in the $z$-direction, and the analysis demonstrates that the 2D theory works well in the region where the two eigenmodes do not split away along the $z$-axis (see Supplementary Note 4).

\begin{figure*}
\includegraphics[width=0.85\textwidth]{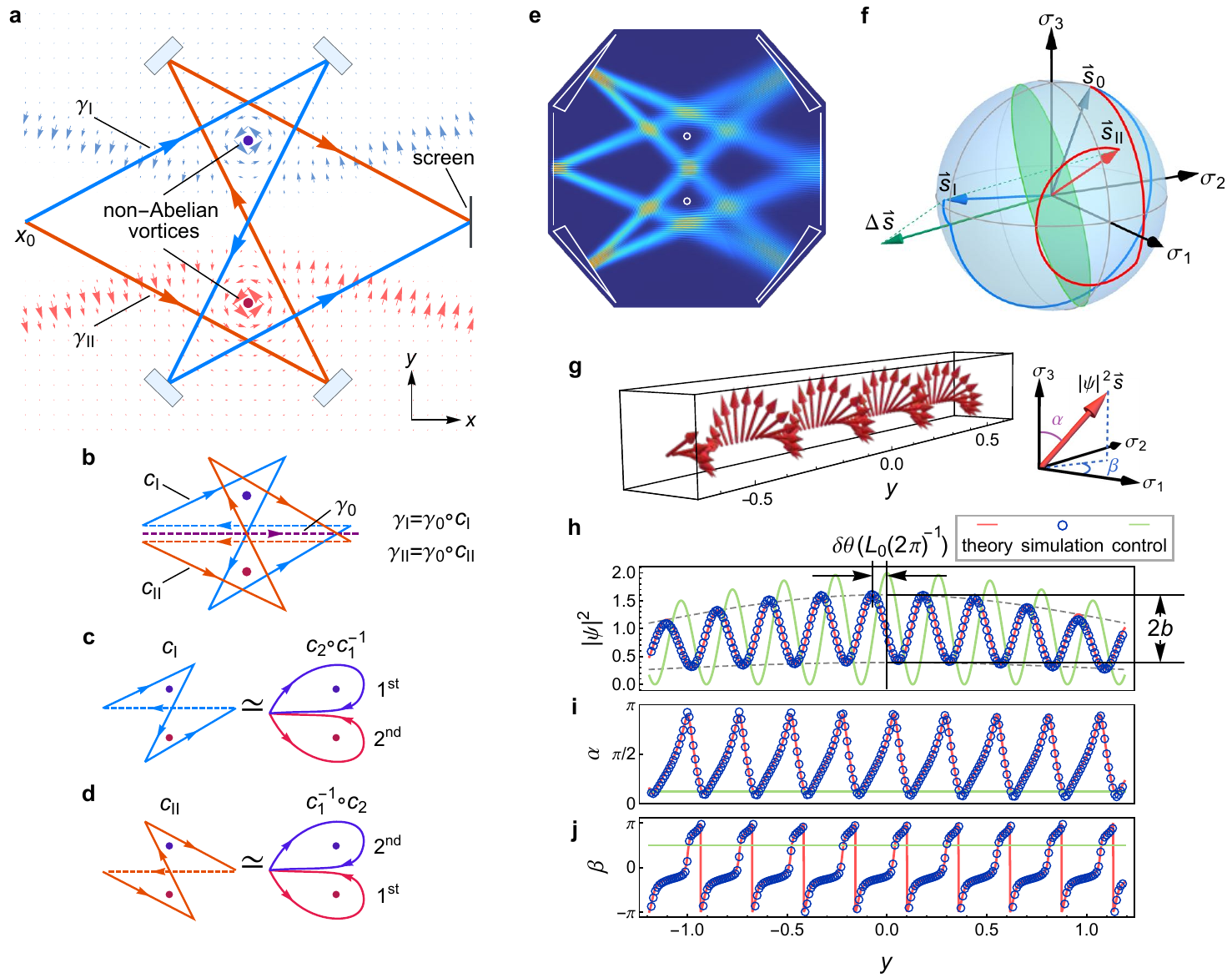}
\caption{\label{fig3} Genuine non-Abelian AB effect for light. \textbf{a}\ Sketch of the non-Abelian AB system with two optical paths $\gamma_{\mathrm{I}}$, $\gamma_{\mathrm{II}}$ interfering on the screen, where the background light blue (red) arrows denote the $\hat{\sigma}_1$ ($\hat{\sigma}_2$) component $\mathcal{A}^1$ ($\mathcal{A}^2$) of the non-Abelian vector potential.  \textbf{b}\ $\gamma_{\mathrm{I}}$ ($\gamma_{\mathrm{II}}$) can be divided into a closed loop $c_{\mathrm{I}}$ ($c_{\mathrm{II}}$) and a common path $\gamma_0$.
\textbf{c},\textbf{d}\ $c_{\mathrm{I}}$ and $c_{\mathrm{II}}$ can, respectively, deform continuously into a closed path that winds around the two vortices successively but in opposite sequences.
\textbf{e}\ Snapshot of the simulated field intensity for the proposed non-Abelian optical interferometer with incident spinor $(1,\mathrm{i}1/5)^\intercal$ for both beams and the vortex fluxes $\Phi_1=-{2\pi}/{3}$, $\Phi_2=-{\pi}/{3}$. \textbf{f}\ Spin evolution on the Bloch sphere along two beams $\gamma_{\mathrm{I}}$, $\gamma_{\mathrm{II}}$, which share the same initial spin $\vec{s}_0$ but achieve different final spins $\vec{s}_{\mathrm{I}}$ and $\vec{s}_{\mathrm{II}}$.
\textbf{g}\ Spin density interference corresponding to \textbf{e}, where each arrow denotes the local pseudo-spin density $|\psi|^2\vec{s}$ at a point on the screen. All of the local spins $\vec{s}(y)$ are perpendicular to $\Delta\vec{s}=\vec{s}_\mathrm{I}-\vec{s}_\mathrm{II}$, and thus fall on the green circle in \textbf{f}. The corresponding intensity interference $|\psi|^2(y)$ and the two Euler angles $\alpha$, $\beta$ of the local spins $\vec{s}(y)$ on the screen are shown in \textbf{h}--\textbf{j}, where blue circles and red curves indicate simulated and theoretical results respectively, and $\delta\theta$, $b$ are the phase shift and relative amplitude relative to the case of $\hat{\mathcal{A}}=0$ ($L_0$ is the period of $\Delta\theta(y)\mod 2\pi$). The green lines correspond to the ``control experiment''.}
\end{figure*}

\textbf{Non-Abelian Aharonov-Bohm system for light.} 
ZB discussed in the previous section can be viewed as the interference between two eigenmodes, each of which evolves with Abelian dynamics. In this sense, ZB is an apparent non-Abelian effect. Next, we will introduce the genuine non-Abelian AB effect, which cannot be reduced to Abelian subsystems.

The AB effect covers a group of phenomena associated with the path-dependent phase factors for particles traveling in a field-free region, but with irremovable gauge potential $\hat{\mathcal{A}}_\mu$,  the discovery of which confirmed the physical reality of gauge potentials and the nonlocality of gauge interactions~\cite{aharonov1959significance,batelaan2009aharonov}. The AB effect was first generalized to non-Abelian by Wu and Yang~\cite{Wuyang}, who showed that the scattering of nucleons (isospinors) around a non-Abelian flux tube (vortex) can generate peculiar phenomena. 
However, their governing Hamiltonian can be globally diagonalized into two decoupled  Abelian subsystems under a proper gauge~\cite{HovathyPRD}, and all relevant phenomena can be interpreted from the superposition of the two subsystems. Hence, Wu-Yang's proposal is now viewed as an apparent non-Abelian effect~\cite{Goldman2014,Raman1986}. 
According to a rigorous definition~\cite{Raman1986}, a genuine non-Abelian AB system requires its holonomy group $\mathrm{Hol}(\hat{\mathcal{A}})$ to be non-Abelian
(see Methods and Supplementary Note 5). As such, there should exist such loops based at a fixed point that  their non-Abelian AB phase factors (holonomies) are noncommutable, i.e. if a particle travels along two such loops in opposite sequences, the obtained AB phase factors would be different. This implies that at least two vortices exist in a genuine non-Abelian system~\cite{Raman1986}.

Indeed, we can use anisotropic and gyrotropic materials (see Table~\ref{gauge potentials}) to synthesize two vortices of $\mathrm{SU}(2)$ vector potential $\hat{\mathcal{A}}=\mathcal{A}^{1}\hat{\sigma}_1+\mathcal{A}^2\hat{\sigma}_2$ ($\hat{\mathcal{A}}_0=0$) with vanishing field $\hat{\mathcal{B}}=0$ in the whole space except for two small domains, taken as point singularities for simplicity. Here, we provide the synopsis of our scheme, and more details are given in Supplementary Note 6 (also see Supplementary Note 8 for an alternative design).
As illustrated in Fig.\ref{fig3}a, we demand $\hat{\mathcal{A}}=\mathcal{A}^{1}\hat{\sigma}_1$ ($\mathcal{A}^2=0$) in the upper half-space,
while $\hat{\mathcal{A}}=\mathcal{A}^{2}\hat{\sigma}_2$ ($\mathcal{A}^1=0$) in the lower half-space. We also require that  $\mathcal{A}^{1}$, $\mathcal{A}^{2}$ smoothly tend to zero in the middle region without overlap. In the vicinity of the upper (lower) singularity, $\mathcal{A}^{1}$ ($\mathcal{A}^{2}$) forms an irrotational vortex carrying the flux $\Phi_1$ ($\Phi_2$) (see Supplementary Equation 44 for the concrete expression of $\hat{\mathcal{A}}$ fulfilling these requirements). 
For a closed loop with a fixed base-point, its non-Abelian holonomy is invariant against continuous deformation of the path within the $\hat{\mathcal{B}}=0$ region. As a consequence, for the two homotopy classes of loops $[c_1]$ and $[c_2]$ (where $[c_i]$ denote the path homotopy classes; see Methods), based at $\mathbf{x}_0$ and encircling the upper (for $[c_1]$) or lower  (for $[c_2]$)  vortex once, their holonomies are $\hat{U}_i=\hat{\mathcal{U}}_{[c_i]}[\mathbf{x}_0]=\exp\left[\mathrm{i}\Phi_i\hat{\sigma}_i\right]$ ($i=1,2$) respectively. As $\hat{U}_1$ and $\hat{U}_2$ do not commute with each other, this double-vortex system is a genuine non-Abelian AB system.

In order to realize the vector potential shown in Fig.~\ref{fig3}a, the background media are set up as $\mathbf{g}_1=-\mathbf{g}_2^*$ (i.e. $\mathbf{g}_+=0$) and $\varepsilon_T=\varepsilon_z=\mu_T=\mu_z=\text{const.}$ to guarantee $\hat{\mathcal{A}}_0\equiv0$ and $V_0=\text{const.}$ Also, we use reciprocal anisotropic materials with purely real off-block-diagonal components $\mathbf{g}_1=-\mathbf{g}_2$
to build  the vector potential $\hat{\mathcal{A}}=\mathcal{A}^1\hat{\sigma}_1$ in the upper half plane but gyrotropic materials with purely imaginary $\mathbf{g}_1=\mathbf{g}_2$ to build $\hat{\mathcal{A}}=\mathcal{A}^2\hat{\sigma}_2$ in the lower half plane (see Supplementary Note 6 for details).
As a result, we have designed a genuine non-Abelian AB system for light. Then, we will show how the genuine non-Abelian nature of the system can be detected from interference effects.

\begin{figure*}[ht!]
\includegraphics[width=0.75\textwidth]{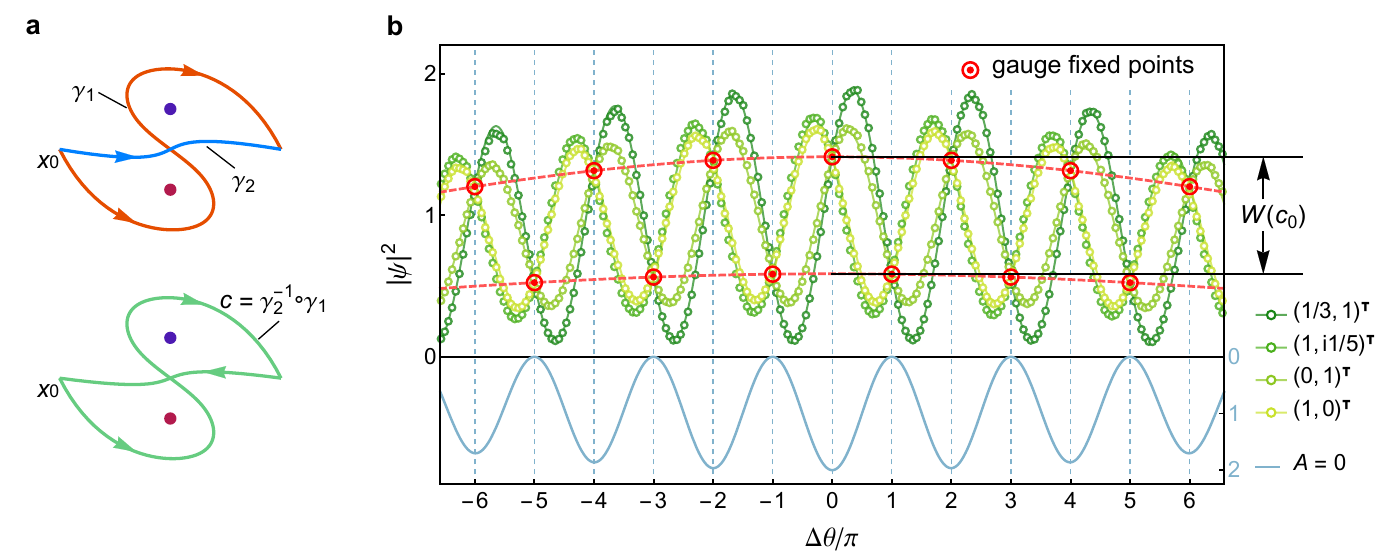}
\caption{\label{wilson loop} Extracting Wilson loops from gauge fixed points. \textbf{a}\ For two arbitrary beams $\gamma_1$, $\gamma_2$ interfering on the screen, the Wilson loop of the concatenate path $c=\gamma_2^{-1}\circ\gamma_1$ can be extracted from the interference fringes of the two beams. \textbf{b}\ Four intensity interference curves corresponding to four different incident spinors $(1,0)^\intercal$, $(0,1)^\intercal$, $(1,\mathrm{i}1/5)^\intercal$, $(1/3,1)^\intercal$ for the non-Abelian AB system shown in Fig.~\ref{fig3}\textbf{a} with vortex fluxes $\Phi_1=0.22\pi$, $\Phi_2=-0.33\pi$,  where circles and solid curves represent numerical and analytical results respectively. Their intersections, marked by red targets, are the gauge fixed points, which are located at the crests and troughs of the interference fringes (light blue curve) of $\hat{\mathcal{A}}=0$. The maximal difference between the envelops of even and odd gauge fixed points gives the Wilson loop $W(c_0)$ of $c_0=\gamma_{\mathrm{II}}^{-1}\circ\gamma_{\mathrm{I}}$.}
\end{figure*}

\textbf{Non-Abelian AB interference.} 
Consider two coherent light beams with the same initial spin $\vec{s}_0$ propagating separately along the two folded paths $\gamma_{\mathrm{I}}$ and $\gamma_{\mathrm{II}}$, 
and finally superposing on the screen (Fig.~\ref{fig3}a).  For the trivial situation of $\hat{\mathcal{A}}=0$, the two beams are uniformly polarized along the whole paths, thus their final states are given by $|\psi_{i}(y)\rangle=a(y)\mathrm{e}^{\mathrm{i}\theta_i(y)}|s_0\rangle\ $ ($i=\mathrm{I},\mathrm{II}$), where $a(y)$ is the envelope of both beams on the screen, $|s_0\rangle$ is the normalized initial spinor at $\mathbf{x}_0$, and $\theta_i(y)$ denote the dynamic phases which have included the initial phases. The dynamic phase difference, $\Delta\theta(y)=\theta_\mathrm{I}(y)-\theta_\mathrm{II}(y)$, determines the interference pattern: $|\psi_\mathrm{I}+\psi_\mathrm{II}|^2(y)=2a(y)\big[1+\cos(\Delta\theta(y))\big]$.

In the presence of the two non-Abelian vortices of $\hat{\mathcal{A}}$, the two optical paths are unchanged thanks to the null gauge field. However, the gauge potential drives the pseudo-spins to rotate along the paths, 
and the two final states convert to
\begin{equation}
  |\psi_{i}(y)\rangle=a(y)\,\hat{U}_{\gamma_{i}}\mathrm{e}^{\mathrm{i}\theta_i(y)}|s_0\rangle,\quad\ (i=\mathrm{I},\mathrm{II})
\end{equation}
where an additional non-Abelian AB phase factor $\hat{U}_{\gamma_{i}}=\mathcal{P}\exp\big[\mathrm{i}\int_{\gamma_i}\hat{\mathcal{A}}\cdot d\mathbf{r}\big]$ appears in each state.
The optical path of each beam can be regarded as a concatenation of a closed loop $c_{i}$ and a common path $\gamma_0$, i.e.,$\gamma_{i}=\gamma_{0}\circ c_{i}$  ($i=\mathrm{I},\mathrm{II}$), as illustrated in Fig.\ref{fig3}b.
The closed loop $c_{\mathrm{I}}$ can be further deformed continuously into two successive loops $c_2\circ c_1^{-1}$, which winds around the upper vortex (clockwise) first and subsequently the lower vortex (anticlockwise) (Fig.~\ref{fig3}c). Likewise, $c_{\mathrm{II}}$ is homotopic to $c_1^{-1}\circ c_2$, namely  $c_{\mathrm{II}}$ winds around the lower vortex first before it does the upper vortex  (Fig.~\ref{fig3}d).
Because of the noncommutativity of the sequences of winding around the two vortices, the AB phase factors of the two beams are different:
\begin{equation}
  \hat{U}_{\gamma_{\mathrm{I}}}=\hat{U}_{\gamma_0}\hat{U}_2\hat{U}_1^{-1}\neq
  \hat{U}_{\gamma_{\mathrm{II}}}=\hat{U}_{\gamma_0}\hat{U}_1^{-1}\hat{U}_2.
\end{equation}
Consequently, the two beams will end up with distinct spins $\vec{s}_\mathrm{I}$ and $\vec{s}_\mathrm{II}$  on the screen (Fig.~\ref{fig3}f), and they will interfere with each other in a nontrivial way. The term spin density interference was coined for this phenomenon and it can be calculated as follows:
\begin{equation}\label{spindensityinterference}
\big\langle {\psi}_\mathrm{I}(y)+{\psi}_{\mathrm{II}}(y)\big|\,\vec{\hat{\sigma}}\,\big|{\psi}_\mathrm{I}(y)+{\psi}_{\mathrm{II}}(y)\big\rangle=|{\psi}|^2(y)\, \vec{s}(y).
\end{equation}
Here, the angle bracket denotes the spinor inner product at a local position $y$ on the screen, the obtained result describes the spin density distribution on the screen. The spin density can be further decomposed into two parts: the intensity interference
\(|{\psi}|^2(y)\) and the spin orientation interference \(\vec{s}(y)\). 
The intensity interference part can be
 derived as
\begin{equation}\label{intensity interference}
\begin{split}
\left|{\psi}\right|^2(y)
=&\ 2\left[ a(y)^2+\mathrm{Re}\langle{\psi}_\mathrm{II}\big|\,{\psi}_\mathrm{I}\rangle(y)\right]\\
=&\ 2\,a(y)^2\left[1+\mathrm{Re}\left(\mathrm{e}^{\mathrm{i}\Delta\theta(y)}\big\langle s_0\big|\,\hat{\mathcal{U}}_{[c_0]}\,\big|s_0\big\rangle\right)\right]\\
=& \ 2\,a(y)^2\,\left[1+b\cos\left(\Delta\theta(y)+\delta\theta\right)\right],
\end{split}
\end{equation}
where $\hat{\mathcal{U}}_{[c_0]}=\hat{U}_{\gamma_\mathrm{II}}^{-1}\hat{U}_{\gamma_\mathrm{I}}=\hat{U}_2^{-1}\hat{U}_1\hat{U}_2\hat{U}_1^{-1}$ is the non-Abelian holonomy of the closed path $c_0=\gamma_{\mathrm{II}}^{-1}\circ\gamma_{\mathrm{I}}$.  
The nontrivial expectation value of the holonomy of $c_0$, $\big\langle s_0\big|\,\hat{\mathcal{U}}_{[c_0]}\,\big|s_0\big\rangle=b\,\mathrm{e}^{\mathrm{i}\delta\theta}\neq1$,  leads to a phase shift $\delta\theta$ and a change of the relative amplitude $b$ $(\leq 1)$ in comparison with the interference result of $\hat{\mathcal{A}}=0$.
In the mean time, the interfering spin orientation $\vec{s}(y)$ turns out to be always perpendicular to $\Delta\vec{s}=\vec{s}_\mathrm{I}-\vec{s}_\mathrm{II}$, namely lying on the green great circle of $\vec{s}(y)\cdot\Delta\vec{s}\equiv0$ in Fig.~\ref{fig3}f, and fluctuates around it (see Supplementary Note 7).

We have performed a full-wave simulation of this non-Abelian AB interference as shown in Fig.~\ref{fig3}e. In the simulation, the envelope $a(y)$ of each beam is set to be Gaussian type with a central  amplitude $a(0)=1/2$. The {spin density interference} is shown in Fig.~\ref{fig3}g, with  the intensity interference \(|{\psi}|^2(y)\) in Fig.~\ref{fig3}h, and the spin orientation given by Euler angles in Fig.~\ref{fig3}i,j. In Fig.~\ref{fig3}h-j, the blue circles are the simulated results,  which are fairly consistent with the theoretical results (red curves) obtained from Eq.~(\ref{spindensityinterference}).

To demonstrate that the non-Abelian feature of the above design is indeed genuine, we consider a control experiment with an almost identical system except that
the vector potential is $\hat{\mathcal{A}}\propto \hat{\sigma}_1$ in the whole space. In this case,  $\hat{U}_i=\exp[\mathrm{i}\Phi_i\hat{\sigma}_1]$ ($i=1,2$) commute with each other, and their winding around  the two vortices in opposite sequences gives the same AB phase factor $\hat{U}_{\gamma_{\mathrm{I}}}=\hat{U}_{\gamma_{\mathrm{II}}}=\exp\left[\mathrm{i}(\Phi_2-\Phi_1)\hat{\sigma}_1\right]$.
Thus, the interfering spin density is uniformly orientated, and there is no phase shift ($\delta\theta\equiv0$) and amplitude contraction ($b\equiv1$) compared with the case of $\hat{\mathcal{A}}=0$ (see green lines in Fig.~\ref{fig3}h-j).

\textbf{Measurement of Wilson loops.} 
In Abelian AB systems, the AB phase factor (holonomy) of a closed loop only
depends on the flux inside the loop but independent of the choice of gauge. However, in non-Abelian
systems, the holonomy $\hat{\mathcal{U}}_{[c]}[\mathbf{x}_0]$ of a closed path $c$ based at $\mathbf{x}_0$ varies as $\hat{\mathcal{U}}^\prime_{[c]}[\mathbf{x}_0]=\hat{U}(\mathbf{x}_0)\hat{\mathcal{U}}_{[c]}[\mathbf{x}_0]\hat{U}^\dagger(\mathbf{x}_0)$,
under a gauge
transformation
\(\hat{\mathcal{A}}^\prime=\hat{U}\hat{\mathcal{A}}\hat{U}^\dagger+\mathrm{i}\hat{U}\nabla_{T} \hat{U}^\dagger\).
Nevertheless, the trace of holonomy is an important gauge invariant observable, called the {Wilson loop} of the closed
path \(c\):
\begin{equation}
  W(c)=\mathrm{Tr}\left(\mathcal{P}\exp\left[\mathrm{i}\oint_c\hat{\mathcal{A}}\cdot d\mathbf{r}\right]\right)=\mathrm{Tr}\,\hat{\mathcal{U}}_{[c]}=\mathrm{Tr}\,\hat{\mathcal{U}}^\prime_{[c]}.
\end{equation}
In what follows, we show how to extract the Wilson loop of an
arbitrary closed path via interferometry.

In order to obtain the Wilson loop of a homotopy class $[c]$ in a non-Abelian AB system, we consider the interference of two beams along any two paths $\gamma_1$ and $\gamma_2$ as long as $\gamma_2^{-1}\circ\gamma_1=c$ forms a closed loop in the class $[c]$ as sketched in Fig.~\ref{wilson loop}a. As we deduced in Eq.~(\ref{intensity interference}), the holonomy of $c$, together with the initial spinor $|s_0\rangle$, determines the phase shift and the relative amplitude through the term $\big\langle s_0\big|\,\hat{\mathcal{U}}_{[c]}\,\big|s_0\big\rangle=b\,\mathrm{e}^{\mathrm{i}\delta\theta}$. In fact, its real part depends solely on the Wilson loop of $c$ (see proof in Methods):
\begin{equation}\label{real part}
  W(c)=2\,\mathrm{Re}\,\langle s_0|\,\hat{\mathcal{U}}_{[c]}\,|s_0\rangle=2\,b\cos\delta\theta.
\end{equation}
Thus, at certain positions $y_n$ satisfying $\Delta\theta(y_n)=n\pi$ ($n$ belongs to integers), the intensities only depend on the Wilson loop of $c$ and hence are  fixed under gauge transformation:
\begin{equation}\label{loopPattern}
\left|{\psi}\right|^2(y_n)\equiv a(y_n)^2\left[2+(-1)^n W(c)\,\right],
\end{equation}
where the two beams are supposed to share the same  envelope $a(y)$ on the screen, and the  locations $y_n$ correspond to  the crests and troughs in the interference fringes of $\hat{\mathcal{A}}=0$. These particular points in the intensity fringes are termed the gauge fixed points for the closed path $c$. Since the change of incident spin at $\mathbf{x}_0$ is equivalent to a global gauge transformation, the interference fringes for different incident spins should intersect at the gauge fixed points.

Using the above method, we examine the two optical paths $\gamma_{\mathrm{I}}$, $\gamma_{\mathrm{II}}$ in Fig.~\ref{fig3}a to extract the Wilson loop of $c_0=\gamma_{\mathrm{II}}^{-1}\circ\gamma_{\mathrm{I}}\simeq  c_2^{-1}\circ c_1\circ c_2\circ c_1^{-1}$. Figure~\ref{wilson loop}b shows the intensity interference curves corresponding to four different incident spins. Indeed, they intersect exactly at the gauge fixed points (red targets in Fig.~\ref{wilson loop}b) whose locations $y_n$ coincide  with the crests and troughs of the interference fringe pattern for $\hat{\mathcal{A}}=0$. By fitting the even and odd subsequences of the gauge fixed points, we obtain two curves $a(y)^2\left[2\pm W(c_0)\right]$ corresponding to the two red dashed lines in Fig.~\ref{wilson loop}b.
Thus, the Wilson loop $W(c_0)$ can be identified from the difference of the two dashed curves. 

\textbf{Discussion}\\
We have shown that the dynamics of 2D optical waves in a broad class of anisotropic media can be understood through an emergent  $\mathrm{SU}(2)$ gauge interaction in real space. We predicted that the \emph{Zitterbewegung} effect of light can be realized even in homogeneous anisotropic media, and we proposed a biaxial metamaterial to achieve synthetic non-Abelian electric field and ZB in microwave regime.  We have also designed a genuine non-Abelian AB system with two synthetic non-Abelian vortices, and suggested a spin density interferometry to demonstrate the noncommutative feature of non-Abelian holonomies. Our scheme opens the door to the colorful  non-Abelian world for light. In addition to inspiring new ideas to manipulate the flow and polarization of light, the scheme offers an optical platform to study physical effects relevant to $\mathrm{SU}(2)$ gauge fields, such as synthetic spin-orbit coupling~\cite{Galitski2013Spin} and topological band structures in periodic non-Abelian gauge fields~\cite{Osterloh2005Cold,goldman2009non,goldman2009ultracold,lepori2016double}.
Furthermore, since the $\mathrm{SU}(2)$ gauge field description is valid for photons down to quantum scale,
this approach might be applicable to the design of geometric gates for realizing non-Abelian holonomic quantum computation~\cite{zanardi1999holonomic,duan2001geometric} with photons.

\textbf{Methods}

\textbf{Notations.} 
In this paper, vectors in real space and in pseudo-spin space are indicated, respectively, by bold letters and letters with an overhead arrow ``$\rightarrow$''. Letters with an overhead bidirectional arrow ``$\leftrightarrow$'' denote 2-order tensors in real space. Symbols with an overhead hat ``$\wedge$'' denote operators acting on the spinor wave functions.  We use Greek letters, e.g. $\mu,\nu$, to denote indices of (2+1)-dimensional spacetime. Latin letters $i,j$ denote 2D spatial coordinate indices, and Latin letters $a,b,c$ denote indices in pseudo-spin space. We follow the Einstein summation convention for repeated indices. The orthonormal coordinate bases in real space and pseudo-spin space are expressed as $\mathbf{e}_i$ and $\vec{e}_a$ respectively.

\textbf{$\mathrm{SU}(2)$ gauge covariance of 2D Maxwell equations.}
In block-diagonalized duality symmetric media, $\matr{\varepsilon}/\varepsilon_0=\matr{\mu}/\mu_0=\mathrm{diag}(\matr{\varepsilon}_T,\varepsilon_z)$, the Maxwell's equations for 2D waves can be rearranged as
\begin{equation}\label{maxwell}
\hspace{-13pt}\underbrace{\left(\begin{array}{c|c}
0 & \mathrm{i}\hat{\sigma}_2 (\mathrm{i}\nabla_T\times)\mathbf{e}_z\\ \hline
\mathrm{i}\hat{\sigma}_2\, \mathbf{e}_z\cdot(\mathrm{i}\nabla_T\times) & 0
\end{array}\right)}_{\mathcal{M}}
\Psi
=k_0
\underbrace{\left(\begin{array}{c|c}
\hat{\sigma}_0\matr{\varepsilon}_T & 0\\\hline
0  & \hat{\sigma}_0\varepsilon_z
\end{array}\right)}_{\mathcal{N}}
\Psi,
\end{equation}
with $\Psi=(\mathbf{E}_T, \eta_0\mathbf{H}_T,E_z, \eta_0 H_z)^\intercal$.
For an arbitrary (global) transformation $\hat{U}\in \mathrm{SU}(2)$ acting on $|\psi\rangle=(E_z,\eta_0 H_z)^\intercal$, the corresponding transformation for $\Psi$ is defined as
\begin{equation}\label{EM transform}
\tilde{U}=\hat{U}_T\oplus \hat{U}=(\hat{\sigma}_2 \hat{U}\hat{\sigma}_2)\oplus \hat{U},
\end{equation}
which belongs to a 4D representation of $\mathrm{SU}(2)$. It turns out that $\mathcal{M}$ and $\mathcal{N}$ defined in Eq.~(\ref{maxwell}) transform according to
\begin{gather}
\hspace{-13pt}\tilde{U}\mathcal{M}\tilde{U}^\dagger=\left(\begin{array}{c|c}
0 & \mathrm{i}\hat{U}_T\hat{\sigma}_2U^\dagger (\mathrm{i}\nabla_T\times)\mathbf{e}_z\\ \hline
\mathrm{i}\hat{U}\hat{\sigma}_2\hat{U}_T^\dagger\, \mathbf{e}_z\cdot(\mathrm{i}\nabla_T\times) & 0
\end{array}\right)=\mathcal{M},\\
\tilde{U}\mathcal{N}\tilde{U}^\dagger=\mathcal{N}.
\end{gather}
Hence, the 2D Maxwell's equations are invariant under this $\mathrm{SU}(2)$ transformation.
As the EM duality transformation $\hat{R}\in \mathrm{SO}(2)$ 
is a special case of $\hat{U}$, the emergent $\mathrm{SU}(2)$ symmetry for the 2D Maxwell's equations in block-diagonalized duality symmetric materials is indeed the generalization of the original EM duality symmetry.

If $\hat{U}(x,y)$ is dependent on the $x,y$ coordinates, the transformation of $\mathcal{M}$ changes to
\begin{equation}
\tilde{U}(x,y)\mathcal{M}\tilde{U}^\dagger(x,y)=\mathcal{M}+\Delta\mathcal{M}
\end{equation}
with an additional term 
\begin{equation}
\begin{split}
\Delta\mathcal{M}
&=\left(\begin{array}{c|c}
0 & \mathrm{i}\hat{\sigma}_2 (\mathrm{i}\hat{U}\nabla_T \hat{U}^\dagger )\times\mathbf{e}_z \\ \hline
-\mathrm{i} (\mathrm{i}\hat{U}\nabla_T \hat{U}^\dagger )\times\mathbf{e}_z \hat{\sigma}_2 & 0
\end{array}\right)\\
&=\left(\begin{array}{c|c}
0 & (\hat{\sigma}_3 \mathcal{A}^1+\mathrm{i}\hat{\sigma}_0\mathcal{A}^2)\times\mathbf{e}_z \\ \hline
 (\hat{\sigma}_3 \mathcal{A}^1-\mathrm{i}\hat{\sigma}_0\mathcal{A}^2)\times\mathbf{e}_z & 0
\end{array}\right),
\end{split}
\end{equation}
where $\mathcal{A}^a\hat{\sigma}_a=\mathrm{i}\hat{U}\nabla_T \hat{U}^\dagger$ is precisely the vector potential induced purely by the gauge transformation, and only the components $\mathcal{A}^1, \mathcal{A}^2$ are supposed to exist. If we move the term $\Delta\mathcal{M}$ to the right side of the Maxwell equations~(\ref{maxwell}), it can be alternatively interpreted as a part of the constitutive tensor. By rotating $\Psi$ to the ordinary basis of EM field, 
\begin{equation} 
\begin{pmatrix}
\mathbf{E}\\
\eta_0\mathbf{H}
\end{pmatrix}=U_0\Psi=
\left(\begin{array}{cc|cc}
1 & 0 & 0 & 0 \\
0 & 0 & 1 & 0 \\\hline
0 & 1 & 0 & 0 \\
0 &0 & 0 & 1
\end{array}\right)\begin{pmatrix}
\mathbf{E}_T \\ \eta_0\mathbf{H}_T \\\hline
{E}_z \\ \eta_0{H}_z
\end{pmatrix},
\end{equation}
we obtain explicitly the contribution of $\Delta\mathcal{M}$ to the material tensors
\begin{equation}\label{material transform}
\begin{split}
&\left(\begin{array}{cc}
\matr{\varepsilon} & {0} \\
{0} & \matr{\mu}
\end{array}\right)=
U_0\Big(\mathcal{N}-\Delta\mathcal{M}/k_0\Big)U_0^\dagger,
\end{split}
\end{equation}
which shows that the effective $\mathrm{SU}(2)$ vector potential $\hat{\mathcal{A}}$ emerging in $\Delta\mathcal{M}$ just corresponds to the off-diagonal terms of $\matr{\varepsilon}$, $\matr{\mu}$:
\begin{equation}
  \mathbf{g}_1=-\mathbf{g}_2^*=\mathbf{e}_z\times(\mathcal{A}^1+\mathrm{i}\mathcal{A}^2)/k_0.
\end{equation}
Indeed, this relation is valid for arbitrary $\hat{\mathcal{A}}=\mathcal{A}^1\hat{\sigma}_1+\mathcal{A}^2\hat{\sigma}_2$ but not limited to the pure gauge case $\hat{\mathcal{A}}=\mathrm{i}\hat{U}\nabla_T \hat{U}^\dagger$. Furthermore, this correspondence can be generalized to any media satisfying  in-plane duality condition $\matr{\varepsilon}_T=\alpha \matr{\mu}_T$ where $\mathrm{SU}(2)$ scalar potential may also appear (Supplementary Note 1).

\textbf{Quasi-degenerate approximation for ZB. }
Eq.~(\ref{wave equation}) is essentially the stationary wave equation describing spin-1/2 particles coupling to the background $\mathrm{SU}(2)$ gauge fields without any approximation. However, the semiclassical trajectories of non-degenerate eigenmodes  often split away from each other. To manifest the coupling effects of different eigenmodes in the geometric optics, the media of concern are usually assumed to be weakly anisotropic~\cite{bliokh2007non}. Nevertheless, if the eigenmodes are approximately degenerate in a particular direction of wave vector but not necessarily in all directions, it turns out that an intact wave composed of modes in the vicinity of the quasi-degenerate direction can be described adequately by the semiclassical approach including the interaction between eigenmodes in their interfering region~\cite{Kravtsov1990}.  

In homogeneous non-Abelian media, we separate the effective Hamiltonian into two parts:
\begin{equation}
  \hat{H}(\mathbf{k})=\underbrace{\left[\frac{1}{2m}\left(\mathbf{k}^2+(\hat{\mathcal{A}})^2\right)+V_0\right]\hat{\sigma}_0}_{\displaystyle\hat{H}_0(\mathbf{k})}+\underbrace{\left(\frac{-1}{m}\mathbf{k}\cdot\hat{\mathcal{A}}-\hat{\mathcal{A}}_0\right)}_{\textstyle\delta\hat{H}(\mathbf{k})=\vec{\Omega}\cdot\vec{\hat{\sigma}}/2}.
\end{equation}
If only $\hat{H}_0$ is present, the isofrequency surface is a doubly degenerate sphere with the radius $k=\sqrt{-2m\,V_0-(\hat{\mathcal{A}})^2}$. When $\delta\hat{H}(\mathbf{k})$ is taken into account, as long as it is sufficiently small for a given direction $\mathbf{e}_k$, the two eigenstates can be regarded as quasi-degenerate at the wave vector 
\begin{equation}\label{quasi-degeneratre k}
\mathbf{k}=k\mathbf{e}_k=\sqrt{-2m\,V_0-(\hat{\mathcal{A}})^2}\ \mathbf{e}_k, 
\end{equation}
and we can implement the eikonal approximation to the  wave function mainly superposed by the two quasi-degenerate modes~\cite{Kravtsov1990}: $\psi=\tilde{\psi}(\mathbf{r})\exp(\mathrm{i}\mathbf{k}\cdot\mathbf{r})$ with a slowly varying envelope $\tilde{\psi}(\mathbf{r})$ (i.e. $|\nabla\tilde{\psi}/\tilde{\psi}|\ll k$). Subsituting $\psi$ into the wave equation~(\ref{wave equation}), we obtain the equation of $\tilde{\psi}$ with accuracy up to the first order of $k$:
\begin{equation}\label{envelope eq}
  \mathrm{i}\,\hat{\mathbf{v}}\cdot\nabla\tilde{\psi}=\hat{H}(\mathbf{k})\tilde{\psi}.
\end{equation}
By adopting the ansatz that the velocity operator $\hat{\mathbf{v}}=\partial\hat{H}(\mathbf{k})/\partial\mathbf{k}=(\mathbf{k}-\hat{\mathcal{A}})/m$ can be replaced by its averaged value $\langle\hat{\mathbf{v}}\rangle$ over the transverse cross section of an optical beam, we find that the operator
$\hat{\mathbf{v}}\cdot\nabla\ \rightarrow\  \langle\hat{\mathbf{v}}\rangle\cdot\nabla=d/d\tau$  corresponds to the total derivative with respect to the ray parameter $\tau$ along the beam. Therefore, Eq.~(\ref{envelope eq}) is reformulated into a time-dependent Schr\"{o}dinger equation
\begin{equation}\label{envelope eq2}
  \mathrm{i}\frac{d}{d\tau}\tilde{\psi}=\hat{H}(\mathbf{k})\tilde{\psi}.
\end{equation}
Consequently, Eqs.~(\ref{momentum equation})--(\ref{spin precession}) can be directly obtained in terms of Ehrenfest theorem.

\textbf{Relation between Poynting vector and velocity operator. }
The in-plane projection of the time-averaged Poynting vector $\bar{\mathbf{S}}_T$ for monochromatic waves can be written as
\begin{equation}\label{poynting vector}
  \begin{split}
  \bar{\mathbf{S}}_T=&\,\frac{1}{2}\mathrm{Re}\left[\mathbf{E}_z^*\times\mathbf{H}_T+\mathbf{E}_T^*\times\mathbf{H}_z\right]\\
  =&\,\frac{1}{2}\mathrm{Re}\left[(\mathbf{E}^*_z,\,\mathbf{H}^*_z)\begin{pmatrix}0 & \matr{\mathrm{I}}\times\matr{\mathrm{I}} \\ -\matr{\mathrm{I}}\times\matr{\mathrm{I}} & 0\end{pmatrix} 
  \begin{pmatrix}
    \mathbf{E}_T\\ \mathbf{H}_T
  \end{pmatrix}\right]\\
  =&\,\frac{1}{2\eta_0}\mathrm{Re}\left[({E}^*_z,\,\eta_0{H}^*_z) \big(\mathrm{i}\,\hat{\sigma}_2\epsilon^{izj}\mathbf{e}_i\big)  
  \begin{pmatrix}
    \mathbf{E}_T\\ \eta_0\mathbf{H}_T
  \end{pmatrix}_j
  \right].
 \end{split}
\end{equation}\vspace{-5pt}
Substituting Maxwell's equations into Eq.~(\ref{poynting vector}) yields
\begin{equation}\label{poynting vector2}
  \begin{split}
    \bar{\mathbf{S}}_T=&\, \frac{1}{2\eta_0}\mathrm{Re}\Big\langle\psi\Big| \big(\mathrm{i}\hat{\sigma}_2\epsilon^{izj} \mathbf{e}_i\big)\frac{-\mathrm{i}\,\hat{\sigma}_2\epsilon_{jkz}}{k_0\varepsilon_T}
    \big(\hat{\mathbf{p}}-\hat{\mathcal{A}}_{(c)}\big)^k \Big|\psi\Big\rangle\\
    =&\,\frac{1}{2\eta_0 k_0}\mathrm{Re}\Big\langle\psi\Big|\frac{1}{2m}\big(\hat{\mathbf{p}}-\hat{\mathcal{A}}_{(c)}\big) \Big|\psi\Big\rangle\\
    =&\,\frac{1}{4\mu_0\omega_0}\Big\langle\psi\Big|\frac{1}{m}\big(\mathbf{k}-\hat{\mathcal{A}}\big) \Big|\psi\Big\rangle
    =\frac{1}{4\mu_0\omega_0}\big\langle\psi\big|\hat{\mathbf{v}}\big|\psi\big\rangle,
\end{split}
\end{equation}
where $\hat{\mathcal{A}}_{(c)}=\frac{k_0}{2}\left\{\left[(\mathbf{g}_1-\mathbf{g}_2)\times\mathbf{e}_z\right]\hat{\sigma}_1-\mathrm{i}\left[(\mathbf{g}_1+\mathbf{g}_2)\times\mathbf{e}_z\right]\hat{\sigma}_2\right\}$.
In the third step, we replaced $\hat{\mathbf{p}}$ with $\mathbf{k}$ according to the eikonal approximation. As a result, the total in-plane energy flux over a transverse cross section of the optical beam is propotional to the expectation value of the velocity operator:
\begin{equation}
  \langle\bar{\mathbf{S}}_T\rangle=\frac{1}{4\mu_0\omega_0}\int d\mathbf{r}_\perp\,\big\langle\psi\big|\hat{\mathbf{v}}\big|\psi\big\rangle= \frac{1}{4\mu_0\omega_0}\langle\hat{\mathbf{v}}\rangle.
\end{equation}
And it shows that the time-averaged Poynting vector $\bar{\mathbf{S}}_T$ is invariant under the gauge transformation Eq.~(\ref{EM transform})  for EM fields (Supplementary Note 2).

\textbf{Holonomy and genuine non-Abelian AB system.}
From a geometric viewpoint, gauge potential and field in the physical space $M$ can be described as the connection and curvature in a $G$-principle fiber bundle~\cite{Wuyang}, where the physical space serves as the base manifold, and $G$ denotes the gauge group, in our case $G=\mathrm{SU}(2)$.  Consider a particle (wave packet) travels in the physical space. Along its trajectory $\gamma$, the gauge potential engenders a matrix-valued geometric phase factor $\mathcal{P}\exp\left[\mathrm{i}\int_{\gamma}\hat{\mathcal{A}}_\mu dx^\mu\right]\in G$ ($\mathcal{P}$ denotes path-ordering) on the state vector, corresponding to the parallel transport of the state in the bundle space. In particular, for a closed path $c$ starting and ending at the same point $c(0)=c(1)=\mathbf{x}_0$, the phase factor of $c$,
\begin{equation}
  \hat{\mathcal{U}}_c(\hat{\mathcal{A}})=\mathcal{P}\exp\left[\mathrm{i}\oint_{c}\hat{\mathcal{A}}_\mu dx^\mu\right],
\end{equation}
is called the holonomy of the closed path $c$ with respect to the gauge $\hat{\mathcal{A}}$. The collection of the holonomies corresponding to all those closed paths based at the same point $\mathbf{x}_0$ forms a subgroup of the gauge group $G$:
\begin{equation}
\mathrm{Hol}(\hat{\mathcal{A}})=\left\{\hat{\mathcal{U}}_c(\hat{\mathcal{A}})\,\big|\ {c}(0)={c}(1)=\mathbf{x}_0 \right\}\subseteq G,
\end{equation}
which is the holonomy group for the gauge $\hat{\mathcal{A}}$. In the literature, a gauge system is regarded as genuinely non-Abelian if and only if the holonomy group is a non-Abelian group, namely the holonomies of some loops are noncommutable with each other~\cite{Goldman2014,Raman1986}. If the base manifold is simply a Euclidean space, the noncommutativity of holonomies can be traced back to noncommutable gauge fields $[\hat{\mathcal{F}}_{\mu\nu},\hat{\mathcal{F}}_{\mu'\nu'}]\neq0$. However, if the base manifold possesses nontrivial topology, noncommutative holonomies can  be achieved even though the gauge field vanishes everywhere (i.e. AB systems).

For an AB system, the corresponding fiber bundle is a flat bundle, since the curvature (field) $\hat{\mathcal{F}}_{\mu\nu}=0$ in the whole base manifold $M$ (flux regions are excluded from $M$). Here, the topology of the base manifold is characterized by its first fundamental group, 
\begin{equation}
  \pi_1(M)=\{[c]\,|\,c(0)=c(1)=\mathbf{x}_0\}, 
\end{equation}
which is the set of path homotopy equivalent classes $[c]$ of closed paths based at $\mathbf{x}_0$. Path homotopy is a topologically equivalent relation ``$\simeq$'' for paths. If two paths $c_1$, $c_2$ with the same fixed base-point $\mathbf{x}_0$
can deform into each other continuously, they are said to be path homotopic $c_1\simeq c_2$ and to belong to the same homotopy class $[c_1]$. In flat bundles, the holonomies (AB phase factors) of all loops in the same homotopy class $[c]$ are identical: 
$\hat{\mathcal{U}}_{[c]}$ (see proof in Supplementary Note 5). Based on this property, two necessary conditions for genuine non-Abelian AB systems can be obtained~\cite{Raman1986}:
\begin{enumerate}[itemsep=0mm]
  \item The gauge group $G$ is non-Abelian;
  \item The first fundamental group $\pi_1(M)$ is non-Abelian.
\end{enumerate}

According to the second criterion, the Wu-Yang AB system is not genuinely non-Abelian, because the fundamental group of its base manifold (a punctured plane $\mathbb{R}^2-\mathbf{0}$) is an Abelian group $\pi_1(\mathbb{R}^2-\mathbf{0})=\mathbb{Z}$. However, for a twice-punctured plane as shown in Fig.~\ref{fig3}(a), its fundamental group is the free group on two generators, $\mathbb{Z}*\mathbb{Z}$  (where $*$ denotes a free product), which is non-Abelian~\cite{massey1967algebraic}. Therefore, a twice-punctured plane is a qualified prototype of a genuine non-Abelian AB system.

\textbf{Gauge fixed points. }
The derivation of the intensity interference given by Eq.~(\ref{intensity interference}) is in fact valid for  two arbitrary interfering beams $\gamma_1$, $\gamma_2$ with the same initial spin $\vec{s}_0$ and final envelop $a(y)$: $|\psi|^2=2a(y)\big[1+\mathrm{Re}\big(\mathrm{e}^{\mathrm{i}\Delta\theta(y)}\langle s_0|\,\hat{\mathcal{U}}_{[c]}\,|s_0\rangle\big)\big]$,
where $\hat{\mathcal{U}}_{[c]}$ is the holonomy of the closed path $c=\gamma_2^{-1}\circ\gamma_1$.
Since
\(\hat{\mathcal{U}}_{[c]}\in \mathrm{SU}(2)\),
it can be generically expressed as
\begin{equation}
  \hat{\mathcal{U}}_{[c]}=\begin{pmatrix}
u_1 & u_2\\
-u_2^* & u_1^*
\end{pmatrix},\quad \text{with}\quad |u_1|^2+|u_2|^2=1.
\end{equation}
The Wilson
loop reads \(W(c)=\mathrm{Tr}\,\hat{\mathcal{U}}_{[c]}=2\,\mathrm{Re}(u_1).\) For an
arbitrary spinor state \(|s_0\rangle=(\cos\frac{\alpha}{2}\mathrm{e}^{-\mathrm{i}\beta/2},\, \sin\frac{\alpha}{2}\mathrm{e}^{\mathrm{i}\beta/2})^\intercal\),
we have
\begin{equation}
  \begin{split}
    &\langle s_0|\,\hat{\mathcal{U}}_{[c]}\,|s_0\rangle
=b\, \mathrm{e}^{\mathrm{i}\delta\theta}\\
=&\ \mathrm{Re}(u_1)+\mathrm{i}\left[\cos\alpha\,\mathrm{Im}(u_1)+\sin\alpha\, \mathrm{Im}\big(u_2\mathrm{e}^{\mathrm{i}\beta}\big)\right].
\end{split}
\end{equation}
Therefore, the identity in Eq.~(\ref{real part}) is established for any \(|s_0\rangle\).

In fact, different incident spinors can interconvert  through a global gauge transformation: $|s_0'\rangle=\hat{U}|s_0\rangle$. Hence, the relation in Eq.~(\ref{real part}) is straightforward
\begin{equation}
\begin{split}
2\,\mathrm{Re}\,\langle s_0'|\,\hat{\mathcal{U}}_{[c]}\,|s_0'\rangle=& 2\,\mathrm{Re}\,\big\langle s_0\big|\,\hat{U}^{-1}\hat{\mathcal{U}}_{[c]}\hat{U}\,\big|s_0\big\rangle \\
=& \mathrm{Tr}\left(\hat{U}^{-1}\hat{\mathcal{U}}_{[c]}\hat{U}\right)\equiv W(c).
\end{split}
\end{equation}
As a result, at the positions such that \(\Delta\theta(y_n)=n\pi\), i.e.~at
the crests and troughs of the original interference fringes when \(\hat{\mathcal{A}}=0\), the intensities given in Eq.~(\ref{loopPattern}) are fixed for arbitrary
incident spins, yet they are only determined by the Wilson loop \(W(c)\), provided that the dynamic phases of $\gamma_1$, $\gamma_2$ are unchanged.

For the two optical path $\gamma_\mathrm{I}$, $\gamma_\mathrm{II}$ in Fig.~\ref{fig3}, the Wilson loop of $c_0=\gamma^{-1}_{\mathrm{II}}\circ\gamma_\mathrm{I}$ is determined by the fluxes of the two vortices as $W(c_0)=2-4\sin^2\Phi_1\sin^2\Phi_2$. Therefore, if $\sin^2\Phi_1\sin^2\Phi_2=1/2$, $W(c_0)$ will be reduced to zero, and the two dashed curves in Fig.~\ref{wilson loop} will completely overlap (see Supplementary Note 7 for details).

\textbf{Simulation of non-Abelian AB interference. }
The full-wave results of the non-Abelian AB interference shown in Figs.~\ref{fig3},~\ref{wilson loop} are simulated using the commercial software COMSOL Multiphysics. In order to avoid spin flip after reflection, the mirrors shown in Fig.~\ref{fig3}a,e are made of an impedance-matched material, namely $\varepsilon_m/\mu_m=1$, with a lower refractive index than the surrounding media to achieve total reflection at their surfaces. Meanwhile, the two mirrors on the right-hand side in Fig.~\ref{fig3}e are slightly concave, so that the reflected beams with reduced widths can bypass the two singularities.

\textbf{Data availability}\\
The authors declare that all data supporting the findings of this study are available from the corresponding authors upon reasonable request.

\def\bibsection{\section*{}}
\textbf{References}\\
\vspace{-65pt}
\bibliographystyle{naturemag}
\bibliography{reference}

\vspace{0pt}
\textbf{Acknowledgements}\\
We thank Profs. Bo Hou, and Shubo Wang for helpful discussions.
This work was supported in part by National Natural Science Foundation of China (Grant Nos. 11874026, 11174250, 11574226, and 11874274), and the Fundamental Research Funds for the Central Universities, HUST: 2017KFYXJJ027. The work in Hong Kong was supported by Research Grants Council of Hong Kong (Grant Nos. AoE/P-02/12 and C6013-18G).

\textbf{Author contributions}\\
Y.C. conceived the original idea. 
Y.C., J.Q.S. developed the standard framework and discovered ZB effect. 
R.-Y.Z., Y.C. extended the framework to more general cases. 
R.-Y.Z. performed analytic study of ZB effect, and designed the scheme of non–Abelian AB effect.
Y.C., Z.X. R.-Y.Z. carried out the numerical simulations. 
Z.H.H., J.L. designed the metamaterial.
C.T.C supervised the whole project. 
R.-Y.Z., Y.C., C.T.C. wrote the manuscript. 
And all authors were involved in the analysis and discussion of the results.

\textbf{Competing interests}\\
The authors declare no competing interests.

%
%


%

\end{document}


\title{\Large Supplementary information for ``Non-Abelian gauge field optics''}

\maketitle





\section{Synthetic $\mathrm{U}(2)$ gauge fields in real space for light}
In the main text, the materials considered are restricted to anisotropic media without magnetoelectric (ME) coupling. In that case, an incomplete set of synthetic $U(2)=\mathrm{SU}(2)\rtimes \mathrm{U}(1)$ gauge potentials is obtained, while the $\mathcal{A}^3\hat{\sigma}_3$ component of $\mathrm{SU}(2)$ vector potential and $\mathrm{U}(1)$ vector potential are absent. Here, we consider the more general case of bi-anisotropic media through adding a specific ME coupling term, in order to obtain a complete set of $\mathrm{U}(2)$ gauge potentials.
In frequency domain, the constitutive relations of nondissipative and nondispersive bi-anisotropic media read
\begin{equation}\label{constitutive relations}
\mathbf{D}_\omega=\tensor{\varepsilon}\cdot\mathbf{E}_\omega+\tensor{\chi}_{em}\cdot\mathbf{H}_\omega,\qquad \mathbf{B}_\omega=\tensor{\mu}\cdot\mathbf{H}_\omega+\tensor{\chi}_{me}\cdot\mathbf{E}_\omega,
\qquad (\omega>0)
\end{equation}
with the constitutive coefficient tensors
\begin{equation}\label{constitutive coefficient}
\arraycolsep=1.4pt\def\arraystretch{1.1}
\tensor{\varepsilon}/\varepsilon_0=\left(\begin{array}{c|c}\tensor{\varepsilon}_T & \mathbf{g}_1 \\\hline
\mathbf{g}_1^\dagger & \varepsilon_z
\end{array}\right),\quad
\tensor{\mu}/\mu_0=\left(\begin{array}{c|c}
\tensor{\mu}_T & \mathbf{g}_2 \\\hline
\mathbf{g}_2^\dagger & \mu_z
\end{array}\right),\quad
\tensor{\chi}_{em}=\tensor{\chi}_{me}^\dagger=\tensor{\chi}=\frac{1}{c}\left(\begin{array}{c|c}
\mathbf{0} & \mathbf{t}_1 \\[2pt]\hline
\mathbf{t}_2^\intercal & \chi_{z}
\end{array}\right).
\end{equation}
Here, all parameters are functions of \(x\) and \(y\). The only constraint on $\tensor{\varepsilon},\ \tensor{\mu}$ is the ``in-plane duality'' $\tensor{\varepsilon}_T=\alpha\tensor{\mu}_T\in\mathbb{R}$ ($\alpha$ is an arbitrary positive constant). In the following derivations, we let $\alpha=1$ for convenience, and the results of $\alpha\neq 1$ is straightforward via replacing $\varepsilon_0\rightarrow\varepsilon'_0= \alpha\,\varepsilon_0 $
 and $c\rightarrow c'=1/\sqrt{\varepsilon'_0\mu_0}$ in Eq.~(\ref{constitutive coefficient}). The ME coupling tensors $\tensor{\chi}_{em},\ \tensor{\chi}_{me}$ are assumed to be purely real and do not have in-plane block for convenience. 

\subsection*{Derivation of in-plane wave equation}
In this section, we derive the 2D wave equations for monochromatic waves in the bi-anisotropic media given by Eq.~(\ref{constitutive coefficient}), where a complete set of synthetic $\mathrm{U}(2)$ gauge potentials emerge.
For 2D propagating waves, the fields $\mathbf{E},\,\mathbf{H}$ are only functions of $x,y$, so all terms associated with $\partial/\partial z$ are dropped, and the source-free Maxwell's equations for the complex-valued EM fields (analytic signals) can be expressed as
\begin{subequations}\label{Maxwell eq}
\begin{align}\label{Maxwell eq 1}
  \left[\left(\begin{array}{cc}\nabla_T\times & 0 \\ 0 & \nabla_T\times\end{array}\right)+
  \frac{\partial}{\partial t}\frac{1}{c}\left(\begin{array}{cc}\tilde{\mathbf{t}}_2\times & \tilde{\mathbf{g}}_2\times \\ -\tilde{\mathbf{g}}_1\times & -\tilde{\mathbf{t}}_1\times\end{array}\right)\right]
  \left(\begin{array}{c}\mathbf{E}_z\\ \eta_0\mathbf{H}_z\end{array}\right)
  &=-\frac{\partial}{\partial t}\frac{1}{c}\left(\begin{array}{cc}0 & \tensor{\mu}_T \\ -\tensor{\varepsilon}_T & 0\end{array}\right)\left(\begin{array}{c}\mathbf{E}_T\\ \eta_0\mathbf{H}_T\end{array}\right),\\[8pt]\label{Maxwell eq 2}
  \left[\left(\begin{array}{cc}\nabla_T\times & 0 \\ 0 & \nabla_T\times\end{array}\right)-
  \frac{\partial}{\partial t}\frac{1}{c}\left(\begin{array}{cc}\tilde{\mathbf{t}}_1\times & \tilde{\mathbf{g}}_2^\dagger\times \\ -\tilde{\mathbf{g}}_1^\dagger\times & -\tilde{\mathbf{t}}_2\times\end{array}\right)\right]
  \left(\begin{array}{c}\mathbf{E}_T\\ \eta_0\mathbf{H}_T\end{array}\right)
  &=-\frac{\partial}{\partial t}\frac{1}{c}\left(\begin{array}{cc}\chi_z & {\mu}_z \\ -{\varepsilon}_z & -\chi_z\end{array}\right)\left(\begin{array}{c}\mathbf{E}_z\\ \eta_0\mathbf{H}_z\end{array}\right),
\end{align}
\end{subequations}
where $\nabla_T=(\partial_x,\partial_y)^\intercal$, $\eta0=\sqrt{\mu_0/\varepsilon_0}$, $\tilde{\mathbf{g}}_i=\mathbf{e}_z\times\mathbf{g}_i=(-g_{iy},g_{ix})^\intercal$, $\tilde{\mathbf{t}}_i=\mathbf{e}_z\times\mathbf{t}_i=(-t_{iy},t_{ix})^\intercal$ ($i=1,2$). For monochromatic waves $\mathbf{E},\,\mathbf{H}\propto \exp(-i\omega t)$,
after substituting $\partial/\partial t\rightarrow -i\omega$, Eq.~(\ref{Maxwell eq}) can be rewritten using Pauli matrices as
\begin{align}\label{maxeq1}
  \Big[ \hat{\sigma}_0\left(\nabla_T\times+ik_0\tilde{\mathbf{t}}_-\times\right)
  +\hat{\sigma}_1(ik_0\tilde{\mathbf{g}}_-\times)
  -i\hat{\sigma}_2(ik_0\tilde{\mathbf{g}}_+\times)
  -\hat{\sigma}_3(ik_0\tilde{\mathbf{t}}_+\times) \Big]\left(\begin{array}{c}\mathbf{E}_z\\ \eta_0\mathbf{H}_z\end{array}\right)
  &=\hat{\sigma}_2\big(-k_0\tensor{\varepsilon}_T\big)\left(\begin{array}{c}\mathbf{E}_T\\ \eta_0\mathbf{H}_T\end{array}\right),\\[8pt]\label{maxeq2}
  \Big[ \hat{\sigma}_0\left(\nabla_T\times+ik_0\tilde{\mathbf{t}}_-\times\right)
  -\hat{\sigma}_1(ik_0\tilde{\mathbf{g}}^\dagger_-\times)
  +i\hat{\sigma}_2(ik_0\tilde{\mathbf{g}}^\dagger_+\times)
  +\hat{\sigma}_3(ik_0\tilde{\mathbf{t}}_+\times) \Big]\left(\begin{array}{c}\mathbf{E}_T\\ \eta_0\mathbf{H}_T\end{array}\right)
  &=-k_0\big[i\,\hat{\sigma}_1\, n_-+\hat{\sigma}_2\, n_+-i\,\hat{\sigma}_3\,\chi_z\big]\left(\begin{array}{c}\mathbf{E}_z\\ \eta_0\mathbf{H}_z\end{array}\right),
\end{align}
with $\tilde{\mathbf{g}}_\pm=\frac{1}{2}\left(\tilde{\mathbf{g}}_1\pm\tilde{\mathbf{g}}_2\right)$,
$\tilde{\mathbf{t}}_\pm=\frac{1}{2}\left(\tilde{\mathbf{t}}_1\pm\tilde{\mathbf{t}}_2\right)$,
and $n_\pm=\frac{1}{2}\left(\varepsilon_z\pm\mu_z\right)$, where we have already imposed the in-plane duality condition $\tensor{\varepsilon}_T=\tensor{\mu}_T$.
Substitution of $\hat{\sigma}_2(-\tensor{\varepsilon}_T^{-1})\cdot$Eq.~(\ref{maxeq1}) into $k_0\hat{\sigma}_2\cdot$Eq.~(\ref{maxeq2}) yields
\begin{equation*}
  \begin{split}
  &\Big[\hat{\sigma}_0\big(\partial_i+ik_0\tilde{t}_{-\,i}\big)+ik_0\left(\hat{\sigma}_1\tilde{g}_{-\,i}^*+i\hat{\sigma}_2\tilde{g}_{+\,i}^*-\hat{\sigma}_3\tilde{t}_{+\,i}\right)\Big]
  {\epsilon^{zij}(\varepsilon_T^{-1})_{jk}\epsilon^{klz}}
  \Big[\hat{\sigma}_0\big(\partial_l+ik_0\tilde{t}_{-\,l}\big)+ik_0\left(\hat{\sigma}_1\tilde{g}_{-\,l}-i\hat{\sigma}_2\tilde{g}_{+\,l}-\hat{\sigma}_3\tilde{t}_{+\,l}\right)\Big]
  \begin{pmatrix} Ez \\ \eta_0 H_z\end{pmatrix}\\
    =&\, k_0^{\,2}\big(\hat{\sigma}_0n_+ +\hat{\sigma}_1\chi_z+\hat{\sigma}_3n_-\big)\begin{pmatrix} Ez \\ \eta_0 H_z\end{pmatrix},\qquad (\ i,j,k,l\in\{x,y\}\ ).
\end{split}
\end{equation*}
On account of $\epsilon^{zij}(\varepsilon_T^{-1})_{jk}\epsilon^{klz}=\epsilon^{ij}(\varepsilon_T^{-1})_{ij}\epsilon^{kl}=-{\varepsilon_T^{ik}}/{\det(\tensor{\varepsilon}_T)}$ ($\epsilon^{ijk},\,\epsilon^{ij}$ are 3D and 2D Levi-Civita symbols), we obtain
\begin{equation}
  \Bigg\{\frac{1}{2}\left[\hat{\mathbf{p}}-\hat{\sigma}_0\mathbf{A}-\big(\hat{\mathcal{A}}-i\hat{\mathcal{A}}_I\big)\right]\cdot\tensor{m}^{-1}\cdot\left[\hat{\mathbf{p}}-\hat{\sigma}_0\mathbf{A}-\big(\hat{\mathcal{A}}+i\hat{\mathcal{A}}_I\big)\right]-k_0^{\,2}\big(\hat{\sigma}_0n_+ +\hat{\sigma}_1\chi_z+\hat{\sigma}_3n_-\big)\Bigg\}
  \left(\begin{array}{c}{E}_z\\ \eta_0{H}_z\end{array}\right)=0,
\end{equation}
where $\tensor{m}^{-1}=\frac{2}{\det(\tensor{\varepsilon}_T)}\tensor{\varepsilon}_T$
is the inverse of an effective anisotropic mass, 
$\mathbf{A}=-k_0\,\tilde{\mathbf{t}}_-=k_0\,\mathbf{t}_-\times\mathbf{e}_z$ denotes an emergent Abelian vector potential, and
\begin{equation}
  \hat{\mathcal{A}}_{(c)}=\hat{\mathcal{A}}+i\hat{\mathcal{A}}_I=\underbrace{k_0\left[\hat{\sigma}_1\mathrm{Re}(\mathbf{g}_-)\times\mathbf{e}_z+\hat{\sigma}_2\mathrm{Im}(\mathbf{g}_-)\times\mathbf{e}_z-\hat{\sigma}_3\mathbf{t}_+\times\mathbf{e}_z\right]
  }_{\displaystyle\hat{\mathcal{A}}=k_0\left(\begin{array}{cc}
  -\mathbf{t}_+\times\mathbf{e}_z & \mathbf{g}_-^\dagger\times\mathbf{e}_z\\
  \mathbf{g}_-\times\mathbf{e}_z & \mathbf{t}_+\times\mathbf{e}_z
\end{array}\right)}
  +i\underbrace{k_0\left(\hat{\sigma}_1\mathrm{Im}(\mathbf{g}_+)\times\mathbf{e}_z
  -\hat{\sigma}_2\mathrm{Re}(\mathbf{g}_+)\times\mathbf{e}_z\right)}_{\displaystyle\hat{\mathcal{A}}_I=k_0\left(\begin{array}{cc} 0 & i\mathbf{g}_+^\dagger\times\mathbf{e}_z\\
  -i\mathbf{g}_+\times\mathbf{e}_z & 0
  \end{array}\right)}
\end{equation}
can be regarded as a complex-valued non-Abelian vector potential with $\mathbf{g}_\pm=\frac{1}{2}(\mathbf{g}_1\pm\mathbf{g}_2^*)$ and $\mathbf{t}_\pm=\frac{1}{2}(\mathbf{t}_1\pm\mathbf{t}_2)$. 

The imaginary part, $\hat{\mathcal{A}}_I$, of non-Abelian potential can be further taken out from the ``kinetic energy part'':
\begin{equation}\label{pre-wave eq}
  \begin{split}
  \Bigg\{&\frac{1}{2}\left(\hat{\mathbf{p}}-\hat{\mathscr{A}}\right)\cdot\tensor{m}^{-1}\cdot\left(\hat{\mathbf{p}}-\hat{\mathscr{A}}\right)
    -k_0^{\,2}\big(\hat{\sigma}_0n_+ +\hat{\sigma}_1\chi_z+\hat{\sigma}_3n_-\big)\\[-5pt]
    & +\frac{1}{2}\left[\hat{\mathcal{A}}_I\cdot\tensor{m}^{-1}\cdot\left(\hat{\sigma}_0\nabla_T-i\hat{\mathscr{A}}\right)-\left(\hat{\sigma}_0\nabla_T-i\hat{\mathscr{A}}\right)\cdot\tensor{m}^{-1}\cdot\hat{\mathcal{A}}_I
  +\hat{\mathcal{A}}_I\cdot\tensor{m}^{-1}\cdot\hat{\mathcal{A}}_I\right]
\Bigg\}
  \left(\begin{array}{c}{E}_z\\ \eta_0{H}_z\end{array}\right)=0,
\end{split}
\end{equation}
where $\hat{\mathscr{A}}=\hat{\sigma}_0\mathsf{A}+\hat{\mathcal{A}}$ is the complete  $\mathrm{U}(2)$ real vector potential, and the second line associated with $\hat{\mathcal{A}}_I$ can  be decomposed as

\begin{equation}\label{AI terms 1}
\begin{aligned}
  &\ \hat{\mathcal{A}}_I\cdot\tensor{m}^{-1}\cdot\left(\hat{\sigma}_0\nabla_T-i\hat{\mathscr{A}}\right)-\left(\hat{\sigma}_0\nabla_T-i\hat{\mathscr{A}}\right)\cdot\tensor{m}^{-1}\cdot\hat{\mathcal{A}}_I+\hat{\mathcal{A}}_I\cdot\tensor{m}^{-1}\cdot\hat{\mathcal{A}}_I\\
  =&\,-\nabla_T\cdot\big(\tensor{m}^{-1}\cdot\hat{\mathcal{A}}_I\big)+i\big(-\hat{\mathcal{A}}_I\cdot\tensor{m}^{-1}\cdot\hat{\mathcal{A}}+\hat{\mathcal{A}}\cdot\tensor{m}^{-1}\cdot\hat{\mathcal{A}}_I\big)+\hat{\mathcal{A}}_I\cdot\tensor{m}^{-1}\cdot\hat{\mathcal{A}}_I\\
  =&\,-\nabla_T\cdot\big(\tensor{m}^{-1}\cdot\hat{\mathcal{A}}_I\big)+i\,\mathrm{Tr}\Big(\tensor{m}^{-1}\cdot\big[\hat{\mathcal{A}},\hat{\mathcal{A}}_I\big]\Big)+\hat{\mathcal{A}}_I\cdot\tensor{m}^{-1}\cdot\hat{\mathcal{A}}_I,
\end{aligned}
\end{equation}
each term of which can be further expressed with the material parameters: 
\begin{subequations}\label{Ai terms 2}
\begin{align}
\begin{aligned}
  -\nabla_T\cdot\big(\tensor{m}^{-1}\cdot\hat{\mathcal{A}}_I\big)\\[90pt]
  \end{aligned}
  & \begin{aligned}
  =&-k_0\nabla\cdot\Big\{\tensor{m}^{-1}\cdot \big[\big(\mathrm{Im}(\mathbf{g}_+)\hat{\sigma}_1-\mathrm{Re}(\mathbf{g}_+)\hat{\sigma}_2\big)\times\mathbf{e}_z\big]\Big\}\\
  =&\, -k_0\partial_i\Big\{{(\tensor{m}^{-1})^{ij}\epsilon_{jk}}\big(\mathrm{Im}(\mathbf{g}_+)\hat{\sigma}_1-\mathrm{Re}(\mathbf{g}_+)\hat{\sigma}_2\big)^k\Big\}\\
  =&\, -k_0\partial_i\Big\{\big(-2\epsilon^{il}(\tensor{\varepsilon}_T^{-1})_{lm}\epsilon^{mj}\epsilon_{jk}\big)\big(\mathrm{Im}(\mathbf{g}_+)\hat{\sigma}_1-\mathrm{Re}(\mathbf{g}_+)\hat{\sigma}_2\big)^k\Big\}\\
  =&-2k_0\epsilon^{il}\partial_i\Big\{(\tensor{\varepsilon}_T^{-1})_{lk}\big(\mathrm{Im}(\mathbf{g}_+)\hat{\sigma}_1-\mathrm{Re}(\mathbf{g}_+)\hat{\sigma}_2\big)^k\Big\}\\
  =&\,-2k_0\,\mathbf{e}_z\cdot\Big\{\nabla\times\big[\tensor{\varepsilon}_T^{-1}\cdot\big(\mathrm{Im}(\mathbf{g}_+)\hat{\sigma}_1-\mathrm{Re}(\mathbf{g}_+)\hat{\sigma}_2\big)\big]\Big\},
  \end{aligned}\\[8pt]
  \begin{aligned}
  i\,\mathrm{Tr}\Big(\tensor{m}^{-1}\cdot\big[\hat{\mathcal{A}},\hat{\mathcal{A}}_I\big]\Big)\\[40pt]
  \end{aligned}
  & \begin{aligned}  
  =&\,i\ (\tensor{m}^{-1})^{ij}\mathcal{A}^a_i\mathcal{A}^b_{I\,j}\big[\hat{\sigma}_a,\hat{\sigma}_b\big]
  =i\ (\tensor{m}^{-1})^{ij}\mathcal{A}^a_i\mathcal{A}^b_{I\,j} (2i\epsilon_{abc}\hat{\sigma}^c)\\
  =&\,4\big[(\mathcal{A}^a_i\epsilon^{il})(\tensor{\varepsilon}^{-1}_I)_{lm}(\epsilon^{mj}\mathcal{A}^b_{I\,j})\epsilon_{abc}\big]\hat{\sigma}^c\\
  =&\, 4k_0^{\,2}\,\Big[\big(\mathbf{t}_+\cdot\tensor{\varepsilon}_T^{-1}\cdot\mathrm{Re}(\mathbf{g}_+)\big)\hat{\sigma}_1
  +\big(\mathbf{t}_+\cdot\tensor{\varepsilon}_T^{-1}\cdot\mathrm{Im}(\mathbf{g}_+)\big)\hat{\sigma}_2
  +\mathrm{Re}\big(\mathbf{g}_-\cdot\tensor{\varepsilon}_T^{-1}\cdot\mathbf{g}_+^\dagger
  \big)\hat{\sigma}_3\Big]\\
  \end{aligned}\\[8pt]
  \begin{aligned}
  \hat{\mathcal{A}}_I\cdot\tensor{m}^{-1}\cdot\hat{\mathcal{A}}_I\\[50pt]
  \end{aligned}
  & \begin{aligned}
 =&\, k_0^{\,2}\,\big[\big(\mathrm{Im}(\mathbf{g}_+)\hat{\sigma}_1-\mathrm{Re}(\mathbf{g}_+)\hat{\sigma}_2\big)\times\mathbf{e}_z\big]\cdot\tensor{m}^{-1}\cdot\big[\big(\mathrm{Im}(\mathbf{g}_+)\hat{\sigma}_1-\mathrm{Re}(\mathbf{g}_+)\hat{\sigma}_2\big)\times\mathbf{e}_z\big]\\
  =&\, 2k_0^{\,2}\,\big(\mathrm{Im}(\mathbf{g}_+)\hat{\sigma}_1-\mathrm{Re}(\mathbf{g}_+)\hat{\sigma}_2\big)\cdot\tensor{\varepsilon}_T^{-1}\cdot\big(\mathrm{Im}(\mathbf{g}_+)\hat{\sigma}_1-\mathrm{Re}(\mathbf{g}_+)\hat{\sigma}_2\big)\\
    =&\, 2k_0^{\,2}\,\big(\mathrm{Im}(\mathbf{g}_+)\cdot\tensor{\varepsilon}_T^{-1}\cdot\mathrm{Im}(\mathbf{g}_+)  +  \mathrm{Re}(\mathbf{g}_+)\cdot\tensor{\varepsilon}_T^{-1}\cdot\mathrm{Re}(\mathbf{g}_+)\big)\\
  =&\,2k_0^{\,2}\,\big(\mathbf{g}_+\tensor{\varepsilon}_T^{-1}\cdot\mathbf{g}_+^\dagger\big)\hat{\sigma}_0
\end{aligned}
\end{align}
\end{subequations}
\newpage
After substituting Eqs.~(\ref{AI terms 1},\ref{Ai terms 2}) into Eq.~(\ref{pre-wave eq}), we arrive at the final form of the in-plane wave equation
\footnote{For a fixed frequency, Eq.~(\ref{wave equation2}) is analogous to the stationary Schrodinger equation: $(\hat{H}-E)|\psi\rangle = 0$. As the zero point of the $\mathrm{U}(1)$ scalar potential $\mathsf{V}_0$ is arbitrary, it can always be selected such that $E = 0$. The free choice of E will not affect the stationary dynamics, but it is meaningless to compare the ``eigen-energy'' for the effective Hamiltonian at different frequencies.}
\begin{equation}\label{wave equation2}\boxed{
  \hat{H}|\psi\rangle=\left[\frac{1}{2}\big(\hat{\mathbf{p}}-\hat{\mathscr{A}}\big)\cdot\tensor{m}^{-1}\cdot\big(\hat{\mathbf{p}}-\hat{\mathscr{A}}\big)-\hat{\mathscr{A}}_0\right]|\psi\rangle=\left[\frac{1}{2}\big(\hat{\mathbf{p}}-{\color{Green}\mathsf{A}}\hat{\sigma}_0-{\color{RubineRed}\hat{\mathcal{A}}}\big)\cdot\tensor{m}^{-1}\cdot\big(\hat{\mathbf{p}}-{\color{Green}\mathsf{A}}\hat{\sigma}_0-{\color{RubineRed}\hat{\mathcal{A}}}\big)-{\color{RubineRed}\hat{\mathcal{A}}_0}+{\color{Green}\mathsf{V}_0}\hat{\sigma}_0\right]
  |\psi\rangle=0.}
\end{equation}
Here, the effective Hamiltonian $\hat{H}$ is precisely like that of a non-relativistic spin-1/2 particle traveling in a $U(2)= \mathrm{SU}(2)\rtimes \mathrm{U}(1)$ background gauge potential $\{\hat{\mathscr{A}}_\mu\}$, where ``$\rtimes$'' denotes the semidirect product of two groups. The effective $\mathrm{U}(2)$ group potential always can be decomposed into two parts $\{\hat{\mathscr{A}}_\mu\}=\{{\color{Green}\mathsf{A}_\mu}\hat{\sigma}_0+{\color{RubineRed}\hat{\mathcal{A}}_\mu}\}$,  where $\{\color{Green}\mathsf{A}_\mu\}=\{-\mathsf{V}_0,\mathsf{A}\}$ ($\mu=0,1,2$) denotes an effective Abelian $\mathrm{U}(1)$ Maxwell-type gauge potential, while $\color{RubineRed}\{\hat{\mathcal{A}}_\mu\}=\{\hat{\mathcal{A}}_0,\hat{\mathcal{A}}\}$ denotes an effective non-Abelian $\mathrm{SU}(2)$ Yang-Mills gauge potential $(\hat{\mathcal{A}}_\mu=\mathcal{A}^a_\mu\hat{\sigma}_a$ are $\mathbf{su}(2)$-Lie-Algebra-valued, namely they are $2\times2$ traceless Hermitian matrices). Their expressions are listed in Supplementary Table~\ref{gauge potentials}.
And the results in the main text correspond to the reduced case with $\mathbf{t}_\pm=0$ and $\chi_z=0$.

As shown in Supplementary Table~\ref{gauge potentials}, different $\hat{\sigma}$-components of both vector and scalar potentials have distinct physical origins. For the $\mathrm{SU}(2)$ vector potential, $\mathcal{A}^1$ is caused by the the deviation of the principal axis of $\mathrm{Re}(\tensor{\varepsilon}),\,\mathrm{Re}(\tensor{\mu})$ from $z$-direction, which can be realized in reciprocal anisotropic materials~\cite{JensenLi2015PRL,liu2015polarization}. $\mathcal{A}^2$ component stems from the the imaginary part of the off-block-diagonal term $\mathrm{Im}(\mathbf{g}_-)$ in $\tensor{\varepsilon},\,\tensor{\mu}$, which can be excited by in-plane magnetic field in magneto-optic media~\cite{fang2013effective}. And $\mathcal{A}^3$ originates from the symmetric off-diagonal part of ME coupling $\mathbf{t}_+$, which, together with $\chi_z$, actually denotes an anisotropic Tellegen media with the ME tensor $\tilde{\chi}_{em}=\tilde{\chi}_{me}^\intercal=\mathrm{diag}\left(\frac{1}{2}(\chi_z+\sqrt{4|\mathbf{t}_+|^2+\chi_z^2}),\frac{1}{2}(\chi_z-\sqrt{4|\mathbf{t}_+|^2+\chi_z^2}),0\right)$ in the principal frame~\cite{jacobs2015photonic}. 

\begin{table}[t]
    \caption{\label{gauge potentials}The complete set of the components of the synthetic $U(2)=\mathrm{SU}(2)\rtimes \mathrm{U}(1)$ gauge potential corresponding to the spinor state $|\psi\rangle=(E_z,\,\eta_0H_z)^\intercal$, where the ``inner product'' $\llangle \cdot,\cdot\rrangle$ for two 2D vectors $\mathbf{a},\, \mathbf{b}$ is defined as
  $\llangle\mathbf{a},\mathbf{b}\rrangle=\mathbf{a}\cdot\tensor{\varepsilon}_T^{-1}\cdot\mathbf{b}^\dagger$. }\vspace{5pt}
  \begin{tabular}{Sc|Sc|Sc|Sc}\hline\hline
      &  &      expression  & physical origin     \\ \hline
  \multirow{12}*{$\color{RubineRed}\mathrm{SU}(2)$} & \multirow{5}*{\makecell[c]{vector\\[-0.75 ex] potential\\$\color{RubineRed}\hat{\mathcal{A}}=\mathcal{A}^a\hat{\sigma}_a$}} &
  ${\color{RubineRed}\mathcal{A}^1}=k_0\mathrm{Re}\left(\mathbf{g}_-\right)\times\mathbf{e}_z$ & \makecell[c]{reciprocal anisotropy of $\tensor{\varepsilon},\,\tensor{\mu}$~\cite{JensenLi2015PRL,liu2015polarization}\\ (deviation of principal axis from $z$-direction)}\\\cline{3-4}
   & & ${\color{RubineRed}\mathcal{A}^2}=k_0\mathrm{Im}\left(\mathbf{g}_-\right)\times\mathbf{e}_z$ & \makecell[c]{gyroelectric or gyromagnetic effects\\induced by in-plane magnetic field~\cite{fang2013effective}}\\\cline{3-4}
   & & ${\color{RubineRed}\mathcal{A}^3}=-k_0\mathbf{t}_+\times\mathbf{e}_z$ & \makecell[c]{real \& symmetric part of ME tensors~\cite{jacobs2015photonic}\\ (anisotropic Tellegen media)}   \\ \cline{2-4}
  &\multirow{4}*{\makecell[c]{scalar\\[-0.75 ex] potential\\$\color{RubineRed}\hat{\mathcal{A}}_0=\mathcal{A}_0^a\hat{\sigma}_a$}}   & \rule[-10pt]{0pt}{25pt} $\displaystyle
  {\color{RubineRed}\mathcal{A}_0^1}=k_0\mathbf{e}_{z}\cdot\left[\nabla\times\left(\tensor{\varepsilon}_T^{-1}\cdot\mathrm{Im}\left(\mathbf{g}_+\right)\right)\right]+{\color{black}k_0^{\,2}\left[\chi_z-2\big\llangle\mathbf{t}_+,\mathrm{Re}(\mathbf{g}_+)\big\rrangle\right]}$ &\multirow{3}*{\makecell[l]{$\bullet$ inhomogeneity of $\mathbf{g}_+$ and $\tensor{\varepsilon}_T$\\[0pt] $\bullet$ coupling between $\mathbf{g}_+$ and $\mathbf{t}_+$\\[0pt] $\bullet$ $\chi_z$ component of ME tensor}} \\\cline{3-3}
  & &\rule[-10pt]{0pt}{25pt} $\displaystyle{\color{RubineRed}\mathcal{A}_0^2}=-k_0\mathbf{e}_{z}\cdot\left[\nabla\times\left(\tensor{\varepsilon}_T^{-1}\cdot\mathrm{Re}\left(\mathbf{g}_+\right)\right)\right]+2k_0^{\,2}\llangle\mathbf{t}_+,\mathrm{Im}(\mathbf{g}_+)\rrangle$ &
  \\\cline{3-4}
  & &  $\displaystyle{\color{RubineRed}\mathcal{A}_0^3}=k_0^{\,2}\left[\frac{\varepsilon_z-\mu_z}{2} -2\mathrm{Re}\llangle\mathbf{g}_-,\mathbf{g}_+\rrangle\right]$
          &  \makecell[l]{$\bullet$ difference between $\varepsilon_z$ and $\mu_z$\\[0pt] $\bullet$ coupling between $\mathbf{g}_+$ and $\mathbf{g}_-$}  \\\hline
  \multirow{5}*{$\color{Green}\mathrm{U}(1)$} & \makecell[c]{vector\\[-0.75 ex] potential}& ${\color{Green}\mathsf{A}}=k_0\,\mathbf{t}_-\times\mathbf{e}_z$ & \makecell[c]{real \& antisymmetric part of ME tensors\\(moving media or static toroidal moment~\cite{cook1995fizeau,leonhardt1999optics,sawada2005optical})}\\\cline{2-4}
  & \makecell[c]{scalar\\[-0.75 ex] potential} &$ \rule[-10pt]{0pt}{25pt} \displaystyle\begin{aligned}{\color{Green}\mathsf{V}_0}=&k_0^{\,2}\left(\llangle\mathbf{g}_+,\mathbf{g}_+\rrangle-\frac{\varepsilon_z+\mu_z}{2}\right)
  \end{aligned}$ & $\mathbf{g}_+$, $\varepsilon_z$, $\mu_z$
   \\\hline\hline
  \end{tabular}
  \vspace{10pt}
\end{table}

The $\mathrm{U}(2)$ gauge potential can induce an emergent non-Abelian $\mathrm{U}(2)$ gauge field acting on the spinor wave function, in terms of the covariant derivative operator $\hat{\mathscr{D}}_i=\partial_\mu\hat{\sigma}_0-i\hat{\mathscr{A}}_\mu=\partial_\mu\hat{\sigma}_0-i{\color{RubineRed}\hat{\mathcal{A}}_\mu}-{\color{Green}\mathsf{A}_\mu}\hat{\sigma}_0$:
\begin{equation}\label{U(2) nonabelian field}
  \hat{\mathscr{F}}_{\mu\nu}=i[\hat{\mathscr{D}}_\mu,\hat{\mathscr{D}}_\nu]=\partial_\mu\hat{\mathscr{A}}_\nu-\partial_\nu\hat{\mathscr{A}}_\mu-i[\hat{\mathscr{A}}_\mu,\hat{\mathscr{A}}_\nu]
=\underbrace{\left(\partial_\mu{\color{Green}{\mathsf{A}}_\nu}-\partial_\nu{\color{Green}\mathsf{A}_\mu}\right)}_{\displaystyle \mathrm{U}(1)\text{ field: }\mathsf{F}_{\mu\nu}}\hat{\sigma}_0\ \ +\ \underbrace{\partial_\mu{\color{RubineRed}\hat{\mathcal{A}}_\nu}-\partial_\nu{\color{RubineRed}\hat{\mathcal{A}}_\mu}-i[{\color{RubineRed}\hat{\mathcal{A}}_\mu},{\color{RubineRed}\hat{\mathcal{A}}_\nu}]}_{\displaystyle \mathrm{SU}(2)\text{ field: }\hat{\mathcal{F}}_{\mu\nu}=i[\hat{\mathcal{D}}_\mu,\hat{\mathcal{D}}_\nu]},
\end{equation}
where $\hat{\mathcal{D}}_\mu=\partial_\mu\hat{\sigma}_0-i{\color{RubineRed}\hat{\mathcal{A}}_\mu}$. As shown, the $\mathrm{U}(2)$ gauge field also can be decomposed into a $\mathrm{U}(1)$ part and a $\mathrm{SU}(2)$ part.
Since the effective $\mathrm{SU}(2)$ gauge field is \(\mathbf{su}(2)\)-Lie-algebra-valued, and also can be expanded by Pauli matrices, $\hat{\mathcal{F}}_{\mu\nu}=\mathcal{F}^a_{\mu\nu}\hat{\sigma}_a$, the $\mathrm{SU}(2)$ part of Eq.~(\ref{U(2) nonabelian field}) can be rewritten in a component form:
\begin{equation}
\mathcal{F}_{\mu\nu}^a=\partial_\mu{\mathcal{A}}_\nu^a-\partial_\nu{\mathcal{A}}_\mu^a+2{\epsilon^a}_{bc}\mathcal{A}_\mu^b \mathcal{A}_\nu^c,
\end{equation}
where \(2\epsilon_{abc}\) is the structure constant of
\(\mathbf{su}(2)\) Lie algebra,
\([\hat{\sigma}_a,\hat{\sigma}_b]=2i\epsilon_{abc}\hat{\sigma}^c\).

Analogous to the real EM field, the effective $\mathrm{U}(1)$ gauge field $\mathsf{F}_{\mu\nu}$ can be alternatively treated as a pair of effective Abelian magnetic and electric fields $\mathsf{B},\,\mathsf{E}$ (not to be confused with the real EM fields $\mathbf{B},\,\mathbf{E}$):
\begin{gather}
  \mathsf{B}=\nabla\times{\mathsf{A}}=-\mathbf{e}_z\,(\nabla\cdot\mathbf{t}_-),\quad
  {\mathsf{E}}=-\nabla{\mathsf{V}}_0,
\end{gather}
Similarly, if we express the tensor of the $\mathrm{SU}(2)$ non-Abelian field as
\begin{equation}
    \big(\hat{\mathcal{F}}_{\mu\nu}\big)=
    \left(\begin{array}{ccc}
      0 & -\hat{\mathcal{E}}_x & -\hat{\mathcal{E}}_y \\
      \hat{\mathcal{E}}_x & 0 & \hat{\mathcal{B}}_z \\
      \hat{\mathcal{E}}_y & -\hat{\mathcal{B}}_z & 0
    \end{array}\right),
\end{equation}
different components of the $\mathrm{SU}(2)$ gauge field also can be classified into an effective non-Abelian magnetic field $\hat{\mathcal{B}}$ and an effective non-Abelian electric field $\hat{\mathcal{E}}$ separately:
\begin{equation}
  \begin{aligned}
  \hat{\mathcal{B}}=\frac{1}{2}\epsilon^{ij}\hat{\mathcal{F}}_{ij}{\mathbf{e}}_z=\nabla\times\hat{\mathcal{A}}-i\hat{\mathcal{A}}\times\hat{\mathcal{A}},\qquad\quad
  \hat{\mathcal{E}}=-\hat{\mathcal{F}}_{0i}{\mathbf{e}}_i=\nabla\hat{\mathcal{A}}_0+i[\hat{\mathcal{A}}_0,\hat{\mathcal{A}}].
  \end{aligned}
\end{equation}
As the system is $z$-invariant, the effective magnetic fields $\mathsf{B}$, $\hat{\mathcal{B}}$ are along $z$ direction, while the effective electric fields $\mathsf{E}$, $\hat{\mathcal{E}}$ always lie in $xy$-plane.
We will show that these
effective non-Abelian $\mathrm{SU}(2)$ magnetic and electric fields can affect the centroid motion of the spinor wave function $|\psi\rangle$ in a similar way as real EM fields acting on charged particles.
According to Supplementary Table~\ref{gauge potentials}, the $\mathrm{SU}(2)$ vector potential $\hat{\mathcal{A}}$ only depends on $\mathbf{g}_-$ and $\mathbf{t}_+$. Hence, the $\mathrm{SU}(2)$ magnetic field has a relatively simple expression depending on $\mathbf{g}_-$ and $\mathbf{t}_+$:
\begin{equation}
\hat{\mathcal{B}}=k_0^2\Big[2\big(\mathbf{t}_{+}\times\mathrm{Im}(\mathbf{g}_-)\big)\hat{\sigma}_1
-2\big(\mathbf{t}_{+}\times\mathrm{Re}(\mathbf{g}_-)\big)\hat{\sigma}_2+i\left(\mathbf{g}_-\times{\mathbf{g}_{-}^*}\right)\hat{\sigma}_3\Big]-k_0\nabla\cdot\left[\mathrm{Re}\left(\mathbf{g}_-\right)\hat{\sigma}_1+\mathrm{Im}\left(\mathbf{g}_-\right)\hat{\sigma}_2-\mathbf{t}_+\hat{\sigma}_3\right]\mathbf{e}_{z}.
\end{equation}
By contrast, the $\mathrm{SU}(2)$ electric field $\hat{\mathcal{E}}$, is determined by all the following components: $\mathbf{g}_\pm$, $\mathbf{t}_+$, $\varepsilon_z,\,\mu_z$, and $\tensor{\varepsilon}_T$.

\newpage
\section{Gauge transformation between synthetic gauge field systems}
 Since the choice of the gauge of the wave function can be arbitrary, we can consider a gauge transformation of the wave function $|\psi'\rangle=\hat{U}(\mathbf{r})|\psi\rangle$, where $\hat{U}(\mathbf{r})=\hat{U}(x,y)$ is now generalized to be a $xy$-dependent $\mathrm{U}(2)$ matrix. The wave equation is covariant under the gauge transformation: $\hat{U}\hat{H}|\psi\rangle=(\hat{U}\hat{H}\hat{U}^\dagger)|\psi'\rangle=0$ with the transformed Hamiltonian
\begin{equation}
  \hat{H}'=\hat{U}\hat{H}\hat{U}^\dagger=-\frac{1}{2}\left(\hat{U}\hat{\mathscr{D}}_i\hat{U}^\dagger\right)\big(m^{-1}\big)^{ij}\left(\hat{U}\hat{\mathscr{D}}_j\hat{U}^\dagger\right)
  -\left(\hat{U}\hat{\mathscr{A}}_0\hat{U}^\dagger\right)=\frac{1}{2}\hat{\mathscr{D}}'_i\big(m^{-1}\big)^{ij}\hat{\mathscr{D}}'_j-\hat{\mathscr{A}}'_0,
\end{equation}
and the transformed covariant derivative satisfies
\begin{equation}
  \hat{\mathscr{D}}'_i=\hat{U}\hat{\mathscr{D}}_i\hat{U}^\dagger=\partial_i\hat{\sigma}_0-i\left(\hat{U}\hat{\mathscr{A}}_i\hat{U}^\dagger+i\hat{U}\partial_i\hat{U}^\dagger\right)=\partial_i\hat{\sigma}_0-i\hat{\mathscr{A}}'_i.
\end{equation}
As a result, we obtain the rule of gauge transformation for the gauge potential:
\begin{equation}\label{U2 gauge transfrom}
  \begin{aligned}  
  \hat{\mathscr{A}}'_i=&\hat{U}\hat{\mathscr{A}}_i\hat{U}^\dagger+i\hat{U}\partial_i\hat{U}^\dagger\\
  \hat{\mathscr{A}}'_0=&\hat{U}\hat{\mathscr{A}}_0\hat{U}^\dagger=\hat{U}\hat{\mathscr{A}}_0\hat{U}^\dagger+\underbrace{i\hat{U}\partial_0\hat{U}^\dagger}_{=0} 
  \end{aligned}  
   \raisebox{10pt}{ $\Bigg\}\quad 
  \hat{\mathscr{A}}'_\mu=\hat{U}\hat{\mathscr{A}}_\mu\hat{U}^\dagger+i\hat{U}\partial_\mu\hat{U}^\dagger.$} \raisetag{-35pt} 
\end{equation}
In addition, since the $2\times2$ unitary matrix can be expressed generically as $\hat{U}=e^{i\varphi}\hat{\mathcal{U}}$ with a $\mathrm{U}(1)$ part $e^{i\varphi}$ and a $\mathrm{SU}(2)$ part $\hat{\mathcal{U}}$ satisfying $\det(\hat{\mathcal{U}})=1$, we have $i\hat{U}\partial_\mu\hat{U}^\dagger=-\partial_\mu\varphi\hat{\sigma}_0+i\hat{\mathcal{U}}\partial_\mu\hat{\mathcal{U}}^\dagger$. And $\mathrm{Tr}\left(i\hat{\mathcal{U}}\partial_\mu\hat{\mathcal{U}}^\dagger\right)=i\det(\hat{\mathcal{U}})\partial_\mu \det(\hat{\mathcal{U}}^\dagger)\equiv0$, hence $i\hat{\mathcal{U}}\partial_\mu\hat{\mathcal{U}}^\dagger\in\mathbf{su}(2)$, the gauge transformation of $\mathrm{U}(1)$ and $\mathrm{SU}(2)$ parts of the gauge potential obeys
\begin{equation}\label{U2 gauge transfrom 2}
  {\color{Green}\mathsf{A}'_\mu=\mathsf{A}_\mu-\partial_\mu\varphi},\qquad\quad
  {\color{RubineRed}\hat{\mathcal{A}}'_\mu=\hat{\mathcal{U}}\hat{\mathcal{A}}_\mu\hat{\mathcal{U}}^\dagger+i\hat{\mathcal{U}}\partial_\mu\hat{\mathcal{U}}^\dagger}.
\end{equation}
Meanwhile, the gauge transformation of gauge fields reads
\begin{equation}
  \hat{\mathscr{F}}'_{\mu\nu}=i[\hat{\mathscr{D}}'_\mu,\hat{\mathscr{D}}'_\nu]=\hat{U}\hat{\mathscr{F}}_{\mu\nu}\hat{U}^\dagger\quad
  \left\{
  \begin{aligned}
    {\color{Green}\mathrm{U}(1):}& & \hat{\mathsf{F}}'_{\mu\nu}=&\,\hat{\mathsf{F}}_{\mu\nu} &\Leftrightarrow & &
     \hat{\mathsf{B}}'=&\,\hat{\mathsf{B}},& & \hat{\mathsf{E}}'=\,\hat{\mathsf{E}} \\
    {\color{RubineRed}\mathrm{SU}(2):}& &  \hat{\mathcal{F}}'_{\mu\nu}=&\,\hat{\mathcal{U}}\hat{\mathcal{F}}_{\mu\nu}\hat{\mathcal{U}}^\dagger &\Leftrightarrow & &
    \hat{\mathcal{B}}'=&\,\hat{\mathcal{U}}\hat{\mathcal{B}}\hat{\mathcal{U}}^\dagger,& & \hat{\mathcal{E}}'=\,\hat{\mathcal{U}}\hat{\mathcal{E}}\hat{\mathcal{U}}^\dagger 
  \end{aligned}
  \right.
\end{equation}
Unlike the $\mathrm{U}(1)$ field which is independent of the gauge choice, the $\mathrm{SU}(2)$ part of the gauge field is gauge-dependent.

It should be noted that the meaning of gauge covariance in the synthetic gauge system is subtly different from that in real gauge systems. For real gauge fields, the choice of gauge refers to the process of regulating the excess unphysical degrees of freedom, hence the gauge transformation leaves all observables unaffected, while all of the gauge-dependent quantities cannot be directly observed. However, in the synthetic gauge system, different gauges correspond to different materials, and the gauge dependence of the effective wave function $|\psi\rangle=(E_z,\eta_0 H_z)^\intercal$ can be directly detected, since the EM fields, including their phases, are measurable. According to the discussions in the Methods of the main text, the gauge transformation acting on the wave function
$|\psi'\rangle=\hat{U}(\mathbf{r})|\psi\rangle$ gives rise to the transformation of the total EM fields as
\begin{equation}\label{EM transform}
  \Psi'=  \tilde{U}(\mathbf{r})\Psi=
  \begin{pmatrix}\hat{\sigma}_2\hat{U}(\mathbf{r})\hat{\sigma}_2 & 0 \\ 0 & \hat{U}(\mathbf{r}) \end{pmatrix}\Psi,
\end{equation}
with $\Psi=(\mathbf{E}_T, \eta_0\mathbf{H}_T,E_z, \eta_0 H_z)^\intercal$.
Considering the 2D Maxwell's equations expressed in terms of $\Psi$:
\begin{equation}
\underbrace{\left(\begin{array}{c|c}
0 & i\hat{\sigma}_2 (i\nabla_T\times)\mathbf{e}_z\\ \hline
i\hat{\sigma}_2\, \mathbf{e}_z\cdot(i\nabla_T\times) & 0
\end{array}\right)}_{\displaystyle\mathcal{M}}
\Psi
=k_0
\underbrace{\left(\begin{array}{cc|cc}
\tensor{\varepsilon}_T & 0 & \mathbf{g}_1 & \mathbf{t}_1/c \\
0 & \tensor{\varepsilon}_T & \mathbf{t}_2/c & \mathbf{g}_2 \\\hline
\mathbf{g}_1^\dagger & \mathbf{t}^\intercal_2/c & \varepsilon_z & \chi_z \\
\mathbf{t}_1^\intercal/c &\mathbf{g}_2^\dagger & \chi_z & \mu_z
\end{array}\right)}_{\displaystyle{\mathcal{N}}}
\Psi,
\end{equation}
the gauge transformation of  Maxwell's equations shows that
\begin{equation}
  (\tilde{U}\mathcal{M}\tilde{U}^\dagger)\Psi'=(\mathcal{M}+\Delta\mathcal{M})\Psi'=k_0(\tilde{U}\mathcal{N}\tilde{U}^\dagger)\Psi',
\end{equation}
where
\begin{equation}
  \Delta\mathcal{M}
=\left(\begin{array}{c|c}
0 & i\hat{\sigma}_2 (i\hat{U}\nabla_T \hat{U}^\dagger )\times\mathbf{e}_z \\ \hline
-i (i\hat{U}\nabla_T \hat{U}^\dagger )\times\mathbf{e}_z \hat{\sigma}_2 & 0
\end{array}\right).
\end{equation}
Then we find that $\Psi'$ satisfies the Maxwell's equations $\mathcal{M}\Psi'=k_0\,\mathcal{N}'\Psi'$ in the transformed material: 
\begin{equation}\label{gauge transform of media}
\mathcal{N}'=\tilde{U}\mathcal{N}\tilde{U}^\dagger-\Delta\mathcal{M}/k_0\quad
\left\{\begin{aligned}
  \tensor{\varepsilon}'_T\quad & =\quad \tensor{\varepsilon}_T\\[5pt]
  \begin{pmatrix}
    \varepsilon'_z & \chi'_z\\
    \chi'_z & \mu'_z
  \end{pmatrix} & =\hat{U}\begin{pmatrix}
    \varepsilon_z & \chi_z\\
    \chi_z & \mu_z
  \end{pmatrix}\hat{U}^\dagger\\[5pt]
  \begin{pmatrix}
     \mathbf{t}'_2/c & \mathbf{g}'_2\\
     -\mathbf{g}'_1 & -\mathbf{t}'_1/c
  \end{pmatrix}& =  \hat{U}
  \begin{pmatrix}
     \mathbf{t}_2/c & \mathbf{g}_2\\
     -\mathbf{g}_1 & -\mathbf{t}_1/c
  \end{pmatrix}\hat{U}^\dagger+
  \frac{1}{k_0}(i\hat{U}\nabla_T \hat{U}^\dagger )\times\mathbf{e}_z,
\end{aligned}
\right.
\end{equation}
moreover, the synthetic $\mathrm{U}(2)$ gauge potential in this transformed material is consistent with the result given by the gauge transformation Eqs.~(\ref{U2 gauge transfrom},\ref{U2 gauge transfrom 2}). 

 To be precise, the gauge covariance of the synthetic gauge system means that if the material parameters in two systems can be related according to the gauge transformation Eq.~(\ref{gauge transform of media}), the solutions of EM fields in the two systems have a one-to-one correspondence in terms of the mapping Eq.~(\ref{EM transform}), and the forms of synthetic gauge potentials and fields in the two systems are transformed in exactly the same way as the real non-Abelian gauge potentials and fields expressed in different gauges. Moreover, all of the gauge-independent quantities in real gauge systems, such as the effective probability density $|\psi(\mathbf{r})|^2$, the quadratic forms of gauge fields $\mathrm{Tr}\left(\hat{\mathcal{F}}_{\alpha\beta}\hat{\mathcal{F}}_{\mu\nu}\right)$, and the Wilson loop, are also identical in the two synthetic gauge systems which are related by the material and field transformations.  Especially, the time-averaged in-plane Poynting vectors for the field $\Psi'$ in the transformed medium is also identical with that in the original medium:
\begin{equation}
   \begin{split}
   \bar{\mathbf{S}}'_T &=\frac{1}{2}\mathrm{Re}\left[(\mathbf{E}'^*_z,\,\mathbf{H}'^*_z)\Big(i\hat{\sigma}_2\otimes(\tensor{I}\times\tensor{I})\Big)\begin{pmatrix}
    \mathbf{E}'_T\\ \mathbf{H}'_T
  \end{pmatrix}\right]\\
  &=\frac{1}{2}\mathrm{Re}\left[(\mathbf{E}^*_z,\,\mathbf{H}^*_z) (\hat{U}^\dagger\otimes\tensor{I})\Big(i\hat{\sigma}_2\otimes(\tensor{I}\times\tensor{I})\Big)\Big((\hat{\sigma}_2\hat{U}\hat{\sigma}_2)\otimes\tensor{I}\Big)\begin{pmatrix}
    \mathbf{E}_T\\ \mathbf{H}_T
  \end{pmatrix}\right]\\
  &=\frac{1}{2}\mathrm{Re}\left[(\mathbf{E}^*_z,\,\mathbf{H}^*_z)\Big(i\hat{\sigma}_2\otimes(\tensor{I}\times\tensor{I})\Big)\begin{pmatrix}
    \mathbf{E}_T\\ \mathbf{H}_T
  \end{pmatrix}\right]
  =\bar{\mathbf{S}}_T.
\end{split}
\end{equation}
This result indicates that we can design material parameters using non-Abelian gauge transformation to manipulate  the spin (polarization) of light without changing its flow.


\newpage

\section{Analogy with \emph{Zitterbewegung} effect arising from Dirac cone}
The original ZB effect for Dirac electrons stems from the superposition of positive and negative energy states~\cite{zawadzki2011zitterbewegung}. And a majority of previous photonic realizations of ZB are based on a Dirac cone dispersion. By contrast, the appearance of ZB effect in our system does not rely on a Dirac cone dispersion.  In this section, we show that, in some limiting situations, the ZB effect in our system can also be understood from a Dirac cone structure.

Consider the transition from a reduced Abelian medium
($\hat{\mathcal{A}}=\mathcal{A}^2_y\,\mathbf{e}_y\hat{\sigma}_2,=k_0\tilde{\mathcal{A}^2_y}\,\mathbf{e}_y\hat{\sigma}_2,\ \hat{\mathcal{A}}_0=0$ ) to a
non-Abelian medium
( $\hat{\mathcal{A}}=k_0\tilde{\mathcal{A}^2_y}\,\mathbf{e}_y\hat{\sigma}_2+k_0\tilde{\mathcal{A}^1_x}\,\mathbf{e}_x\hat{\sigma}_1,\ \hat{\mathcal{A}}_0=0$, i.e. example I in the main text):
\begin{alignat}{3}
&\quad\ \text{Abelian medium} &\qquad& &\qquad&\quad\ \text{non-Abelian medium}\nonumber\\[2pt] \label{transition}
& \tensor{\varepsilon}/\varepsilon_0=\tensor{\mu}/\mu_0=\left(\begin{array}{cc|c} \varepsilon_T &  0 & -i\tilde{\mathcal{A}^2_y}\\
0 & \varepsilon_T  & 0 \\\hline
i\tilde{\mathcal{A}^2_y} & 0 & \varepsilon_z
\end{array}\right),
&& \Rightarrow &&
\left\{\begin{aligned}
\tensor{\varepsilon}/\varepsilon_0=&\left(\begin{array}{cc|c} \varepsilon_T &  0 & -i\tilde{\mathcal{A}^2_y}\\
0 & \varepsilon_T  & {\color{blue}\tilde{\mathcal{A}^1_x}} \\\hline
i\tilde{\mathcal{A}^2_y} & {\color{blue}\tilde{\mathcal{A}^1_x}} & \varepsilon_z
\end{array}\right),\\
\tensor{\mu}/\mu_0=&\left(\begin{array}{cc|c} \varepsilon_T &  0 & -i\tilde{\mathcal{A}^2_y}\\
0 & \varepsilon_T  & {\color{blue}-\tilde{\mathcal{A}^1_x}} \\\hline
i\tilde{\mathcal{A}^2_y} & {\color{blue}-\tilde{\mathcal{A}^1_x}} & \varepsilon_z
\end{array}\right).
\end{aligned}\right.
\end{alignat}

\begin{figure}[b]\vspace{0pt}
\includegraphics[width=0.72\columnwidth,clip]{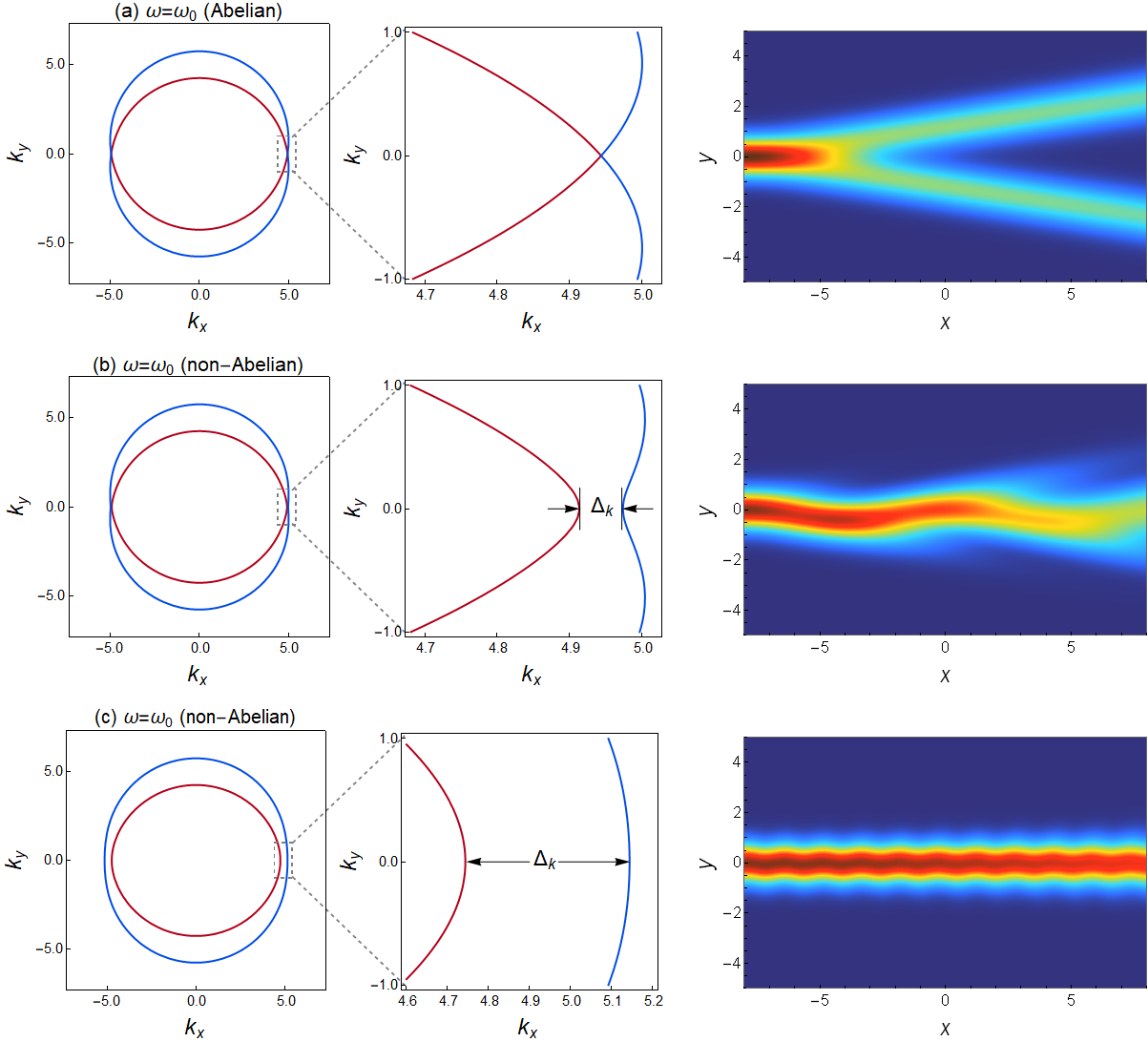}
\caption{ Transition from a Abelian medium to a non-Abelian medium corresponding to Eq.~(\ref{transition}). (a) Reduced Abelian gauge field medium with $\mathcal{A}_y^2=0.15 k_0$; (b) Non-Abelian perturbed medium with $\mathcal{A}_y^2=0.15 k_0$ and a perturbation $\mathcal{A}_x^1=0.006 k_0$; (c) Non-Abelian medium with $\mathcal{A}_y^2=0.15 k_0$ and $\mathcal{A}_x^1=0.04 k_0$. Other parameters of the three media are identical: $\varepsilon_T=\varepsilon_z=1$.
\label{dirac cone}}\vspace{-25pt}
\end{figure}

In the Abelian medium, the sole nonzero component of the $\mathrm{SU}(2)$ gauge potential is $\mathcal{A}^2=k_0\tilde{\mathcal{A}}^2_y\mathbf{e}_y$, the effective Hamiltonian has $\mathrm{U}(1)$ spin rotation symmetry about the $\hat{\sigma}_2$-axis, thus the two eigenmodes are $\psi_\pm=(1,\pm i)^\intercal$ (i.e. $E_z\pm i\eta_o H_z$) whose pseudo-spins are uniformly polarized along the $\hat{\sigma}_2$-axis for an arbitrary direction of the wave vector. The two branches of isofrequency contours are degenerate at $\mathbf{k}=k\,\mathbf{e}_x=\sqrt{\varepsilon_T\varepsilon_z-(\tilde{\mathcal{A}}^2_y)^2}\,k_0$.  In the vicinity of the degenerate
point, the two isofrequency contours intersect linearly, and thus 
can be regarded as a 1D Dirac cone  as shown in Supplementary Figure~\ref{dirac cone}(a),  provided that the \(x\) axis is treated as the pseudo-time dimension.

When the component $\mathcal{A}^1=\mathcal{A}^1_x\mathbf{e}_x$  emerges, the medium turns into non-Abelian characterized by the nonzero non-Abelian magnetic field $\hat{\mathcal{B}}=2\mathcal{A}^1\times\mathcal{A}^2\hat{\sigma}_3=\mathcal{A}^1_x\mathcal{A}^2_y\mathbf{e}_z\hat{\sigma}_3$. Meanwhile, the previous intersected isofrequency contours become fully gapped. As long as $\mathcal{A}^1_x\ll \mathcal{A}^2_y$, the isofrequency contours nearby $\mathbf{k}=k\,\mathbf{e}_x$  can be regarded as a 1D gapped Dirac cone as shown in Supplementary Figure~\ref{dirac cone}(b). And the width of the gap, i.e. the beat wave number corresponds to twice the effective mass of the 1D Dirac electron $M_e$: 
\begin{equation}
   \Delta_k=\big|k_+-k_-\big|=2|\mathcal{A}_x^1|=2M_e.
\end{equation}
According to the standard ZB effect of realistic electrons~\cite{zawadzki2011zitterbewegung}, the superposition of two eigenmodes at $\mathbf{k}$ (positive and negative
energy states) will give rise to a trembling motion whose frequency is exactly determined by effective Dirac mass $2M_e=\Delta_k$. This intuitive explanation is consistent with our analytical result in Eq.~(15) of the main text.


\begin{figure}[t!]
 \centering
 \includegraphics[width=0.8\columnwidth,clip]{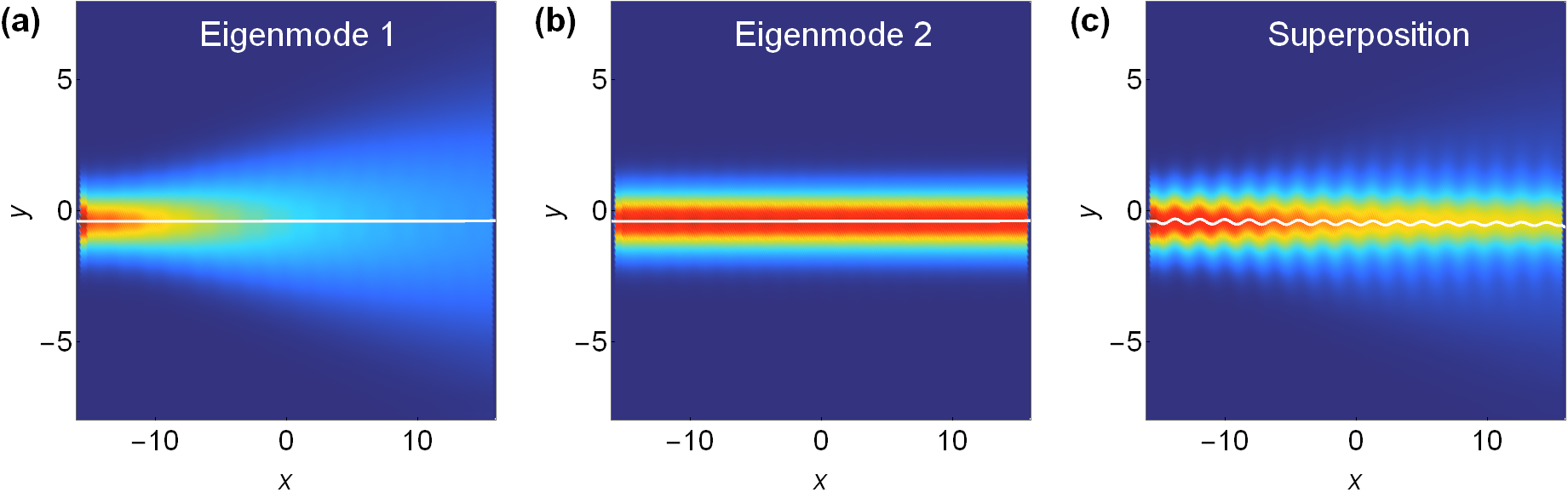}
 \caption{ Spatial ZB of monochromatic beam arising from eigenmode superposition. (a,b) Beams with eigen-polarizations travel in straight lines. (c) A beam superposed by the two eigenfields shown in (a,b) experiences a wavy centroid trajectory (white curve). The parameters of the background non-Abelian medium are given by  $\varepsilon_T=1$, $\varepsilon_z=1.2$, $\mu_z=0.8$, $\mathbf{g}_1=-\mathbf{g}_2=0.3\mathbf{e}_x$. \label{superpostion}  }\vspace{-13pt}
\end{figure}

For broader scenarios, the non-Abelian media cannot be treated as the
perturbation of an Abelian system, and the dispersion does not
 correspond to a gapped Dirac cone either, as illustrated in Supplementary Figure~\ref{dirac cone}(c). However, the ZB effect, arising from the beating of two eigenmodes, still can be observed and is quantitatively in agreement with our analytical result, as we have been demonstrated with examples in the main text. In Supplementary Figure~\ref{superpostion}, we visualize the ZB effect as the consequence of beating between the two eigenmodes. In Supplementary Figure~\ref{superpostion}(a,b), two beams incident from $x$-direction are polarized along the two eigen-polarizations on the $k_y=0$ cross section respectively, and they propagate in straight lines. However, when we superpose the fields of the two eigen-polarized beams, the obtained beam undergoes an obvious trembling motion as shown in Supplementary Figure~\ref{superpostion}(c).

\subsection*{Comparison with $\mathbf{k}$-space Lorentz force induced by Berry curvature}

It is worthwhile to mention that an optical wave packet propagating in \textbf{inhomogeneous} and weakly anisotropic media will also experience a virtual non-Abelian Lorentz force $\big(\hat{\Omega}_{\mathbf{k}}\times\frac{d}{d\tau}\hat{\mathbf{k}}\big)$ in the momentum $\mathbf{k}$-space induced by the non-Abelian Berry curvature $\hat{\Omega}({\mathbf{k}})$~\cite{onoda2006geometrical,bliokh2007non,bliokh2008geometrodynamics}. However, the existence of $\mathbf{k}$-space Lorentz force relies on the inhomogeneity of the media, while our  non-Abelian Lorentz force in the real space can be generated purely from the anisotropy of the media, they are accordingly distinct from each other. In homogeneous media, the wave vector of a wave packet is conserved, i.e.  $\frac{d}{d\tau}\hat{\mathbf{k}}=0$, hence the $\mathbf{k}$-space Lorentz force vanishes $\big(\hat{\Omega}_{\mathbf{k}}\times\frac{d}{d\tau}\hat{\mathbf{k}}\big)\equiv 0$; whereas the $\mathrm{SU}(2)$ non-Abelian gauge fields in real space can still exist $\hat{\mathcal{F}}_{\mu\nu}=-i[\hat{\mathcal{A}}_{\mu},\hat{\mathcal{A}}_{\nu}]\neq 0$ and so does the real-space Lorentz force, $\frac{1}{2}\left(\hat{\mathbf{v}}\times{\color{Black}\hat{\mathcal{B}}}-{\color{Black}\hat{\mathcal{B}}}\times\hat{\mathbf{v}}\right)+{\color{Black}\hat{\mathcal{E}}}$. 

Another difference between the real-space and the $\mathbf{k}$-space schemes is that the real-space scheme is applicable for any monochromatic full-wave phenomena, as the wave equation~(2) shown in the main text duplicates exactly the stationary Schr\"{o}dinger equation for spin-1/2 particles interacting with background $\mathrm{SU}(2)$ gauge fields, while the $\mathbf{k}$-space scheme is limited to geometric optics approximation~\cite{onoda2006geometrical,bliokh2007non,bliokh2008geometrodynamics}. 

%

\section{ZB effect of 3D beams with finite width in $z$ direction }

In our theory, the EM fields are supposed to be translationally invariant along the $z$-axis. In realistic systems, optical beams should have a finite width in the $z$-direction. In this section, we test whether our theory can be applied to describing real 3D beams. 

If the field of a 3D beam has a Gaussian-like distribution in the $z$-direction: $\mathbf{E}\sim\exp\big(-(z-z_0)^2/w_z^{\,2}\big)$, its derivative $\partial_z\mathbf{E}\sim-\frac{2}{w_z^{\,2}}(z-z_0)\exp\big(-(z-z_0)^2/w_z^{\,2}\big)$ tends to zero in the middle section $z=z_0$. Thus we expect that the middle section $z=z_0$ would be a feasible 2D domain where the 2D theory of non-Abelian gauge field optics proves a good description. To substantiate our analysis, we examined the two types of non-Abelian media studied in the main text.

As for the gyrotropic material with $\tensor{\varepsilon}_T=\tensor{\mu}_T$, $\varepsilon_z=\mu_z$, $\mathbf{g}_1=-\mathbf{g}_2^*$, it can exert a synthetic non-Abelian magnetic field on the light beams propagating in $x$-direction. As shown in Supplementary Figure~\ref{3D ZB gyrotropic}, we numerically simulated the propagation a 3D Gaussian type beam in this medium.
According to Supplementary Figure~\ref{3D ZB gyrotropic}(c), the beam propagates strictly along the horizontal plane since the isofrequency contours are symmetric with respect to the $k_z=0$ plane (see Supplementary Figure~\ref{3D ZB gyrotropic}(d)).  Therefore, the vertical center of the beam is fixed at $z=0$ and the ansatz $\partial_z\mathbf{E}(z=0)\approx 0$ is always satisfied. Consequently, the 2D centroid trajectory extracted from the fields falling on the middle plane $z=0$ perfectly matches the analytical result predicted by the 2D theory as shown in Supplementary Figure~\ref{3D ZB gyrotropic}(e).

For the biaxial dielectric material with misaligned principal axes, it performs as a background non-Abelian electric field acting on beams propagating in the $xy$-plane. We simulated three 3D beams with different waists in the $z$-direction, namely $w_z=3\lambda_0,\ 6\lambda_0,\ 8\lambda_0$, traveling in the biaxial medium. The simulation results reveal that the beams will split into two branches in the $z$-direction after propagating a distance as shown in Supplementary Figure~\ref{3D ZB biaxial}(c,f,i). The splitting effect originates from the unparallelism of the group velocities of the two eigenmodes. According to the isofrequency contours of the medium in the $k_y=0$ plane in Supplementary Figure~\ref{3D ZB biaxial}(j), the group velocity of eigenmode $|\mathord{\uparrow}\rangle$ is always along the $x$-axis, while the group velocity of eigenmode $|\mathord{\downarrow}\rangle$ has a nonzero $z$-component. Accordingly, the vertical center of $|\mathord{\downarrow}\rangle$ branch moves along the $z$-axis.  Nevertheless, the numerical centroid trajectories extracted from the $z=z_0$ section shown in Supplementary Figure~\ref{3D ZB biaxial}(k) are still fairly consistent with the 2D theory in the superposed region of the two eigenmodes, they deviate from the analytical curve only when the two eigenmodes split away in the $z$-direction. In particular, it shows that the superposed region increases with wider beam waist $w_z$. Therefore, the 2D theory is even applicable  for beams whose centers are not confined on the $z=z_0$ plane in a considerable large region, if the beam width $w_z$ is wide enough.

\begin{figure}[b]
\includegraphics[width=0.74\columnwidth,clip]{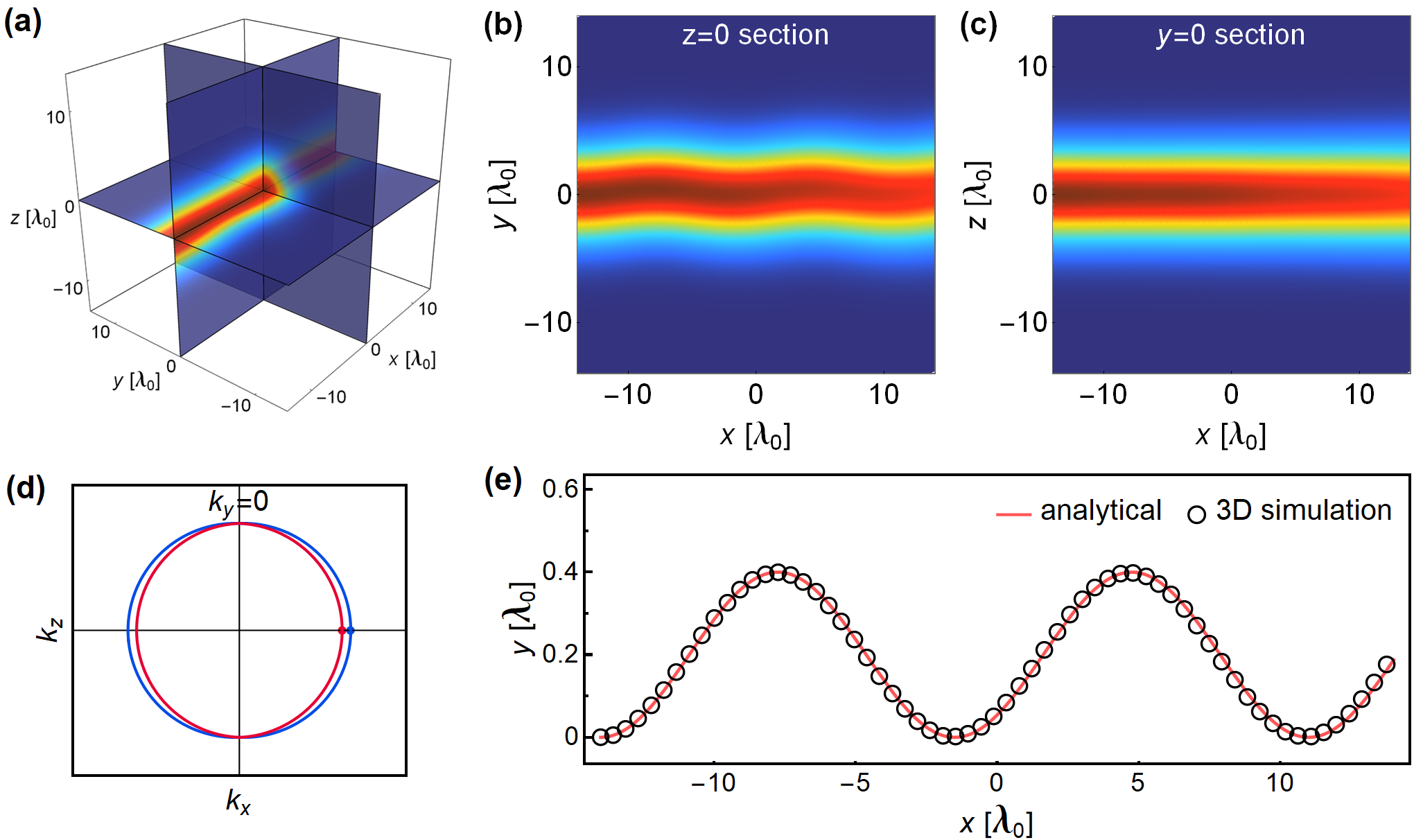}
\caption{ZB effect of a 3D beam with finite widths in both $y$ and $z$ directions in a gyrotropic medium with the parameters $\tensor{\varepsilon}_T=\tensor{\mu}_T=\tensor{I}_{2\times 2}$, $\varepsilon_z=\mu_z=1$, $\mathbf{g}_1=-\mathbf{g}_2^*=(i\,0.1,\ 0.04)^\intercal$. The beam waists in $y$ and $z$ directions are identical $w_y=w_z=5\lambda_0$. (a) Slice view and (b,c) section views of full-wave simulated intensity distribution of the 3D beam. (d) Isofrequency contours of the medium in the $xz$-plane. (e) Centroid trajectories of the 2D theory (red curve) and of the numerical fields on $z=0$ section (black circles) shown in (b).
\label{3D ZB gyrotropic}}\vspace{-15pt}
\end{figure}

\begin{figure}[t]
\includegraphics[width=0.85\columnwidth,clip]{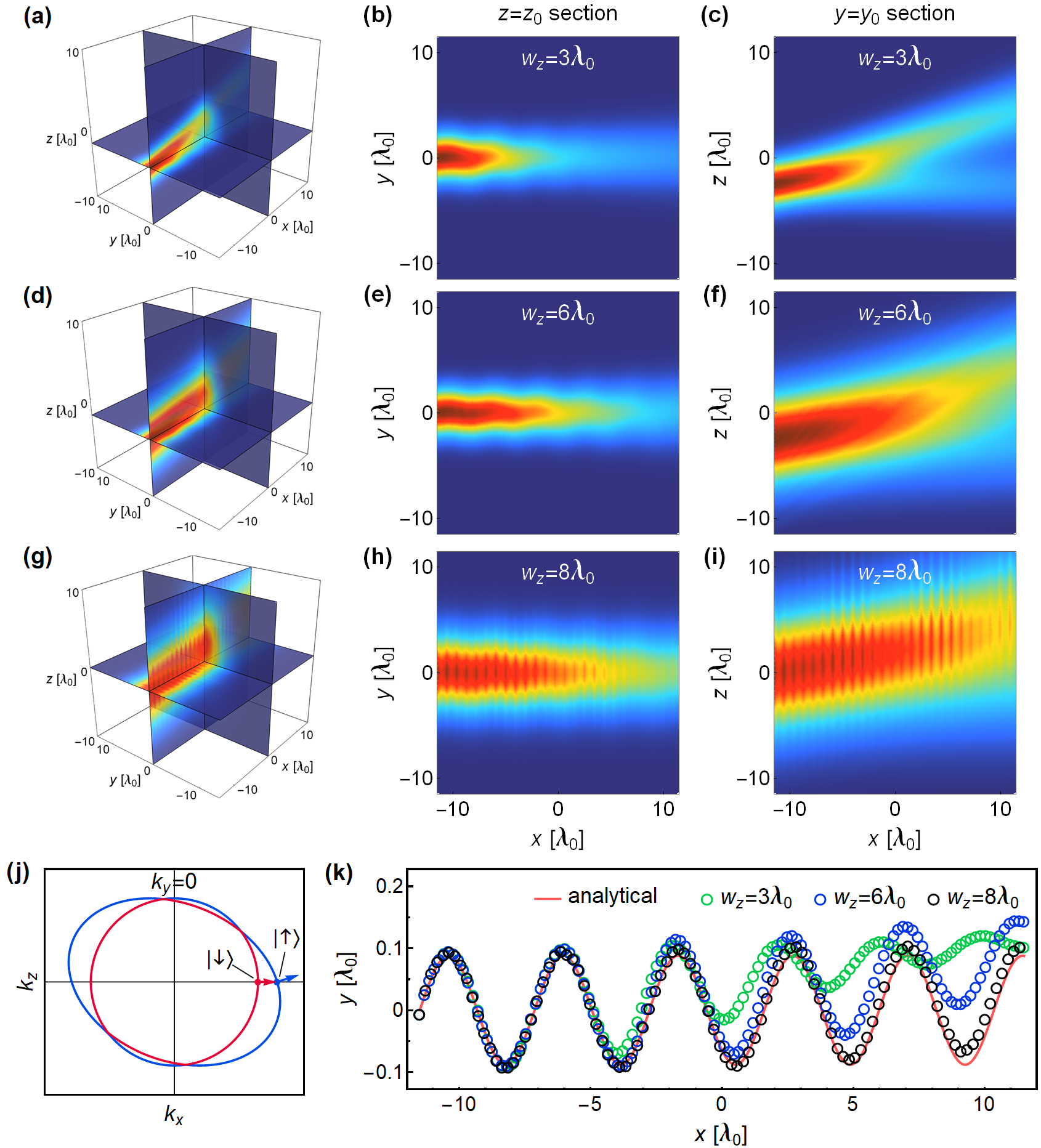}
\caption{ ZB effect of 3D beams in a biaxial dielectric medium with the parameters $\tensor{\varepsilon}_T=\tensor{I}_{2\times 2}$, $\varepsilon_z=1.6$, $\mathbf{g}_1=(0.3,0)^\intercal$, and $\mu/\mu_0=1$. The beam waists in the $z$ directions are respectively (a-c) $w_z=3\lambda_0$, (d-f) $w_z=6\lambda_0$, (g-i) $w_z=8\lambda_0$. (a,d,g) Slice view, (b,e,h) $z=z_0$ section view, and (c,f,i) $y=y_0$ section view of full-wave simulated intensity distributions of the 3D beams, where $(y_0, z_0)$ denotes the center of beam on the initial plane $x=x_0$. (j) Isofrequency contours of the medium in the $xz$-plane, where the red and blue arrows denote the group velocity directions of the two eigenmodes with $\mathbf{k}$ along the positive $k_x$-axis.  (k) Centroid trajectories of the 2D theory (red curve) and of the numerical fields on $z=z_0$ section (black circles) for the three simulated 3D beams shown in (b,e,h).
\label{3D ZB biaxial}}
\end{figure}


\clearpage
\newpage


\section{Theory for genuine non-Abelian Aharonov-Bohm system}

In this section, we supply a self-content introduction of genuine non-Abelian AB system in a more rigorous manner. We will give its precise definition, prove the core property of AB effect, namely the AB phase factors for two homotopic closed paths are identical, and finally deduce the necessary conditions for constructing a genuine non-Abelian AB system.

\subsection*{1. Definition of Genuine non-Abelian AB system}
Consider a simply connected space $M$  endowed with a gauge structure. Mathematically, this system is described by a  $G$-principal fiber bundle, $E\rightarrow M$, where $E$ denotes the total bundle space, $M$ denotes the base manifold, and $G$ denotes the gauge group (the fiber $f_\mathbf{x}$ is homeomorphic to $G$ at each $\mathbf{x}\in M$).
If the gauge field (\(G\)-curvature) is zero \(\hat{\mathcal{F}}_{\mu\nu}=0\)
(i.e. \(G\)- connection is flat) in the whole space, the gauge potential
can be globally written as a pure gauge
\begin{equation}
  \hat{\mathcal{A}}_\mu=\,i\,\hat{U}\partial_\mu \hat{U}^{-1},\qquad (\,\hat{U}\in G\,)
\end{equation}
and can be gauged away
via global gauge transformation \(\hat{U}^{-1}\):
\begin{equation}
  \hat{\mathcal{A}}^\prime_\mu=\hat{U}^{-1}\hat{\mathcal{A}}_{\mu}\hat{U}+i\,\hat{U}^{-1}\partial_\mu \hat{U}=0.
\end{equation}
Therefore, the pure gauge should have no measurable effect in a simply connected
space.
However, if the region of zero curvature \(\hat{\mathcal{F}}_{\mu\nu}=0\) is
only a multiply connected subspace of the whole spacetime, the situation
turns to be very interesting. In this case, 
although the gauge potential still can be written as
\(\hat{\mathcal{A}}_\mu=\,i\,\hat{U}\partial_\mu \hat{U}^{-1}\) locally, it cannot be
gauged away globally via  gauge transformation in the whole multiply connected space, and therefore implies nontrivial physical effect. Such a field-free system with irremovable gauge potential is referred to as an \textbf{Aharonov-Bohm system}~\cite{aharonov1959significance}. 

Consider a particle (wave packet) characterized by a spinor state $|\psi\rangle$ propagating in the zero-field region. From a semi-classical picture, the particle will trace the same trajectory as the case of $\hat{\mathcal{A}}=0$ due to the vanishing non-Abelian Lorentz force.   Nevertheless, the evolution of the state vector \(\left|\psi\right\rangle\) along a
curve \(\gamma: [0,1]\rightarrow M\) is affected by the gauge potential as follows:
\begin{equation}\label{evolution of spinor}
  \left|\psi\right\rangle=\mathcal{P}\,\exp\left[i\int_\gamma \hat{\mathcal{A}}_\mu dx^{\mu}\right] \left|\psi\right\rangle_0=\hat{U}_\gamma\,\left|\psi\right\rangle_0.
\end{equation}
Here, $\left|\psi\right\rangle_0$ denotes the state vector along the same path when $\hat{\mathcal{A}}=0$, where the dynamic phase is included in it. As shown in Eq.~(\ref{evolution of spinor}), the nonzero gauge potential $\hat{\mathcal{A}}$ generates an addition geometric phase factor, namely the non-Abelian AB phase factor, expressed by a path-ordering integration along the curve \(\gamma\) stating at point $\gamma(0)=\mathbf{x}_0$:
\begin{equation}
  \hat{U}_\gamma=\mathcal{P}\,\exp\left[i\int_\gamma \hat{\mathcal{A}}_\mu dx^{\mu}\right]\in G.
\end{equation}
From a geometric point of view, this AB phase factor corresponds to the parallel transport
\(\mathrm{T}_\gamma(p_0)=p_0\cdot\hat{U}_\gamma=\tilde{\gamma}(1)\) along the horizontally
lifted path \(\tilde{\gamma}:[0,1]\rightarrow E\) in the bundle space $E$, where
$p_0=\tilde{\gamma}(0)\in f_{\mathbf{x}_0}$ denotes the starting point of the lifted curve, and $f_{\mathbf{x}_0}$ denotes the fiber at $\mathbf{x}_0$. For a certain reference point $p_0$, the horizontal lift $\tilde{\gamma}$ of $\gamma$ is unique. Physically, the choice of $p_0$ determines the local gauge \(\hat{\mathcal{A}}(\mathbf{x}_0)\) at $\mathbf{x}_0$. And according to the property of parallel transport, the non-Abelian AB phase factor of a concatenate path $\gamma_1\circ\gamma_2$ such that $\gamma_1(0)=\gamma_2(1)$ satisfies $\hat{U}_{\gamma_1\circ\gamma_2}=\hat{U}_{\gamma_1}\circ\hat{U}_{\gamma_2}$.

For a closed path (loop) \({c}\) starting and ending at the same
basepoint \(\mathbf{x}_0\), a local gauge fixes the starting point
\(p_0=\tilde{c}(0)\in f_{\mathbf{x}_0}\) of the lifted path \(\tilde{{c}}\).
However, \(\tilde{{c}}\) need not to be closed, i.e.~the end
point of \(\tilde{{c}}\) can be different from \(p_0\), as shown in Supplementary Figure~\ref{homotopy}(c). The
non-Abelian phase factor for the closed loop \({c}\),
\begin{equation}
  \hat{\mathcal{U}}_c(\hat{\mathcal{A}})=\mathcal{P}\,\exp\left[i\oint_{c} \hat{\mathcal{A}}_\mu dx^{\mu}\right],
\end{equation}
is called the {\bf holonomy} of the horizontally lifted loop \(\tilde{{c}}\)
with respect to the gauge \(\hat{\mathcal{A}}\). The holonomies corresponding to all those closed loops based at
\(\mathbf{x}_0\) constitute a subgroup of the gauge group \(G\):
\begin{equation}
\mathrm{Hol}_{p_0}(\hat{\mathcal{A}})=\Big\{\hat{\mathcal{U}}_c(\hat{\mathcal{A}})\,\Big|\ {c}(0)={c}(1)=\mathbf{x}_0,\ \tilde{{c}}(0)=p_0 \Big\}\subseteq G,
\end{equation}
which is the \textbf{\bf{holonomy group}} of \(\hat{\mathcal{A}}\)
with the reference point \(p_0\).
\(\hat{\mathcal{U}}_c(\hat{\mathcal{A}})=\mathbb{I}\) if and only if
\(\tilde{{c}}\) is also closed, corresponding to the unit
element of \(\mathrm{Hol}_{p_0}(\hat{\mathcal{A}})\)~\cite{chruscinski2012geometric}.
In flat bundles with zero curvature \(\hat{\mathcal{F}}_{\mu\nu}=0\),
the holonomy group of the flat connection $\hat{\mathcal{A}}$ is primarily determined by the topology of the base manifold $M$. The illustration of the holonomies in a flat bundle in shown in Supplementary Figure~\ref{homotopy}(c).

The Aharonov-Bohm system with non-Abelian gauge potential was first introduced by Wu and Yang in their  paper connecting gauge theory with fiber bundle~\cite{Wuyang}, where they predicted intriguing phenomena for the scattering of nucleon carrying weak isospin by a $\mathrm{SU}(2)$ flux tube, such as the spatial fluctuation of  proton-neutron mixing ratio.  Later on, Aharonov and Casher proposed an electric counterpart of the original AB effect, namely  the scattering of neutral magnetic dipoles by a charged line~\cite{aharonov1984topological}. Indeed, both the Wu-Yang scheme and the Aharonov-Casher effect are mathematically equivalent to an AB system with a $\mathrm{SU}(2)$ non-Abelian vortex~\cite{anandan1989electromagnetic,oh1994equivalence}. And the optical realization with anisotropic media has been proposed in Ref.~\onlinecite{JensenLi2015PRL}. However, it can be shown that this kind of systems with a single $\mathrm{SU}(2)$ vortex always can be decomposed into two decoupled Abelian subsystems under a proper gauge~\cite{HovathyPRD}, and the holonomy for  a closed path $c_n$ winding around the vortex $n$ times takes the form~\cite{chruscinski2012geometric}
\begin{equation}
  \hat{\mathcal{U}}_{c_n}=\mathcal{P}\exp\left[\oint_{c_n}\hat{\mathcal{A}}_\mu dx^\mu\right]=\exp\left(i\,n\Phi\,\hat{\sigma}_3\right)=
  \begin{pmatrix}
    \exp(i\,n\Phi) & 0 \\ 0 & \exp(-i\,n\Phi)
  \end{pmatrix} , 
\end{equation}
where $\Phi\hat{\sigma}_3$ is the non-Abelian flux of the vortex. Accordingly, the holonomices of arbitrary two loops commute with each other, thus the holonomy group is still Abelian. In this sense, such systems can be still treated as Abelian or apparently non-Abelian, as specified in the main text~\cite{Raman1986,Goldman2014}. In a more rigorous manner, a genuine non-Abelian AB system can be defined as follows~\cite{Raman1986}:
\begin{definition}
  An Aharonov-Bohm system is called \textbf{genuinely non-Abelian} if and only if the corresponding Holonomy group \(\mathrm{Hol}(\hat{\mathcal{A}})\) of the flat connection (gauge potential) $\hat{\mathcal{A}}$
is a non-Abelian group. 
\end{definition}

It is noteworthy that a different route of generalizing AB effect is the non-Abelian vortex-vortex scattering~\cite{wilczek1990space,bucher1991aharonov,lo1993non,nayak2008non}. In contrast to the non-Abelian vortices in this route being dynamic quasi-particles and respecting non-Abelian statistics, the non-Abelian vortices in our AB systems are merely non-dynamic background.

\subsection*{2. Necessary conditions of genuine non-Abelian AB system}
In order to characterize the holonomy group of a flat bundle, we first introduce the concept of path homotopy.
\textbf{Path homotopy} refers to a topological equivalence relation ``$\simeq$'' for paths sharing the fixed endpoints. Intuitively, if a path $\gamma_1$ can deform continuously into another path $\gamma_2$ while keeping the endpoints fixed, $\gamma_1$ and $\gamma_2$ are said to be homotopic to each other, $\gamma_1\simeq\gamma_2$ (see Supplementary Figure~\ref{homotopy}(a)).  Strictly speaking, path homotopy denotes precisely the mapping of the continuous deformation from $\gamma_1$ to $\gamma_2$~\cite{rotman2013introduction}. All the paths that are homotopic to each other form the \textbf{homotopy equivalence class} $[\gamma]=\big\{\gamma'\,\big| \gamma'(0)=\mathbf{x}_0,\gamma'(1)=\mathbf{x}_1,\gamma'\simeq\gamma\big\}$, where any element $\gamma$  of the class can be chosen as the class representative.

For closed paths through the same basepoint $\mathbf{x}_0$ in the base manifold, their homotopy classes constitute the quotient set with respect to the homotopy equivalence relation ``$\simeq$'':
\begin{equation}
  \pi_1(M,\mathbf{x}_0)=\Big\{[{c}]\,\Big|\ {c}(0)={c}(1)=\mathbf{x}_0\Big\}=\Big\{{c}\,\Big|\ {c}(0)={c}(1)=\mathbf{x}_0\Big\}\big/\simeq.
\end{equation}
By introducing the multiplication between classes, $[c_1]\circ[c_2]=[c_1\circ c_2]$ ($c_1\circ c_2$ denotes the concatenation of the two loops $c_1,\,c_2$), $\pi_1(M,\mathbf{x}_0)$ also forms a group, known as the \textbf{{first fundamental group}} of the base manifold \(M\). Though the definition of fundamental group relies on the basepoint $\mathbf{x}_0$, it can be proved that $\pi_1(M,\mathbf{x}_0)$ and $\pi_1(M,\mathbf{x}'_0)$ are isomorphic for arbitrary $\mathbf{x}_0,\mathbf{x}'_0\in M$ as long as $M$ is path-connected. The fundamental group is an important topological invariant for characterizing the properties of a manifold. 
Moreover, we will prove that the first fundamental group of the base manifold is homomorphic to the holonomy group of the corresponding flat bundle in the following part of this section.

In AB systems with Abelian gauge potentials, it is well-known that the AB phase factors along two homotopic loops $c_1$, $c_2$ are identical due to the null magnetic flux passing through the enclosed area of the combined loop $c_1-c_2$. Indeed, this result can be generalized to a theorem for generic AB systems with non-Abelian gauge potentials as follows.

 \begin{figure}[t]
 \centering
 \includegraphics[width=0.7\columnwidth,clip]{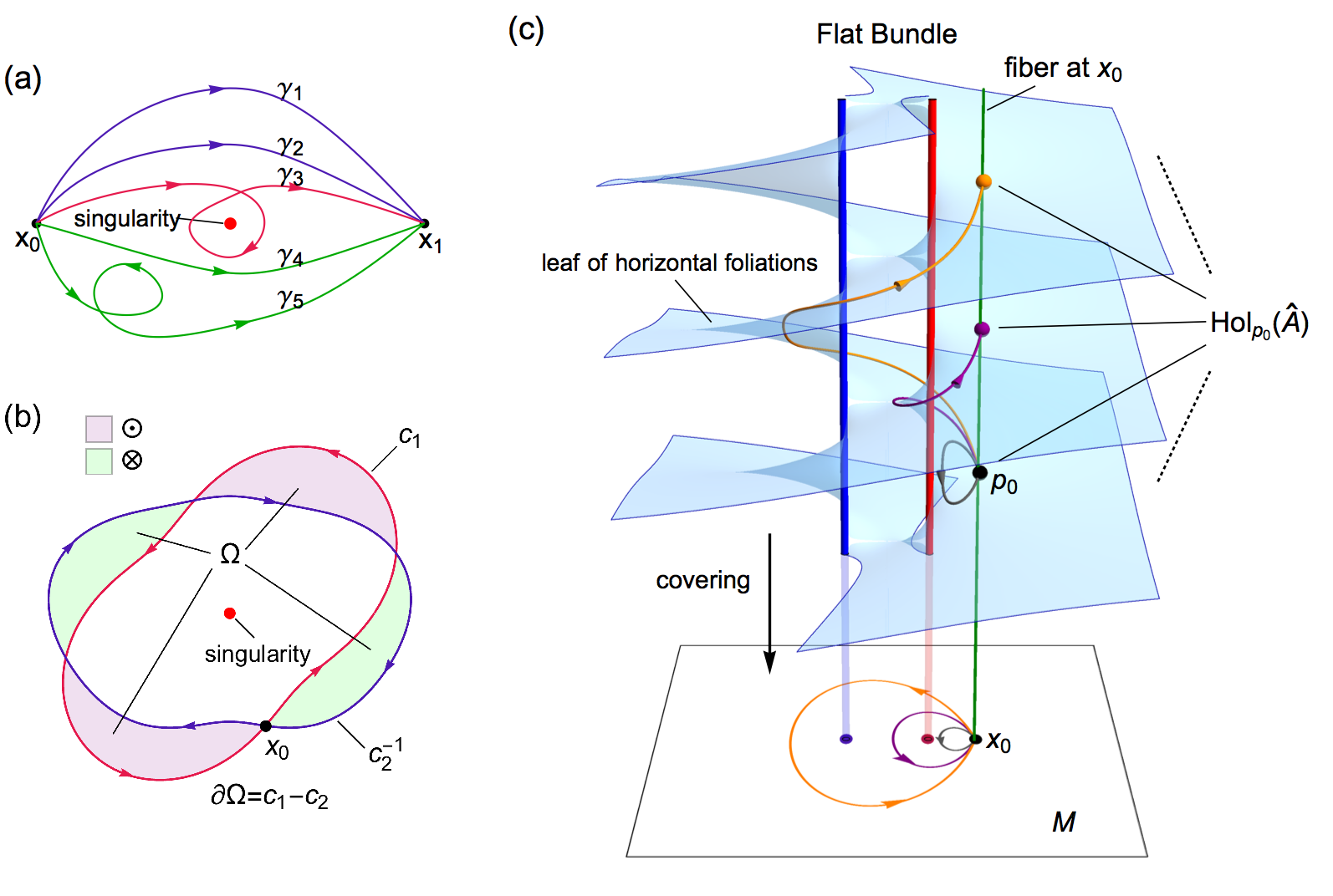}
 \caption{(a) Illustration of homotopic paths with fixed endpoints $\mathbf{x}_0$, $\mathbf{x}_1$ in a punctured plane.  The curves in a same color can deform continuously into each other and thus are path-homotopic, namely $\gamma_1\simeq\gamma_2$, $\gamma_4\simeq\gamma_5$, whereas the curves in different colors are not homotopic.  (b) Two homotopic loops based at $\mathbf{x}_0$ form the boundary of a 2D orientable region $\Omega$ (mathematically, a 2-chain). The purple (green) color indicates the orientation of the region is outward (inward) the paper plane. 
(c) Holonomies in a flat bundle with a twice-punctured base manifold $M$. For the three closed paths based at $\mathbf{x}_0$ in different homotopy classes, their horizontally lifted paths in the bundle space start from $p_0$ to different terminals on the fiber $f_{\mathbf{x}_0}$. All of the terminals of such lifted paths compose the holonomy group $\mathrm{Hol}_{p_0}(\hat{\mathcal{A}})$. The light blue surface illustrates a leaf of horizontal foliations through the point $p_0$ embedding in the bundle space, which covers the base manifold $M$. Note: the surface is merely a cartoon schematic but not an actual universal covering space of $M$. Thus the lifted paths on this surface cannot differentiate all of different homotopy classes. 
 \label{homotopy}}
 \end{figure}

\begin{theorem}\label{theom}
In a flat bundle, the holonomies of
two path-homotopic loops \({c}_1\) and \({c}_2\),
i.e.~they belong to the same homotopy class
\({c_1},{c_2}\in[c]\), are identical
\begin{equation}
\hat{\mathcal{U}}_{c_1}(\hat{\mathcal{A}})=\hat{\mathcal{U}}_{c_2}(\hat{\mathcal{A}})\equiv\hat{\mathcal{U}}_{[c]}(\hat{\mathcal{A}}).
\end{equation}
\end{theorem}

\begin{proof}
The proof of this theorem can be found in the textbooks of differential geometry~\cite{kobayashi1963foundations} and in the literature~\cite{Raman1986,oh1988propagator}. For completeness, we give a concise but accessible proof imitating the derivation of the Abelian case.   Let us examine the following product
\[\hat{\mathcal{U}}_{c_1}(\hat{\mathcal{A}})\,\hat{\mathcal{U}}_{c_2}(\hat{\mathcal{A}})^{-1}=\hat{\mathcal{U}}_{c_1}(\hat{\mathcal{A}})\,\hat{\mathcal{U}}_{c_2^{-1}}(\hat{\mathcal{A}})
=\hat{\mathcal{U}}_{c_1\circ c_2^{-1}}(\hat{\mathcal{A}})
=\mathcal{P}\,\exp\left[i\oint_{{c}_1\circ{c}_2^{-1}} \hat{\mathcal{A}}_\mu dx^{\mu}\right].
\]
The Hurewicz's theorem implies that the two homotopic loops \(c_1\),
\(c_2\) are necessarily homologous when treated as
1-cycles~\cite{rotman2013introduction}. According to the definition of homology, \(c_1\circ c_2^{-1}\sim c_1-c_2=\partial \Omega\) must be
the boundary of a 2-chain\(\Omega\) (the 2D orientable region enclosed by $c_1-c_2$ as shown in Supplementary Figure~\ref{homotopy}(b)).
In terms of the \textbf{non-Abelian Stokes theorem}~\cite{aref1980non,fishbane1981stokes,broda2002non} and  $\hat{\mathcal{F}}_{\mu\nu}=0$, the path-ordered exponential along the loop $c_1\circ{c}_2^{-1}$ equals  the identity  element of $G$: 
\begin{equation}
\hat{\mathcal{U}}_{c_1}(\hat{\mathcal{A}})\,\hat{\mathcal{U}}_{c_2}(\hat{\mathcal{A}})^{-1}=\mathcal{P}\,\exp\left[i\oint_{\partial\Omega} \hat{\mathcal{A}}_\mu dx^{\mu}\right]=\mathscr{P}_{\hspace{-1pt}s}
\exp\Bigg[\frac{i}{2}\int_\Omega\,\hat{U}(\mathbf{x},\mathbf{x}_0)^{-1}\underbrace{\hat{\mathcal{F}}_{\mu\nu}(\mathbf{x})}_{=0}\hat{U}(\mathbf{x}_0,\mathbf{x})\,dx^\mu\wedge dx^\nu\Bigg]=\mathbb{I},
\end{equation}
where \(\mathscr{P}_{\hspace{-1pt}s}\) denotes the ``surface ordering'', and $\hat{U}(\mathbf{x}_0,\mathbf{x})=\mathcal{P}\exp\left[i\int_\mathbf{x}^{\mathbf{x}_0}\hat{\mathcal{A}}_\mu dx^\mu\right]$ is the non-Abelian phase factor along a path connecting the field point $\mathbf{x}$ and the basepoint $\mathbf{x}_0$. As such, the homotopic invariance of the holonomies in flat bundle has been proved.
\vspace{5pt}
\end{proof}

Physically, the theorem manifests the  quantization of the AB phase factors along paths with fixed endpoints, namely they can only take discrete values (in general, they are matrices) for different homotopic paths, while the AB phase factor is invariant against the continuous deformation of the path in the field-free region. This is the key signature of AB systems associated with either Abelian or non-Abelian gauge potentials. In a simply connected space, the fundamental group is trivial, namely all
loops belong to \([\mathbf{x}_0]\). Hence, any state will
return to itself as it travels along a closed loop. \textbf{Therefore, an AB system requires the fundamental group of the base manifold is nontrivial.}
According to \textbf{Theorem}~\ref{theom}, the holonomy group of a flat bundle and the fundamental group of the corresponding base manifold are naturally connected.

\begin{corollary}\label{corollary}
For a flat bundle $E\rightarrow M$, there is a \textbf{surjective group homomorphism} from the first fundamental group \(\pi_1(M,\mathbf{x}_0)\) of the base manifold $M$
 to the holonomy group \(\mathrm{Hol}_{p_0}(\hat{\mathcal{A}})\) of the bundle:
\begin{equation}
\pi_1(M,\mathbf{x}_0)\ni[c]\ \longrightarrow\ \hat{\mathcal{U}}_{[c]}(\hat{\mathcal{A}})\in\mathrm{Hol}_{p_0}(\hat{\mathcal{A}})=\big\{\hat{\mathcal{U}}_{[c]}\,\big|\ [c]\in\pi_1(M,\mathbf{x}_0),\ \tilde{c}(0)=p_0\big\}.
\end{equation}
In other words, the holomony group
\(\mathrm{Hol}_{p_0}(\hat{\mathcal{A}})\) of the flat bundle is just a representation of the fundamental group
\(\pi_1(M,\mathbf{x}_0)\) of the base manifold. 
\end{corollary}


As the holonomy group $\mathrm{Hol}(\hat{\mathcal{A}})$ is a subgroup of the gauge group $G$, an Abelian gauge group indicates all of the holonomies commute with each other. Meanwhile, in terms of \textbf{Corollary}~\ref{corollary}, for an AB system with a non-Abelian holonomy group,  there should exit two loops $c_1$, $c_2$ such that ${\hat{\mathcal{U}}}_{[c_1]}{\hat{\mathcal{U}}}_{[c_2]}\neq {\hat{\mathcal{U}}}_{[c_2]}{\hat{\mathcal{U}}}_{[c_1]}$. Thus $[c_1\circ c_2]\neq [c_2\circ c_1]$, in other words, $\pi_1(M)$ must also be non-Abelian. To summarize, we obtain

\begin{necessary*}
 \begin{enumerate}[itemsep=0mm]
   \item \textit{The
gauge group \(G\) is non-Abelian;}
 \item \textit{The first fundamental group
\(\pi_1(M)\) of the base manifold is non-Abelian.}
 \end{enumerate} 
\end{necessary*}

\newpage
\section{Material construction of non-Abelian AB system}

Inspired by the scheme of two noncommutative $SU(3)$ vortices proposed in Ref.~\onlinecite{Raman1986}, we will proceed to construct a $\mathrm{SU}(2)$ genuine non-Abelian AB system.
We assume the synthetic $\mathrm{SU}(2)$ gauge potential only carrying two nonzero components: $\hat{\mathcal{A}}=\mathcal{A}^1\hat{\sigma}_1+\mathcal{A}^2\hat{\sigma}_2$. Besides, a nontrivial gauge potential with vanishing field $\hat{\mathcal{F}}_{ij}=0$ can be locally expressed in the
pure gauge form $\hat{\mathcal{A}}=i \hat{U}\nabla \hat{U}^{-1}$. If we demand $\mathcal{A}^2=0$ when $y>0$, and $\mathcal{A}^1=0$ when $y<0$, the unitary transformation matrix can be simply expressed as follows in upper and lower semi-space respectively: 
\begin{equation}
  \hat{U}=\exp\left(i\,\zeta_1(\mathbf{r})\,\hat{\sigma}_1\right)\quad (y>0),\qquad \hat{U}=\exp\left(i\,\zeta_2(\mathbf{r})\,\hat{\sigma}_2\right)\quad (y<0).
\end{equation}
Here, $\zeta_1(\mathbf{r}),\,\zeta_2(\mathbf{r})$, referred to as pre-potentials,  are
some multivalued functions of $x,y$ that smoothly tend to zero at $y=0$, and the corresponding components of the vector potential take the form $\mathcal{A}^i=\nabla\zeta_i(\mathbf{r})\quad (i=1,2)$.

\begin{figure}[b]
 \centering
 \includegraphics[width=0.8\columnwidth,clip]{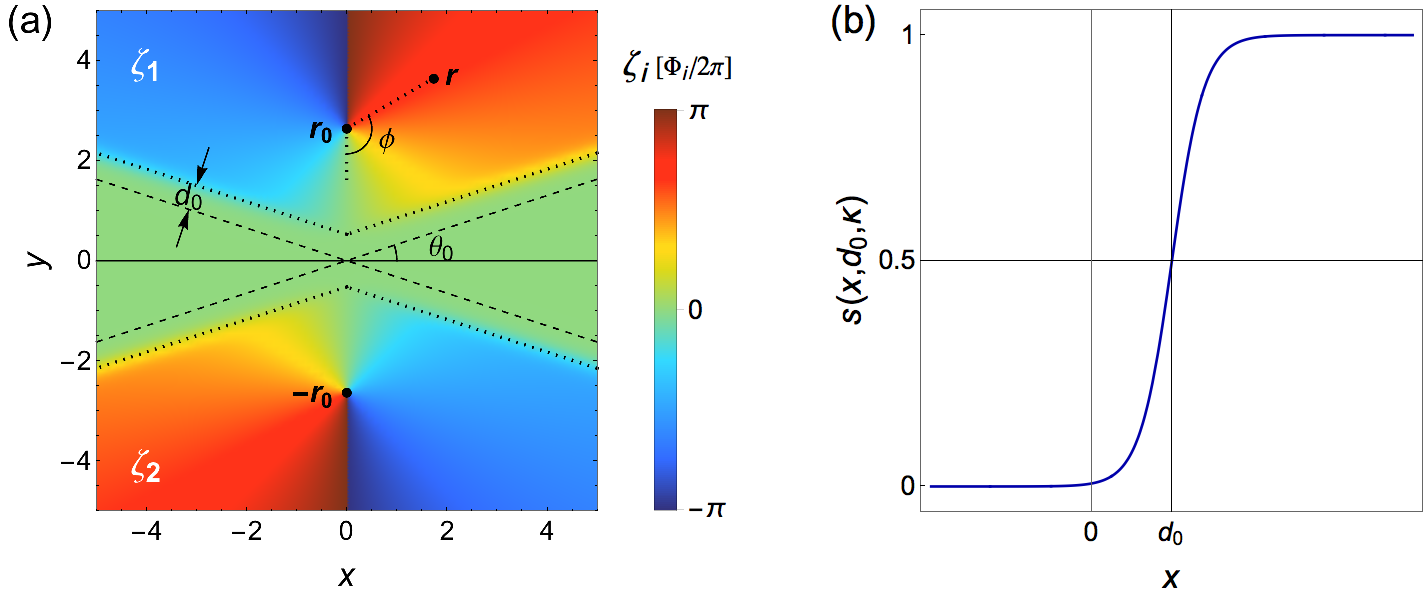}
 \caption{(a) Spatial distributions of the pre-potentials $\zeta_1$ (in the upper half-plane $y>0$) and $\zeta_2$ (in the lower half-plane $y<0$), whose gradients give the corresponding components of the synthetic $\mathrm{SU}(2)$ gauge potential $\mathcal{A}^i=\nabla\zeta_i\ (i=1,2)$. (b) The profile of the sigmoid function $s(x,d_0,\kappa)$ for constructing the cutoff function $S(\mathbf{r})$ in Eq.~(\ref{sigmoid function}).  \label{prepotantial}}
\end{figure}

The pre-potentials with nontrivial topology can be constructed in the
following way:
\begin{equation}
  \zeta_1(\mathbf{r})=\frac{\Phi_1}{2\pi}S(\mathbf{r})\,\phi(\mathbf{r}-\mathbf{r}_0),\quad \zeta_2(\mathbf{r})=\frac{\Phi_2}{2\pi}S(-\mathbf{r})\,\phi(\mathbf{r}+\mathbf{r}_0),
\end{equation}
with the corresponding synthetic gauge potential
\begin{equation}\label{gauge potential in AB system}
\raisebox{7pt}{ $\hat{\mathcal{A}}=\,\left\{\rule{0cm}{1cm}\ \right.$}\begin{aligned}
\mathcal{A}^1\hat{\sigma}_1=&\ \frac{\Phi_1}{2\pi}\nabla\Big(S(\mathbf{r})\,\phi(\mathbf{r}-\mathbf{r}_0)\Big)\hat{\sigma}_1\ =\frac{\Phi_{1}}{2\pi}\,\Big(S(\mathbf{r})\ {\color{Red}\nabla\phi(\mathbf{r}-\mathbf{r}_0)}+\phi(\mathbf{r}-\mathbf{r}_0)\nabla S(\mathbf{r})\Big)\hat{\sigma}_1,\quad\ \  (y>0)\\
\mathcal{A}^2\hat{\sigma}_2=&\,\frac{\Phi_2}{2\pi}\nabla\Big(S(-\mathbf{r})\,\phi(\mathbf{r}+\mathbf{r}_0)\Big)\hat{\sigma}_2=\frac{\Phi_{2}}{2\pi}\Big(S(-\mathbf{r}){\color{Red}\underbrace{\nabla\phi(\mathbf{r}+\mathbf{r}_0)}_{\displaystyle\text{vortex}}}+\phi(\mathbf{r}+\mathbf{r}_0)\nabla S(-\mathbf{r})\Big)\hat{\sigma}_2,\quad (y<0)
\end{aligned}
\end{equation}
where \(\phi(\mathbf{r}\mp\mathbf{r}_0)\) is the polar angle (see Supplementary Figure~\ref{prepotantial}(a)) with
respect to \(\pm\mathbf{r}_0=\pm r_0\,\mathbf{e}_y\) whose gradient represents an irrotational free vortex located at $\pm\mathbf{r}_0$:
\begin{equation}
  \nabla\phi(\mathbf{r}\mp\mathbf{r}_0)=\frac{\mathbf{e}_{\phi}}{\left|\mathbf{r}\mp\mathbf{r}_0\right|}=\frac{-(y\mp r_0)\,\mathbf{e}_{x}+x\,\mathbf{e}_{y}}{\left|\mathbf{r}\mp\mathbf{r}_0\right|^2},
\end{equation}
and \(S(\pm\mathbf{r})\) is a smooth cutoff function such that (i) the two components
\(\mathcal{A}^{1},\mathcal{A}^2\) are gradually compressed
in the upper and lower half-space without overlap, and (ii) $\mathcal{A}^{1,2}$ tends to a free vortex nearby $\pm\mathbf{r}_0$ respectively, i.e.
\begin{equation*}\label{limit AB potential}
  \begin{aligned}
  \text{(i)}\quad  \mathcal{A}^i&\rightarrow 0\qquad\qquad\quad\ \, \Leftrightarrow\  S(\pm\mathbf{r})\rightarrow 0,  \ \ \,\text{as } y\rightarrow 0;\\
 \text{(ii)}\quad  \mathcal{A}^i&=\frac{\Phi_i}{2\pi}\nabla\phi(\mathbf{r}\mp\mathbf{r}_0)\ \Leftrightarrow\  S(\pm\mathbf{r})=1,  \quad\text{as } |\mathbf{r}\mp\mathbf{r}_0|<\delta r.
  \end{aligned}
\end{equation*}
Meanwhile, since $\phi(\mathbf{r}\mp\mathbf{r}_0)$ is multivalued, the branch cut of $\phi(\mathbf{r}\mp\mathbf{r}_0)$ is selected along $\mathbf{r}=\pm(\mathbf{r}_0+y\mathbf{e}_y)\ (y>0)$  as shown in Supplementary Figure~\ref{prepotantial}(a). To ensure the monodromy of $\mathcal{A}^i$, we should also require $S(\pm\mathbf{r})=1$
in a neighborhood of the branch cut of
\(\phi(\mathbf{r}\mp\mathbf{r}_0)\). The form of \(S(\mathbf{r})\) is
not unique, here we adopt an approximate but rather simple expression:
\begin{equation}\label{cutoff function}
  \begin{split}
  S(\mathbf{r})=&\,s(y \cos\theta_0+x\sin\theta_0,\kappa,d_0)\cdot s(y \cos\theta_0-x\sin\theta_0,\kappa,d_0)
  \end{split}
\end{equation}
with the sigmoid function 
 \begin{equation}\label{sigmoid function}
   s(x,k,d_0)=\frac{1}{\exp[\kappa(d_0-x)]+1},
 \end{equation} 
and the values of the parameters $\theta_0$, $\kappa$, $d_0$ do not affect the AB phase factors as long as the asymptotic conditions of $S(\pm\mathbf{r})$ are met. The obtained pre-potentials $\zeta_1,\,\zeta_2$ and the gauge potential $\hat{\mathcal{A}}$ are plotted  in Supplementary Figure~~\ref{g-vector}(a) and in Fig.~2a of the main text respectively. 

\begin{figure}
 \centering
 \includegraphics[width=0.92\columnwidth,clip]{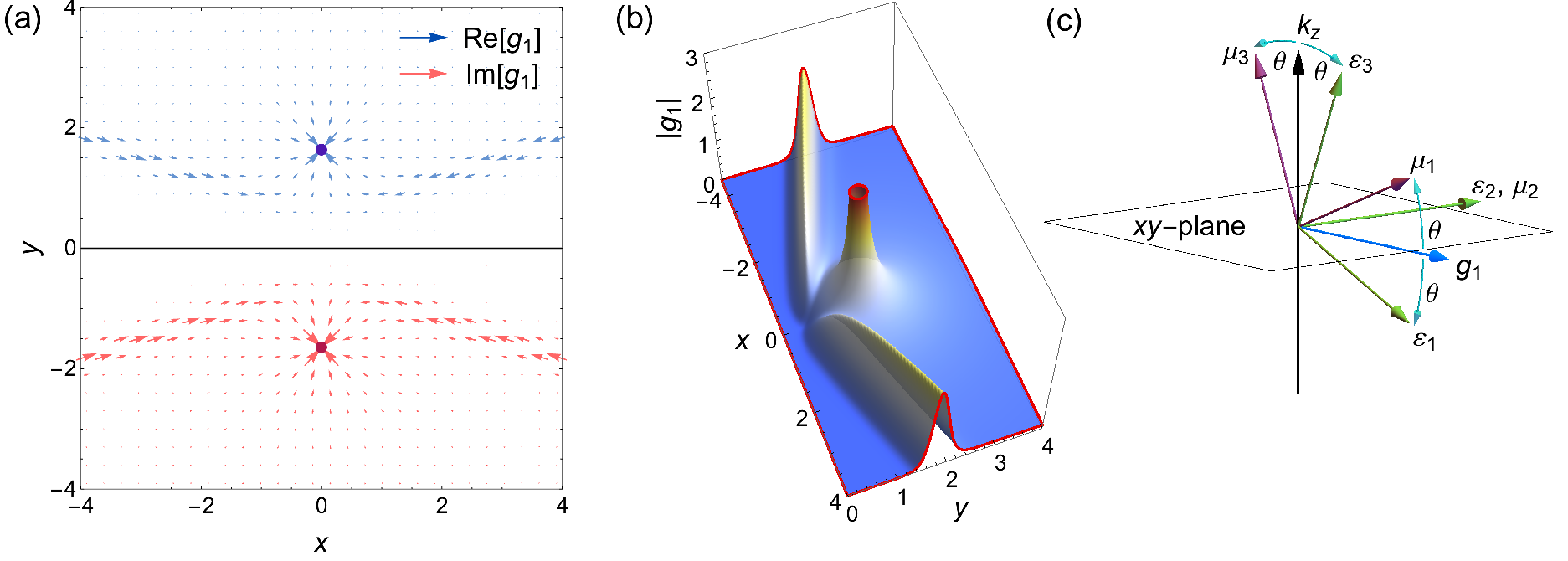}
 \caption{(a) Spatial distribution of  $\mathbf{g}_1=-\mathbf{g}_2^*$ (the off-diagonal term of $\tensor{\varepsilon}$, $\tensor{\mu}$). (b) Distribution of $|\mathbf{g}_1|$, which determines the principal components of $\tensor{\varepsilon}$, $\tensor{\mu}$ (see Eq.~(\ref{principal components})), for the reciprocal anisotropic medium in the upper half-plane. (c) The orientation of the principal frames of $\tensor{\varepsilon}$, $\tensor{\mu}$ for the medium in the upper half-plane.  \label{g-vector}}
\end{figure}

Notably, the non-Abelian holonomy of a loop encircling a vortex depends on the choice of basepoint. Nevertheless, in our system, the gauge
potential is commutative in either upper
or lower half-space. For an arbitary loop $c_1$ ($c_2$) lying entirely in the upper (lower) half-space and encircling the upper (lower) vortex anticlockwise once, its holonomy, $\hat{\mathcal{U}}_{[c_i]}$ ($i=1,2$), is uniquely determined as 
\begin{equation}\label{holonomy generator}
  \hat{\mathcal{U}}_{c_i}=\mathcal{P}\exp\left[i\oint_{c_i} \hat{\mathcal{A}}^i\cdot d\mathbf{r}\,\hat{\sigma}_i\right]=\exp\left[i\oint_{c_i} d\zeta_i\,\hat{\sigma}_i\right]=\exp\left[i\frac{\Phi_i}{2\pi}\oint_{|\mathbf{r}\mp\mathbf{r}_0|<\delta r} d\phi\,\hat{\sigma}_i\right]=\exp\left[i\Phi_i\hat{\sigma}_i\right],\quad (i=1,2)
\end{equation}
where the integral in the second-to-last step is  along a sufficiently small circle enclosing the vortex, and $\Phi_1,\,\Phi_2 $ can be regarded as the flux of the two vortices respectively.
For loops through a basepoint on the $x$-axis $(y=0)$, Eq.~(\ref{holonomy generator}) is still valid. 
Therefore, $\hat{\mathcal{U}}_{c_1}$, $\hat{\mathcal{U}}_{c_2}$ can be used as the two generators to obtain the holonomy group (being homomorphic to the free group $\mathbb{Z}*\mathbb{Z}$) of this AB system.
As shown in Supplementary Figure~\ref{prepotantial}(a) as well as Fig.~2a of the main text, the cutoff function $S(\mathbf{r})$ in Eq.~(\ref{cutoff function}) renders the gauge potential $\hat{\mathcal{A}}$ vanishing in the whole conic region of $|\arctan(y/x)|<\theta_0$. As such, each point on the cross section of an optical beam with a finite width can undergo a unique non-Abelian AB phase factor when arriving at the screen. 

From the synthetic gauge potential in Eq.~(\ref{gauge potential in AB system}) and Table.~1 in the main text, we can inversely construct the desired material parameters. To guarantee $\hat{\mathcal{A}}_0\equiv0$ and $V_0=\text{const.}$, we demand that the diagonal components of $\tensor{\varepsilon}$, $\tensor{\mu}$  obey $\varepsilon_T=\mu_T$, $\varepsilon_z=\mu_z$ and are homogeneous in the whole space, meanwhile, their off-diagonal parts satisfy $\mathbf{g}_1=-\mathbf{g}_2^*$. As such, the synthetic  $\hat{\sigma}_1$-vortex in the upper half-space $(y>0)$ can be made using reciprocal anisotropic metamaterials~\cite{JensenLi2015PRL} with $\mathbf{g}_1=-\mathbf{g}_2=\mathbf{e}_z\times\mathcal{A}^1/k_0$. The constitutive tensors take the following forms:
\begin{equation}\label{AB material upper}
  \tensor{\varepsilon}/\varepsilon_0=\left(\begin{array}{cc|c} \varepsilon_T &  0 & -{\mathcal{A}^1_y}/k_0\\
0 & \varepsilon_T  & {\mathcal{A}^1_x}/k_0 \\\hline
-{\mathcal{A}^1_y}/k_0 & {\mathcal{A}^1_x}/k_0 & \varepsilon_z
\end{array}\right),\qquad
\tensor{\mu}/\mu_0=\left(\begin{array}{cc|c} \varepsilon_T &  0 & {\mathcal{A}^1_y}/k_0\\
0 & \varepsilon_T  & -{\mathcal{A}^1_x}/k_0 \\\hline
{\mathcal{A}^1_y}/k_0 & -{\mathcal{A}^1_x}/k_0 & \varepsilon_z
\end{array}\right).
\end{equation}
Whereas the synthetic $\hat{\sigma}_2$-vortex in the lower half-space $(y<0)$ need to be constructed out of gyrotropic materials with $\mathbf{g}_1=\mathbf{g}_2=i\,\mathbf{e}_z\times\mathcal{A}^2/k_0$, so the relative permittivity and permeability are of entirely equal form 
 \begin{equation}\arraycolsep=3.4pt\def\arraystretch{1.5}\label{AB material lower}
\tensor{\varepsilon}/\varepsilon_0=\tensor{\mu}/\mu_0=\left(\begin{array}{cc|c} \varepsilon_T &  0 & -i{\mathcal{A}^2_y}/k_0\\
0 & \varepsilon_T  & i{\mathcal{A}^2_x}/k_0 \\\hline
i{\mathcal{A}^2_y}/k_0 & -i{\mathcal{A}^2_x}/k_0 & \varepsilon_z
\end{array}\right).
 \end{equation}
The orientation of the off-diagonal vector $\mathbf{g}_1=-\mathbf{g}_2^*$ is plotted in Supplementary Figure~\ref{g-vector}(a).

We further analyze the property of the reciprocal anisotropic medium used in the upper half-plane.
By diagonalizing $\tensor{\varepsilon}$ and $\tensor{\mu}$ given by Eq.~(\ref{AB material upper}) into $\tilde{\varepsilon}/\varepsilon_0=\mathrm{diag}(\varepsilon_1,\varepsilon_2,\varepsilon_3)$ and $\tilde{\mu}/\mu_0=\mathrm{diag}(\mu_1,\mu_2,\mu_3)$, we obtain their principal components:
\begin{equation}\label{principal components}
  \begin{gathered}
  \varepsilon_1=\mu_1=\frac{1}{2}\left(\varepsilon_T+\varepsilon_z-\sqrt{(\varepsilon_T-\varepsilon_z)^2+4 |\mathbf{g}_1|^2}\right),\\
   \varepsilon_2=\mu_2=\varepsilon_T,\\
  \varepsilon_3=\mu_3=\frac{1}{2}\left(\varepsilon_T+\varepsilon_z+\sqrt{(\varepsilon_T-\varepsilon_z)^2+4 |\mathbf{g}_1|^2}\right).
 \end{gathered}
\end{equation}
The distribution of $|\mathbf{g}_1|$ is plotted in Supplementary Figure~\ref{g-vector}(b).
Despite the identical principal components of $\tensor{\varepsilon}$ and $\tensor{\mu}$ , their principal frames do not coincide as shown in Supplementary Figure~\ref{g-vector}(c). For a certain $\mathbf{g}_1=-\mathbf{g}_2$,  we have
\begin{equation}
  \begin{aligned}
  \mathbf{e}_{\varepsilon_1}=&\,\cos\theta\,\mathbf{e}_{\mathbf{g}_1}-\sin\theta\,\mathbf{e}_z,\quad
  \mathbf{e}_{\varepsilon_2}=\mathbf{e}_z\times\mathbf{e}_{\mathbf{g}_1},\quad
  \mathbf{e}_{\varepsilon_3}=\sin\theta\,\mathbf{e}_{\mathbf{g}_1}+\cos\theta\,\mathbf{e}_z,\\
  \mathbf{e}_{\mu_1}=&\,\cos\theta\,\mathbf{e}_{\mathbf{g}_1}+\sin\theta\,\mathbf{e}_z,\quad
  \mathbf{e}_{\mu_2}=\mathbf{e}_z\times\mathbf{e}_{\mathbf{g}_1},\quad
  \mathbf{e}_{\mu_3}=-\sin\theta\,\mathbf{e}_{\mathbf{g}_1}+\cos\theta\,\mathbf{e}_z,
\end{aligned}
\end{equation}
where $\theta=\arctan\left[{\frac{\sqrt{(\varepsilon_T-\varepsilon_z)^2+4|\mathbf{g}_1|^2}+(\varepsilon_T-\varepsilon_z)}{2|\mathbf{g}_1|}}\right]$ is the angle between $\varepsilon_1$-axis and $xy$-plane. In particular, if $\varepsilon_z=\varepsilon_T$, the parameters are reduced to $\varepsilon_1=\mu_1=\varepsilon_T-|\mathbf{g}_1|,\ \varepsilon_2=\mu_2=\varepsilon_T,\ \varepsilon_3=\mu_3=\varepsilon_T+|\mathbf{g}_1|,$ and $\theta=\pi/4$.

\newpage
\section{Non-Abelian AB interference}~\label{A: SDF of AB}\vspace{-35pt}
\subsection*{1. Spin density interference}
In the main text, we have shown that the genuine non-Abelian nature of the AB system is characterized by the nontrivial interference of the two beams $\gamma_\mathrm{I}$, $\gamma_\mathrm{II}$.
In what follows, we will prove that the interfering spin density is always perpendicular to $\Delta\vec{s}=\vec{s}_\mathrm{I}-\vec{s}_\mathrm{II}$  on the screen. The final normalized spinor states of the two beams are given by $|s_i\rangle=\hat{U}_{\gamma_i}|s_0\rangle$ ($i=\mathrm{I},\mathrm{II}$), thus the orientation of the spin density satisfies
\begin{equation}
  \begin{split}
  \vec{s}(y)\propto&\,\left(\langle s_\mathrm{I}|e^{-i\theta_\mathrm{I}(y)}+\langle s_\mathrm{II}|e^{-i\theta_\mathrm{II}(y)}\right)\vec{\hat{\sigma}}\left(e^{i\theta_\mathrm{I}(y)}|s_\mathrm{I}\rangle+e^{i\theta_\mathrm{II}(y)}|s_\mathrm{II}\rangle\right)\\
  \propto&\, \vec{s}_\mathrm{I}+\vec{s}_\mathrm{II}+2\mathrm{Re}\left(e^{i\Delta\theta(y)}\langle s_\mathrm{I}|\vec{\hat{\sigma}}|s_\mathrm{II}\rangle\right).
 \end{split}
\end{equation}
It is obvious that $(\vec{s}_\mathrm{I}+\vec{s}_\mathrm{II})\cdot\Delta\vec{s}=0$, so we only need to prove $\langle s_\mathrm{I}|\vec{\hat{\sigma}}|s_\mathrm{II}\rangle\cdot\Delta\vec{s}=0$. Assuming $|s_\mathrm{I}\rangle=(a,b)^\intercal$, $|s_\mathrm{II}\rangle=(c,d)^\intercal$, we can obtain
\begin{align*}
  \langle s_\mathrm{I}|\vec{\hat{\sigma}}|s_\mathrm{II}\rangle=&\,(a^*d+b^*c)\vec{e}_1+i(-a^*d+b^*c)\vec{e}_2+(a^*c-b^*d)\vec{e}_3,\\
  \vec{s}_\mathrm{I}=&\, 2\,\mathrm{Re}(ab^*)\vec{e}_1+2i\,\mathrm{Im}(ab^*)\vec{e}_2+(|a|^2-|b|^2)\vec{e}_3,\\
  \vec{s}_\mathrm{II}=&\, 2\,\mathrm{Re}(cd^*)\vec{e}_1+2i\,\mathrm{Im}(cd^*)\vec{e}_2+(|c|^2-|d|^2)\vec{e}_3.
\end{align*}
Then, we have
\begin{equation*}
  \langle s_\mathrm{I}|\vec{\hat{\sigma}}|s_\mathrm{II}\rangle\cdot(\vec{s}_\mathrm{I}-\vec{s}_\mathrm{II})
  =(a^*b+b^*d)\left(|a|^2+|b|^2-|c|^2-|d|^2\right)=0.
\end{equation*}
Consequently, the interfering pseudo-spins $\vec{s}(y)$ lie on the great circle perpendicular to the axis $\Delta\vec{s}=\vec{s}_\mathrm{I}-\vec{s}_\mathrm{II}$ on the Bloch sphere. In addition,  the $\mathrm{SU}(2)$ phase factor $\exp\left[i\Phi_i\sigma_i\right]$ of winding around a vortex corresponds to a $SO(3)$ rotation of the pseudo-spin about the $\vec{e}_i$-axis clockwise through an angle $2\Phi_i$, i.e.
$  R_i=\exp\left[-2\Phi_i\vec{e}_i\times\mathbb{I}\right]$. As a result, the final spins of the two beams can be expressed as
\begin{subequations}
  \begin{align}
  \vec{s}_\mathrm{I}=&\,R_{\gamma_\mathrm{I}}\vec{s}_0=R_{\gamma_0}R_2 {R_1^{-1}}\,\vec{s}_0,\\
  \vec{s}_\mathrm{II}=&\,R_{\gamma_\mathrm{II}}\vec{s}_0=R_{\gamma_0}{R_1^{-1}} {R_2}\,\vec{s}_0.
  \end{align}
\end{subequations}
Therefore, the axis $\Delta\vec{s}$ perpendicular to the spin density on screen is determined by the commutator of the two rotation matrices:
\begin{equation}\label{spin axis}
  \begin{split}
  \Delta\vec{s}=&\,R_{\gamma_0}\left[R_2,{R_1^{-1}}\right]\vec{s}_0
  =\,4\sin\Phi_1\,\sin\Phi_2\,\cdot \Big(\vec{u}(\Phi_1,\Phi_2)\times\vec{s}_0-2\vec{u}(\Phi_1,\Phi_2)\cdot(\vec{s}_0\times\vec{e}_3)\vec{e}_3\Big),
\end{split}
\end{equation}
where $R_{\gamma_0}=\mathbb{I}$ in our case, and $\vec{u}(\Phi_1,\Phi_2)$ is a vector defined in the spin space:
 \begin{equation}
\vec{u}=\cos\Phi_1\sin\Phi_2\,\vec{e}_1+\sin\Phi_1\cos\Phi_2\,\vec{e}_2
+\cos\Phi_1\cos\Phi_2\,\vec{e}_3,
\end{equation}
whose norm $|\vec{u}|^2=1-\sin^2\Phi_1\,\sin^2\Phi_2\leq1$.  
If the initial spin satisfies $\vec{u}(\Phi_1,\Phi_2)\times\vec{s}_0-2\vec{u}(\Phi_1,\Phi_2)\cdot(\vec{s}_0\times\vec{e}_3)\vec{e}_3=0$, the spins of the two beams will finally evolve into the same direction $\vec{s}_\mathrm{I}=\vec{s}_\mathrm{II}$, and hence the orientation of the spin density will not fluctuate on the screen. 
Nevertheless, the nontrivial AB effect can still be detected from the phase shift $\delta\theta$ and the contracted amplitude $b<1$ in intensity interference pattern in this situation.

\begin{figure}
 \includegraphics[width=0.6\textwidth]{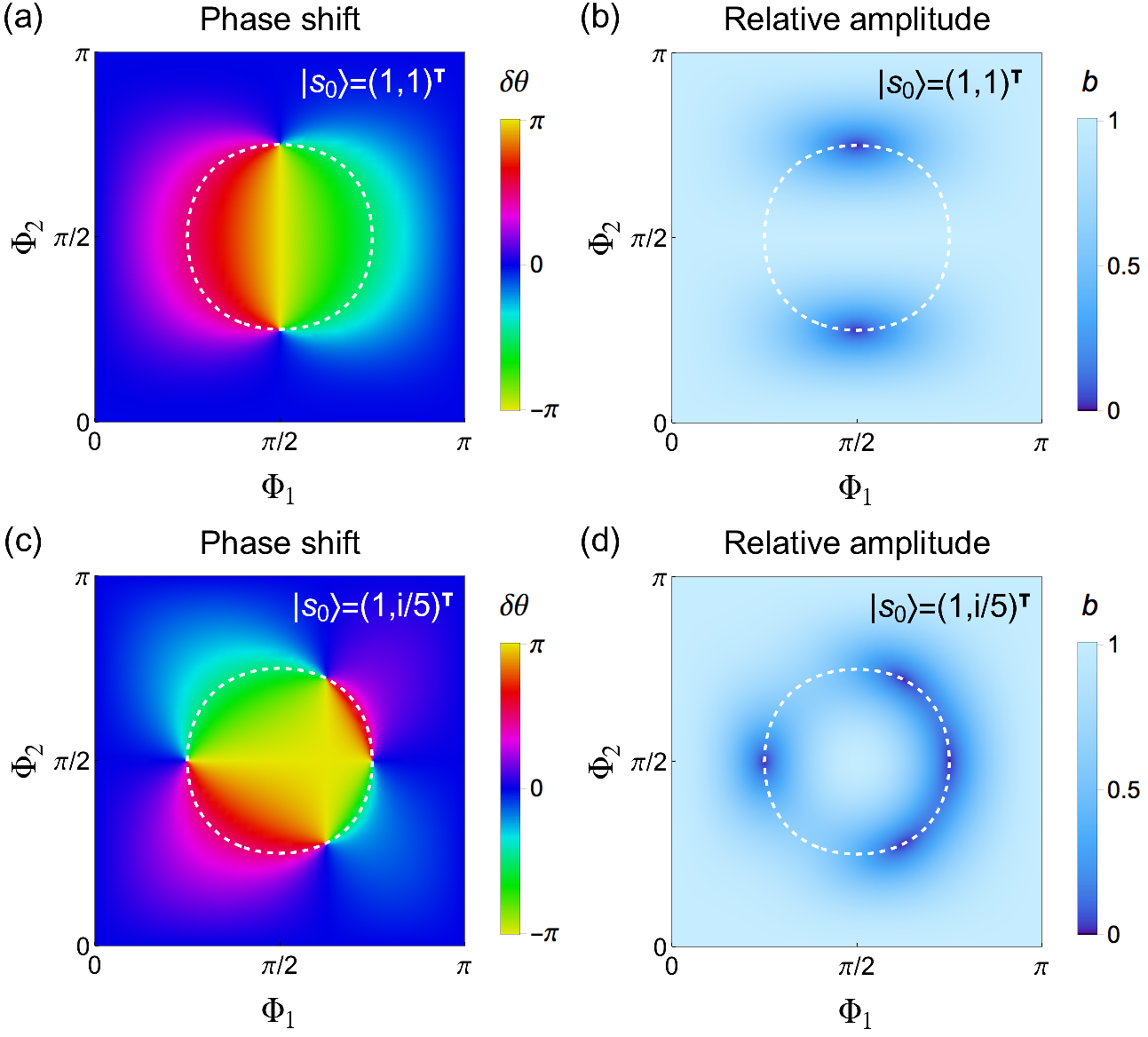}
 \caption{\label{fig5}  Phase shift $\delta\theta$ and  relative amplitude $b$ of the intensity interference fringes varying with the vortex fluxes $\Phi_1$, $\Phi_2$ for two incident spinors: (a,b) $|s_0\rangle=(1,1)^\intercal$, and (c,d) $|s_0\rangle=(1,i/5)^\intercal$. On the white dashed circle, the Wilson loop of the closed path $c_0=\gamma_\mathrm{I}\circ\gamma^{-1}_{\mathrm{II}}$ equals zero, $W(c_0)=0$. For a fixed $|s_0\rangle$, there exist singularities of the phase shift $\delta\theta$ lying on the white circle, and the relative amplitude is reduced to zero, i.e. $b=0$, at such singularities.}
\end{figure}

As for the intensity interference part $|\psi|^2$, the phase shift $\delta\theta$ and the relative amplitude $b$ of the interference pattern can be explicitly written as functions of the vortex fluxes $\Phi_1$, $\Phi_2$ and the initial spin $\vec{s}_0$: 
\begin{align}
&\delta\theta(\Phi_1,\Phi_2,\vec{s}_0)=\arctan\left[\frac{2\sin\Phi_1\sin\Phi_2}{1-2\sin^2\Phi_1 \sin^2\Phi_2}\vec{s}_0\cdot\vec{u}(\Phi_1,\Phi_2)\right],\\[3pt]
&\begin{aligned}
b(\Phi_1,\Phi_2,\vec{s}_0)=&\Big[\big(1-2\sin^2\Phi_1 \sin^2\Phi_2\big)^2+4\sin^2\Phi_1\sin^2\Phi_2 \big(\vec{s}_0\cdot\vec{u}(\Phi_1,\Phi_2)\big)^2\Big]^{1/2}.
\end{aligned}
\end{align}
They are, obviously, periodic with respect to $\Phi_1$ and $\Phi_2$.
In general, the relative amplitude $b\leq1$ with equality if and only if $\vec{s}_0$ is parallel to $\vec{u}$. Therefore, a contracted amplitude ($b<1$) is a signature of the genuine non-Abelian AB interference. 
At the same time, the Wilson loop of  $c_0$ is given by
\begin{equation}
  W(c_0)=\mathrm{Tr}\,\hat{\mathcal{U}}_{c_0}=2b\cos\delta\theta=2-4\sin^2\Phi_1 \sin^2\Phi_2.
\end{equation}
In particular, when $\sin^2\Phi_1\sin^2\Phi_2=1/2$, the phase shift converges to two discrete numbers: $\delta\theta=\pm\pi/2$, and the relative amplitude is reduced to $b(\Phi_1,\Phi_2)=
\sqrt{2}\left| \vec{s}_0\cdot\vec{u}(\Phi_1,\Phi_2)\right|$, and the Wilson loop is reduced to zero: $W(c_0)=b\cos\delta\theta=0$. Furthermore, if $\sin^2\Phi_1\sin^2\Phi_2=1/2$ and $\vec{s}_0\cdot\vec{u}=0$ are achieved simultaneously, the relative amplitude is reduced to zero $b=0$, as a result, the spatial fluctuation of the quasi-intensity $|\psi|^2$ vanishes. According to Eq.~(25) in the main text, $b=0$ implies that the finial spinors of the two beams are orthogonal: $\langle\psi_\mathrm{I}(y)|\psi_\mathrm{II}(y)\rangle=a(y)^2\langle s_\mathrm{I}| s_\mathrm{II}\rangle=0$, and hence their superposition leads to  no intensity fluctuation but a pure spin rotation around the axis $\Delta\vec{s}=2\vec{s}_\mathrm{I}$ on the screen.


Supplementary Figure~\ref{fig5} shows the non-Abelian-honolomy induced phase shift $\delta\theta$  and relative amplitude   $b$ as functions of $\Phi_1$ and $\Phi_2$ for two fixed incident spinors.  With the initial spinor $|s_0\rangle=(1,1)^\intercal$, we observe, in Supplementary Figure~\ref{fig5}(a), two  singularities  at $\{\Phi_1,\Phi_2\}=\{\pi/2,\pi/4\}$ and $\{\Phi_1,\Phi_2\}=\{\pi/2,3\pi/4\}$, where the phase shift $\delta\theta(\Phi_1,\Phi_2)$ is ill-defined. Meanwhile, the interference amplitude  $b$ drops to zero at the singularities, as shown in Supplementary Figure~\ref{fig5}(b). This result is supported by examining the spinors of the two beams, which evolve eventually into the pair of orthogonal states $|\mathord{\uparrow}\rangle=(1,0)^\intercal$, $|\mathord{\downarrow}\rangle=(0,1)^\intercal$ for both singularities.
Similarly, with the incident spinor $|{s}_0\rangle=(1,i/5)^\intercal$, there are four singularities of $\delta\theta(\Phi_1,\Phi_2)$ corresponding to the zero points of $b(\Phi_1,\Phi_2)$ as  shown in Supplementary Figure~\ref{fig5} (c,d). 
All of such singularities of $\delta\theta(\Phi_1,\Phi_2)$ lie on the curve of $\sin^2\Phi_1\sin^2\Phi_2=1/2$ for any initial spin $\vec{s}_0$.

\subsection*{2. Spin projected interference}

Apart from directly observing the total intensity $|\psi|^2(y)$ on the screen discussed in the main text,  the nontrivial interference effect can alternatively be detected via measuring the projected intensity onto a certain spin direction with a spin filter.
For a certain spin direction $\vec{n}$, the corresponding projection operator reads $|n\rangle\langle n|=\frac{1}{2}(\mathbb{I}+\vec{n}\cdot\vec{\hat{\sigma}})$. The projected intensity  on this  spin direction is given by
\begin{equation}\label{projection}
\begin{split}
\left|\langle n\,|\,{\psi}(y)\rangle\right|^2 &=\langle{\psi(y)}\left|n\rangle\langle n\right|{\psi(y)}\rangle=\left\langle{\psi}(y)\left|\frac{1}{2}\left(\mathbb{I}+\vec{n}\cdot\vec{\hat{\sigma}}\right)\right|{\psi}(y)\right\rangle  \\
&=\frac{1}{2}|{\psi}|^2(y)\Big(1+\vec{n}\cdot\vec{s}(y)\Big),
\end{split}
\end{equation}
where $|\psi\rangle=|\psi_\mathrm{I}\rangle+|\psi_\mathrm{II} \rangle$, and $\vec{s}= \langle{{\psi} |\vec{\hat{\sigma}}| {\psi}}\rangle/|\psi|^2$.  In general, a projected interference pattern can exhibit the nontrivial information of the intensity fringes $|\psi|^2(y)$ and the spin fluctuation $\vec{s}(y)$ simultaneously. If $\vec{n}$ is parallel to $\Delta\vec{s}$ (see Eq.~(\ref{spin axis}), we have $\vec{n}\cdot\vec{s}(y)\equiv0$, and hence the projected interference pattern would be in the same shape as the total intensity $|\psi|^2(y)$. 

 \begin{figure}[t]
 \includegraphics[width=1\textwidth]{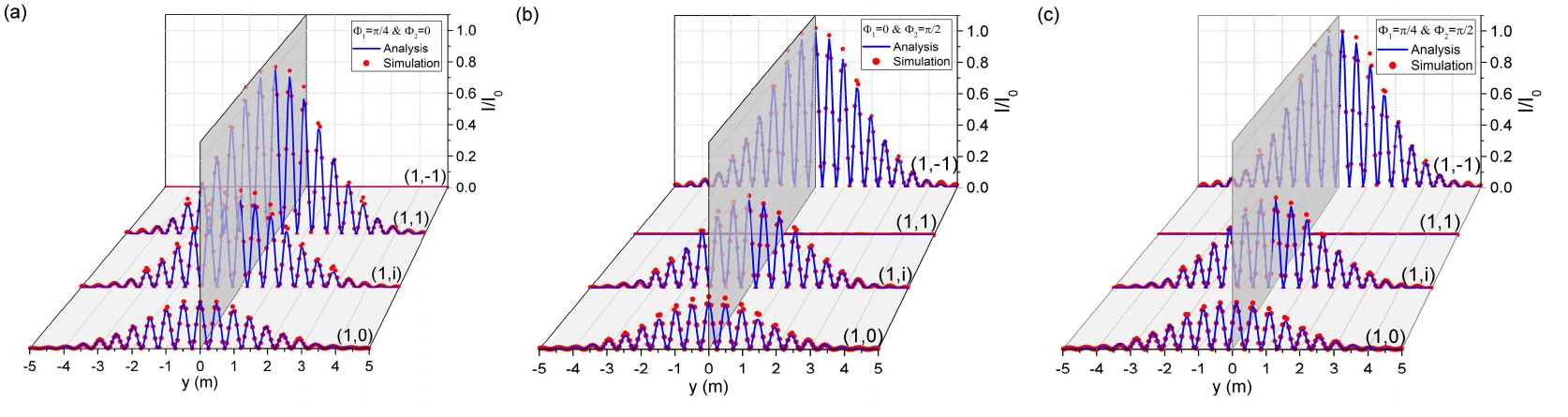}
 \caption{\label{fig4cd}   Spin projected interference patterns for the non-Abelian AB system shown in Fig.2 of the main text with the vortex fluxes (a) $\{\Phi_1,\Phi_2\}=\{\pi/4,0\}$, (b) $\{\Phi_1,\Phi_2\}=\{0,\pi/2\}$, and (c) $\{\Phi_1,\Phi_2\}=\{\pi/4,\pi/2\}$.
 In the three cases, the incident spinor  is given by $|s_0\rangle=(1,1)^\intercal$,  the spin-filter states $|n\rangle$ for the projection are selected as  $(1,0)^\intercal$, $(1,i)^\intercal$, $(1,1)^\intercal$ and $(1,-1)^\intercal$.
 The red dotes and the blue lines correspond to the full-wave simulation and the theoretical results using \Eq{projection} respectively. The grey reference planes located at $y$=0 marks the original position of the central peak without phase shift. In (a,b), the central peak is coincident with the grey plane, thus there is no phase shift $\delta\theta=0$, while in (c) the phase shift is $\delta\theta=\pi/2$.}
 \end{figure}

We carry out the numerical simulation of the interference in our designed non-Abelian AB system where the finial states of the two beams $\gamma_\mathrm{I},\ \gamma_\mathrm{II}$ are projected onto certain spin directions  on the detection plane. First, we study the situation of a single vortex by setting either $\Phi_1=0$ or $\Phi_2=0$.  
In this case, the system is reduced to being essentially Abelian. Because of the same AB phase factor for the two beams, the superposed spin density $|\psi|^2(y)\vec{s}$ are uniformly oriented on the screen, and there is no phase shift ($\delta\theta=0$) and amplitude contraction ($b=1$). Thus, the projected interference onto any direction $\vec{n}$ is consistent with the total intensity up to a scaling parameter $(1+\vec{n}\cdot\vec{s})/2$.
In Supplementary Figure~\ref{fig4cd}(a), the projected interference patterns onto 4 different directions are plotted, where only the $\hat{\sigma}_1$ vortex exists with the flux $\Phi_1=\pi/4$  and the incident spinor reads $|s_0\rangle=(1,1)^\intercal$.  As $|s_0\rangle$ is an eigenstate of $\hat{\sigma}_1$, the spin is conserved during propagation. Accordingly, the projection onto $(1,1)^\intercal$ are identical to the total intensity $|\psi|^2(y)$, while the projection onto $(1,-1)^\intercal$ vanishes as shown in Supplementary Figure~\ref{fig4cd}(a), and the projected intensities onto $(1,0)^\intercal$ and  $(1,i)^\intercal$ are exactly halved. 
Supplementary Figure~\ref{fig4cd}(b) shows the projected interference patterns for a single $\hat{\sigma}_2$ vortex carrying the flux $\Phi_2={\pi}/{2}$.  As  the incident spinor $|s_0\rangle=(1,1)^\intercal$ is not an eigenstate of $\hat{\sigma}_2$, the spins of the two beams rotate along the trajectories and finally convert to $|s_\mathrm{I}\rangle=|s_\mathrm{II}\rangle=(1,-1)^\intercal$, evidenced by the vanishing projection onto $(1,1)^\intercal$.


In Supplementary Figure~\ref{fig4cd}(c), we illustrate the projected interference patterns in a genuine non-Abelian AB system with two noncommutative vortices. And the specific fluxes $\{\Phi_1,\Phi_2\}=\{\pi/4,\pi/2\}$ indicate a zero Wilson loop $W(c_0)=0$. And since $\vec{u}=\sqrt{2}/2\vec{e}_1$, $\vec{s}_0=\vec{e}_1$ satisfy $\vec{u}\times\vec{s}_0-2\vec{u}\cdot(\vec{s}_0\times\vec{e}_3)\vec{e}_3=0$, the finial spins of the two beams fall into a uniform direction $\vec{s}_\mathrm{I}=\vec{s}_\mathrm{II}=-\vec{e}_1$, leading to the vanishing projection onto $(1,1)^\intercal$. Therefore, the projected intensity on any direction $\vec{n}$ is always proportional to the total intensity: $|\langle n|\psi(y)\rangle|^2=\frac{1}{2}|\psi(y)|^2(1-\vec{n}\cdot\vec{e}_1)$. Meanwhile, the phase shift and relative amplitude for all the projected interference patterns are fixed as $\delta\theta=\pi/2$, $b=\sqrt{2}\,|\vec{s}_0\cdot\vec{u}|=1$, which are confirmed by the numerical results in Supplementary Figure~\ref{fig4cd}(c).

\section{\hspace{-4pt}Designing non-Abelian AB system with gyroelectric materials}
Although the genuine non-Abelian AB system designed in the main text has a rather simple geometry, the background materials  is required to support both gyroelectric and gyromagnetic responses, meanwhile, their electric and magnetic gyrotation vectors should obey a rigorous relation. In this section, we offer an alternative design of the non-Abelian AB system using gyroelectric materials without gyromagnetic response which would be more easily realized in practice.

 \begin{figure}[h]
\includegraphics[width=0.33\columnwidth,clip]{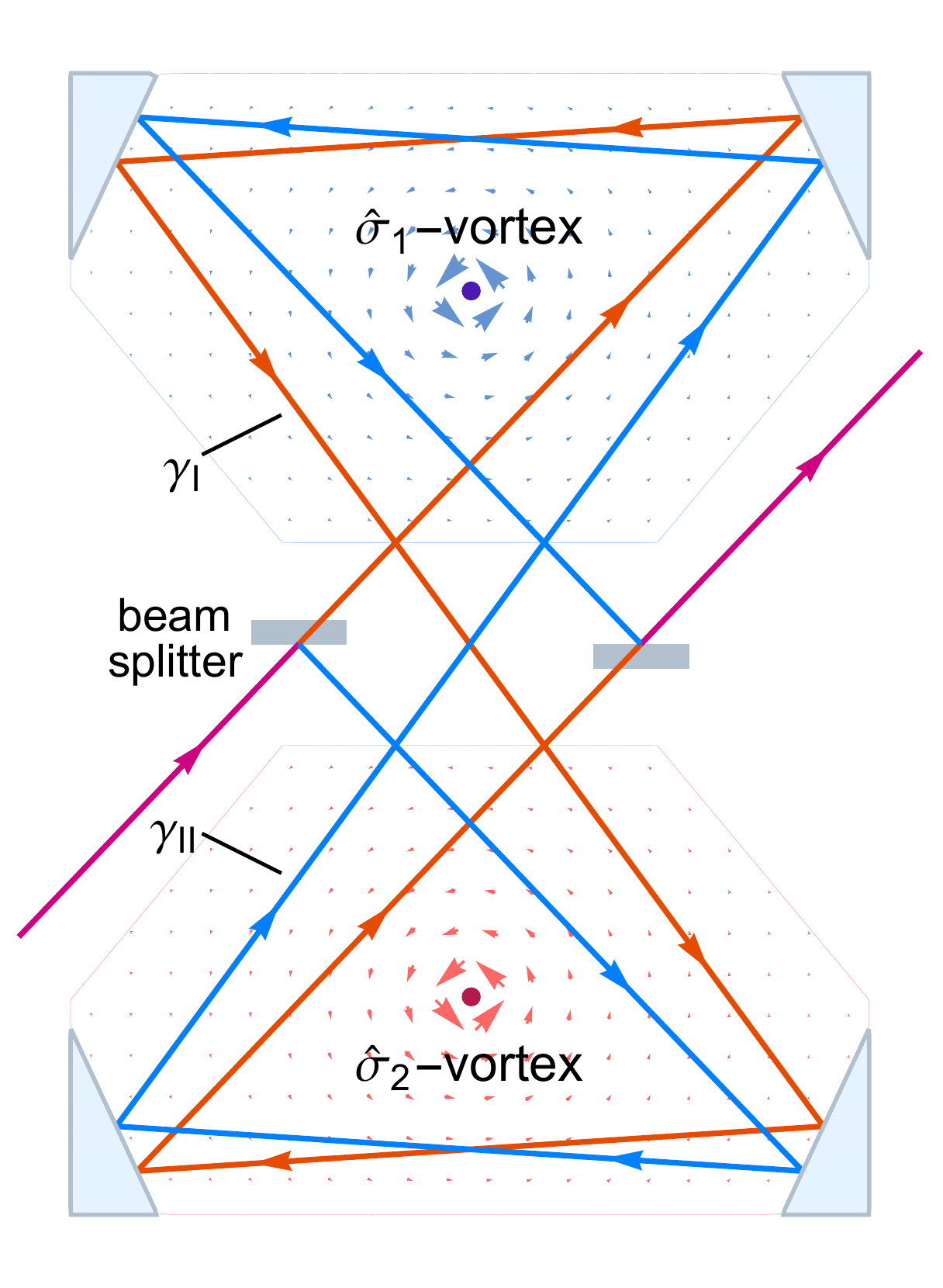}
\caption{Alternative design of the non-Abelian AB system with two interfering optical paths $\gamma_\mathrm{I}$, $\gamma_\mathrm{II}$, where the background light blue (red) arrows denote the $\hat{\sigma}_1$ ($\hat{\sigma}_2$) component $\mathcal{A}^1$ ($\mathcal{A}^2$) of the synthetic non-Abelian vector potential. The media used to imitate $\hat{\sigma}_1$-vortex and $\hat{\sigma}_1$-vortex  in the upper and lower half-spaces are, respectively, an inhomogeneous reciprocal anisotropic material and an inhomogeneous gyroelectric material.
\label{reducedABsystem}}
\end{figure}

As shown in Supplementary Figure~\ref{reducedABsystem}, the lower half-plane contains a synthetic non-Abelian vortices made up of gyroelectric materials:
\begin{equation}\label{gyroelectric}
  \tensor{\varepsilon}/\varepsilon_0=\left(\begin{array}{c|c} \varepsilon_T \tensor{I}_{2\times 2} & {i\,}\mathbf{g}_1\\\hline
-i\,\mathbf{g}_1 & \varepsilon_z
\end{array}\right),\qquad \tensor{\mu}/\mu_0=\left(\begin{array}{c|c} \mu_T \tensor{I}_{2\times 2} & 0\\\hline
0 & \mu_z
\end{array}\right),
\end{equation}
where $\varepsilon_T=\alpha\,\mu_T\ (\alpha>0)$ is supposed to be homogeneous, while $\mathbf{g}_1$, $\varepsilon_z$, and $\mu_z$ are real-valued functions of coordinates $x,\,y$. By rescaling the vacuum permittivity $\varepsilon'_0=\alpha\,\varepsilon_0$, we obtain the synthetic vector and scalar gauge potentials:
\begin{equation}
  \hat{\mathcal{A}}=k_0\frac{\mathbf{g}_1\times\mathbf{e}_z}{2\sqrt{\alpha}}\,\hat{\sigma}_2,\quad
  \hat{\mathcal{A}}_0=\frac{k_0}{2\sqrt{\alpha}}\mathbf{e}_z\cdot\left(\nabla\times\mathbf{g}_1\right) \hat{\sigma}_1\, + \, \frac{{k_0}^2}{2}\left(\varepsilon_z-\alpha\mu_z-\frac{|\mathbf{g}_1|^2}{\alpha}\right)\hat{\sigma}_3,\quad
V_0=\frac{k_0^{\,2}}{2}\left[\frac{|\mathbf{g}_1|^2}{2\alpha}-\varepsilon_z-\alpha\mu_z\right].
\end{equation}
Since $\mathbf{g}=\mathbf{g}_1\times\mathbf{e}_z$ is precisely the gyration vector of the gyroelectric materials, $\mathcal{A}_2$ is parallel to the gyration vector everywhere.
In order to eliminate the synthetic non-Abelian magnetic fields $\hat{\mathcal{B}}=\nabla\times\hat{\mathcal{A}}_2\hat{\sigma}_2=0$ in the whole medium (except for a small domain which is simplified as a singularity), we let ${\mathcal{A}}_2$ be an irrotational vortex with the center at $-\mathbf{r}_0=-r_0\mathbf{e}_y$:
\begin{equation}
  {\mathcal{A}}_2=\frac{\Phi_2}{2\pi}\nabla \arctan\left(\frac{x}{y+r_0}\right)=\frac{\Phi_2}{2\pi |\mathbf{r}+\mathbf{r}_0|^2}\left[-(y+r_0)\mathbf{e}_x+x\mathbf{e}_y\right].
\end{equation}
As such, the off-diagonal term of $\tensor{\varepsilon}$ is given by
\begin{equation}
  \mathbf{g}_1=\frac{2\sqrt{\alpha}}{k_0}\,\mathbf{e}_z\times\hat{\mathcal{A}}_2=-\frac{\Phi_2\sqrt{\alpha}}{\pi k_0|\mathbf{r}+\mathbf{r}_0|^2}\left[x\mathbf{e}_x+(y+r_0)\mathbf{e}_y\right].
\end{equation}
In order to extinguish non-Abelian and Abelian electric fields $\hat{\mathcal{E}}=\nabla\hat{\mathcal{A}}_0+i[\hat{\mathcal{A}}_0,\hat{\mathcal{A}}]$ and $\mathsf{E}=-\nabla V_0$, the scalar potentials should satisfy $\hat{\mathcal{A}_0}=0$ and $V_0=\mathrm{const.}$.
The benefit of adopting the irrotational vortex to form $\mathcal{A}_2$ is that  $\mathcal{A}_0^1\propto \mathbf{e}_z\cdot\left(\nabla\times\mathbf{g}_1\right)=0$ can be simultaneously satisfied. Therefore, the requirements to the scalar potentials can be met providing that the $z$-components of permittivity and permeability satisfy
\begin{subequations}
  \begin{align}
  \varepsilon_z&=\frac{1}{2}\left(\frac{|\mathbf{g}_1|^2}{\alpha}+\mathrm{const.}\right)=\frac{1}{2}\left(\frac{\Phi_2^{\,2}}{\pi^2k_0^{\,2}}|\mathbf{r}+\mathbf{r}_0|^{-2}+\mathrm{const.}\right),\\
    \mu_z&=\frac{1}{2\alpha}\left(-\frac{|\mathbf{g}_1|^2}{\alpha}+\mathrm{const.}\right)=\frac{1}{2\alpha}\left(-\frac{\Phi_2^{\,2}}{\pi^2k_0^{\,2}}|\mathbf{r}+\mathbf{r}_0|^{-2}+\mathrm{const.}\right).
\end{align}
\end{subequations}
To realize this kind of gyroelectric materials, a promising approach is to design suitable gyroelectric metamaterials with either passive magneto-optic composites~\cite{sadatgol2016enhanced} or active structures~\cite{wang2012gyrotropic} which have the advantages of strong gyrotropic response and turnable anisotropy. In particular, if the gyroelectricity is induced by magneto-optic effect, the gyration vector $\mathbf{g}$ should be parallel and proportional to the external magnetic field. For the present case, if a line current alone $z$-axis is placed at the $-\mathbf{r}_0$, the generated magnetic field  forms an irrotational vortex, and the induced gyration vector gives exactly the $\hat{\sigma}_2$-vortex.

In the upper half-plane, the synthetic $\hat{\sigma}_1$-vortex with the center at $\mathbf{r}_0=r_0\mathbf{e}_y$ can be constructed by reciprocal anisotropic materials
\begin{equation}
  \tensor{\varepsilon}/\varepsilon_0=\left(\begin{array}{c|c} \varepsilon_T \tensor{I}_{2\times 2} & \mathbf{g}'_1\\\hline
\mathbf{g}'_1 & \varepsilon_z
\end{array}\right),\qquad \tensor{\mu}/\mu_0=\left(\begin{array}{c|c} \mu_T \tensor{I}_{2\times 2} & 0\\\hline
0 & \mu_z
\end{array}\right),
\end{equation}
with $\varepsilon_T=\alpha\,\mu_T\ (\alpha>0)$ being homogeneous and $\mathbf{g}'_1$, $\varepsilon_z$, $\mu_z$ being real-valued functions of $x,\,y$.
Similarly to the lower half-space, we let $\hat{\mathcal{A}}=\mathcal{A}_1\hat{\sigma}_1$ form an irrotational vortex in the anisotropic material:
\begin{equation}
  {\mathcal{A}}_1=\frac{\Phi_1}{2\pi}\nabla \arctan\left(\frac{x}{y-r_0}\right)=\frac{\Phi_1}{2\pi |\mathbf{r}-\mathbf{r}_0|^2}\left[-(y-r_0)\mathbf{e}_x+x\mathbf{e}_y\right],
\end{equation}
and let $\hat{\mathcal{A}_0}=0$ and $V_0=\mathrm{const.}$. Then the material parameters are obtained as follows
\begin{subequations}
\begin{align}
  \mathbf{g}'_1 &=-\frac{\Phi_1\sqrt{\alpha}}{\pi k_0|\mathbf{r}-\mathbf{r}_0|^2}\left[x\mathbf{e}_x+(y-r_0)\mathbf{e}_y\right],\\
  \varepsilon_z&=\frac{1}{2}\left(\frac{\Phi_1^{\,2}}{\pi^2k_0^{\,2}}|\mathbf{r}-\mathbf{r}_0|^{-2}+\mathrm{const.}\right),\\
    \mu_z&=\frac{1}{2\alpha}\left(-\frac{\Phi_1^{\,2}}{\pi^2k_0^{\,2}}|\mathbf{r}-\mathbf{r}_0|^{-2}+\mathrm{const.}\right).
\end{align}
\end{subequations}

In comparison with the design in the main text, a disadvantage of the reduced non-Abelian AB system is that the sudden change of the synthetic gauge potentials will induce nonzero non-Abelian gauge flux at the boundaries of the materials. 
In order to prevent the boundary flux from affecting the total flux enclosed by the optical paths, the optical paths encircling the non-Abelian vortices should form closed loops before traversing the boundaries of the upper or lower pieces of materials. Thus we designed a new optical path diagram shown in Supplementary Figure~\ref{reducedABsystem}. The new AB system can give rise to the identical interfering results as the original  non-Abelian system in the main text, namely the optical beams passing through the two paths $\gamma_\mathrm{I}$ and $\gamma_\mathrm{II}$ can gain the non-Abelian phase factors $\hat{U}_{\gamma_\mathrm{I}}=\hat{U}_2^{-1}\hat{U}_1$ and $\hat{U}_{\gamma_\mathrm{II}}=\hat{U}_1\hat{U}_2^{-1}$ respectively, where $\hat{U}_i=\exp\left[i\Phi_i\hat{\sigma}_i\right]$ ($i=1,2$). Since $\hat{U}_1$ and $\hat{U}_2$ do not commute with each other, the opposite sequences of winding around the two vortices lead to different non-Abelian phase factors $\hat{U}_{\gamma_\mathrm{I}}\neq\hat{U}_{\gamma_\mathrm{II}}$.

\newpage
\def\bibsection{\section*{Supplementary References}}
\bibliographystyle{naturemag}
\bibliography{reference}